\documentclass[3p,12pt]{elsarticle}
\usepackage{lineno,hyperref}
\usepackage{amsmath,bm,geometry}
\usepackage{amsfonts}
\usepackage{color}
\usepackage{nomencl,etoolbox,ragged2e,siunitx}
\usepackage{mathtools}
\usepackage{multirow}
\usepackage{caption}
\usepackage{parskip}
\usepackage{calligra}
\usepackage{makeidx}
\usepackage{tikz}
\usepackage{subcaption}
\usepackage[utf8]{inputenc}
\usepackage[english]{babel}
\usepackage{amsthm}
\newtheorem{proposition}{Proposition}[section]
\usetikzlibrary{matrix}
\makeindex
\DeclareMathAlphabet{\mathcalligra}{T1}{calligra}{m}{n}
\newcommand*{\defeq}{\stackrel{\text{def}}{=}}
\begin{document}
	\title{Unified gas-kinetic wave-particle methods V: diatomic molecular flow }
	\author[ad1]{Xiaocong Xu}
	\ead{xxuay@connect.ust.hk}
	\author[ad1]{Yipei Chen}
	\ead{ychendh@connect.ust.hk}
	\author[ad1]{Chang Liu}
	\ead{cliuaa@connect.ust.hk}
	\author[ad3]{Zhihui Li}
	\ead{zhli0097@x263.net}
	\author[ad1,ad2]{Kun Xu\corref{cor1}}
	\ead{makxu@ust.hk}
	\address[ad1]{Department of Mathematics, Hong Kong University of Science and Technology, Hong Kong}
\address[ad3]{National Laboratory for Computational Fluid Dynamics, Beijing University of Aeronautics and Astronautics, Beijing 100191, China}
\address[ad2]{Shenzhen Research Institute, Hong Kong University of Science and Technology, Shenzhen, China}
	\cortext[cor1]{Corresponding author}
	
	\begin{abstract}
In this paper, the unified gas-kinetic wave-particle (UGKWP) method is further developed for diatomic gas with the energy exchange between  translational and rotational modes for flow study in all regimes.
The multiscale transport mechanism in UGKWP is coming from the direct modeling in a discretized space, where the cell's Knudsen number,  defined by the ratio of particle mean free path over the numerical cell size, determines the flow physics simulated by the wave particle formulation.
 The non-equilibrium distribution function in UGKWP is tracked by the discrete particle and analytical wave. The weights of distributed particle and wave in different regimes are controlled by the accumulating evolution solution of particle transport and collision within a time step,
 where distinguishable macroscopic flow variables of  particle and wave are updated inside each control volume.
 With the variation of local cell's Knudsen number, the UGKWP becomes a particle method in the highly rarefied flow regime and converges to the gas-kinetic scheme (GKS) for the Navier-Stokes solution in the continuum flow regime without particles.
Even targeting on the same solution as the discrete velocity method (DVM)-based unified gas-kinetic scheme (UGKS), the computational cost and memory requirement in UGKWP could be reduced by several orders of magnitude for the high speed and high temperature flow simulation, where the translational and rotational non-equilibrium becomes important in the transition and rarefied regime.
As a result,  3D hypersonic computations around a flying vehicle in all regimes can be conducted using a personal computer. The UGKWP method for diatomic gas will be validated in various cases from one dimensional shock structure to three dimensional flow over a sphere, and the numerical solutions will be compared with the reference DSMC results and experimental measurements.
	\end{abstract}
	
	\begin{keyword}
	Unified Gas-kinetic Wave-particle Method, Multiscale Transport, Diatomic Gas, Hypersonic Non-equilibrium Flow
	\end{keyword}
	\maketitle

	\section{Introduction}
Kinetic equations describe the gas flow by modeling the evolution of the gas distribution function, or called probability density function (PDF).  The Boltzmann equation is constructed on the particle mean free path and mean collision time scale. In order to simplify the collisional operator of the Boltzmann equation, many relaxation models, such as the Bhatnagar-Gross-Krook (BGK) model \cite{BGK1954}, the ellipsoidal statistical BGK (ES-BGK) model \cite{holway1966new}, and the Shakhov BGK (S-BGK) model \cite{shakhov1968generalization}, have been developed and used in theoretical analysis and engineering applications. Theoretically, these kinetic equations are valid in all flow regimes, but with a resolved dynamics on the kinetic scale.
The numerical methods for solving kinetic equations can be classified into the stochastic methods and the deterministic methods.

	For the stochastic methods, the discrete particles are employed to simulate the evolution of the PDF. This kind of Lagrangian-type scheme can easily keep the positivity and conservation properties with super stability.
The direct simulation Monte Carlo (DSMC) method \cite{bird1963approach} is one of the most representative stochastic methods. By modelling the particle transport and collision separately, DSMC achieves great success in the simulation of high speed and rarefied non-equilibrium gas flow. However, the DSMC method will suffer from the statistical noise in the low speed flow simulation.   Meanwhile, due to the splitting treatment of particle transport and collision, the cell size and time step are restricted to be less than the particle mean free path and collision time.
Under this circumstance, the computational cost for the near continuum flow will increase rapidly.
	For the deterministic approaches, the most popular methods are the so-called  discrete velocity method (DVM) for the Boltzmann and kinetic equations \cite{chu1965,JCHuang1995,Mieussens2000,tcheremissine2005direct,Kolobov2007,LiZhiHui2009,ugks2010,wu2015fast,aristov2012direct}. When the Knudsen number is small, the collision operator in the kinetic equation becomes stiff, which will strongly restrict the time step. To improve the computational efficiency, the asymptotic preserving (AP) schemes  have been proposed and developed \cite{jin1999efficient}. Based on the DVM framework, many kinetic solvers have been developed for diatomic gas as well \cite{dubroca2001conservative,chen2008computation}.
	
	Following the direct modelling methodology \cite{xu-book}, an effective multiscale unified gas-kinetic scheme (UGKS) has been
	proposed for both monatomic and diatomic gases for the flow study in all regimes  \cite{ugks2010,huang2012,liu2016,wang2017unified,liu2014unified}.
Under the UGKS framework, the numerical flux is constructed from the integral solution of the kinetic model, where the effect of particle transport and collision is accumulated in a time step to identify the flow regime.
The UGKS has an asymptotic limit to the Navier-Stokes (NS) equations in the continuum flow regime
without kinetic scale restriction on time step and cell size, which has the so-called unified preserving (UP) property \cite{guo}.
Moreover, the implicit and multi-grid techniques have also been incorporated into the UGKS \cite{zhu2016implicit,zhu2017implicit,zhu2018implicit} to improve the computational efficiency. Recently, the unified gas-kinetic wave-particle (UGKWP) method  \cite{liu2020unified, zhu2019unified} based  on the BGK model has been developed for monatomic gas. The essential idea of UGKWP is to use stochastic particles to replace the discretization of
particle velocity space. In UGKWP, the gas particles are divided into hydro-particle, collisional particle, and collisionless particle.
	The hydro-particle is described by the analytic PDF, while the collisional and collisionless particles are described by the simulation particles.
In UGKWP, the macroscopic flow variables will be updated under the finite volume framework, where both analytical PDF and simulating particles
will contribute to the cell interface flux.
	One of the distinguishable features of UGKWP is that the dynamic evolution of hydro-particle can be described analytically and the computational cost of the flux from the hydro-particles is comparable to the hydrodynamic solver. The proportion of three kinds of particles varies dynamically in different flow regimes. Physically, the collisionless particles are mainly used for the description of non-equilibrium transport and the hydro-particles for the equilibrium one. The dynamic evolution among three kinds of particles is coupled with the variation of flow regime.
In the continuum flow regime, the number of collisional and collisionless particles will be greatly reduced and the UGKWP method automatically converges to the gas-kinetic scheme (GKS) for the Navier-Stokes solutions \cite{xu2001}, which has the similar efficiency as a conventional NS solver. For the simulation of hypersonic flow, the UGKWP method will be much more efficient than the original DVM-based UGKS due to the use of simulation particles instead of discretizing the particle velocity space. In summary, the computational cost of the UGKWP method is similar to the particle methods in rarefied regime and becomes the hydrodynamic flow solver in continuum regime.
The UGKWP is also extended to other multiscale transport, such as radiation and plasma \cite{li2020unified,liu2020plasma}.

This paper is about to construct the UGKWP method for diatomic gas. The Rykov kinetic model will be used in the construction of the evolution solution of the gas distribution function, which controls the distribution of particle and wave and the rate of energy exchange between translational and rotational degrees of freedom. A simple and efficient way to set up the correct transport coefficients is presented in this paper.
A weighted method is applied to sample the particles from a modified distribution function of the Rykov model. The overall UGKWP method for diatomic gas is very efficient and has excellent performance for high speed flow simulation with the translational and rotational non-equilibrium.
	
	The rest of the paper is organised as follows.
	In Section \ref{method}, the unified gas-kinetic particle method (UGKP) for diatomic gas will be introduced first. Then, based on the analytical formulation of the hydro-particle, the unified gas-kinetic wave-particle method will be presented, which is an improved version of UGKP.
	The asymptotic preserving property of the UGKWP method for diatomic gas in the continuum regime will be introduced in section \ref{discussion}.
	 Section \ref{numericaltest} includes various numerical tests to validate the new scheme.
	Section \ref{conclusion} is the conclusion.
	
	\section{Unified gas-kinetic wave-particle method for diatomic gas}\label{method}
	\subsection{The Rykov kinetic model for diatomic gas}
	Diatomic gas is associated with translational, rotational, and vibrational modes. For relatively low temperature flow, there are two rotational degrees of freedom while the vibrational degrees of freedom get basically frozen. In this paper, the diatomic gas with translational and rotational degrees of freedom will be considered. For nitrogen gas, the molecule has three translational degrees of freedom and two rotational ones.
The state of the gas can be described by the particle velocity distribution function $f(\vec{x},t,\vec{v}, \vec{\xi})$, where $\vec{x}$ is the spatial coordinate, $t$ is the time, $\vec{v}$ is the molecular translational velocity, and $\vec{\xi}$ is the rotational variable. The relations between distribution function and macro-variables are defined as
	\begin{equation*}
		\vec{W} = \int \vec{\psi} f d\Xi,
	\end{equation*}
	where $\vec{\psi}=\left(1,\vec{v},\frac12(\vec{v}^2+\vec{\xi}^2),\frac12\vec{\xi}^2 \right)$ is the vector for the moments of distribution function and  $\vec{W}=(\rho,\rho\vec{U},\rho E, \rho E_{rot})$ is the macroscopic variables with $\rho E_{rot}$ as the rotational energy density. The stress tensor $\mathbf{P}$ and the heat fluxes $\vec{q}_t$ and $\vec{q}_r$ produced by the transfer of translational and rotational energies, can be calculated by $f$ as,
	\begin{align*}
		\mathbf{P} &= \int \vec{c} \vec{c} f d\Xi , \\
		\vec{q}_t &= \frac{1}{2} \int \vec{c} \vec{c}^2 f d \Xi , \\
		\vec{q}_r &= \frac{1}{2} \int \vec{c} \vec{\xi}^2 f d \Xi ,
	\end{align*}
	where $\vec{c} = \vec{v} - \vec{U}$ is the peculiar velocity.
The total heat flux $\vec{q}$ is the sum of  $\vec{q}_t$ and $\vec{q}_r$, such as
	\begin{equation*}
		\vec{q} = \frac{1}{2} \int \vec{c} \left( \vec{c}^2 + \vec{\xi}^2 \right) f d \Xi = \vec{q}_t + \vec{q}_r .
	\end{equation*}
	The evolution of the diatomic gas distribution function $f$ is governed by the Rykov kinetic model equation,
	\begin{equation}\label{Rykov}
		\frac{\partial f}{\partial t}+ \vec{v}\cdot\nabla_{\vec{x}} f = \frac{\tilde{M}_t-f}{\tau} +\frac{\tilde{M}_{eq}-\tilde{M}_t}{Z_{rot}\tau},
	\end{equation}
	where the collision operator on the right-hand side describes the elastic and inelastic collisions. The elastic collision conserves the translational energy, while the inelastic collision exchanges the translational and rotational energy. $\tilde{M}_t$ is the modified equilibrium distribution function for the elastic collision while $\tilde{M}_{eq}$ is the modified equilibrium state for inelastic collision, which are expressed as,
	\begin{align*}
		\begin{aligned}
			\tilde{M}_t &= M_t + M_t^{+} , \\
			M_t &= \rho \left( \frac{\lambda_t}{\pi} \right)^{\frac{3}{2}}e^{-\lambda_t\vec{c}^2} \frac{\lambda_r}{\pi} e^{-\lambda_r \vec{\xi}^2} , \\
			M_t^{+} &= M_t \left(\frac{4\lambda_t^2\vec{q}_t \cdot \vec{c}}{15\rho} \left( 2\lambda_t \vec{c}^2 - 5\right) + \frac{4(1 - \sigma)\lambda_t \lambda_r \vec{q}_r \cdot \vec{c}}{\rho} \left( \lambda_r \vec{\xi}^2 - 1\right)    \right),
		\end{aligned}
	\end{align*}
	and
	\begin{align*}
		\begin{aligned}
			\tilde{M}_{eq} &= M_{eq} + M_{eq}^{+} ,\\
			M_{eq} &= \rho \left( \frac{\lambda_{eq}}{\pi} \right)^{\frac{3}{2}}e^{-\lambda_{eq}\vec{c}^2} \frac{\lambda_{eq}}{\pi} 	e^{-\lambda_{eq} \vec{\xi}^2} ,\\
			M_{eq}^{+} &= M_{eq} \left(\omega_0\frac{4\lambda_{eq}^2\vec{q}_t \cdot \vec{c}}{15\rho} \left( 2\lambda_{eq} \vec{c}^2 - 5\right) + \omega_1\frac{4(1 - \sigma)\lambda_{eq}^2 \vec{q}_r \cdot \vec{c}}{\rho} \left( \lambda_{eq} \vec{\xi}^2 - 1\right)    \right) ,
		\end{aligned}
	\end{align*}
	where $\lambda_{t,r,eq} = 1/ (2RT_{t,r,eq})$, $\tau = \mu(T_t) / (\rho RT_t)$ is the translational collision time with $\mu(T_t)$ as the dynamic viscosity coefficient and 
the subscript $t,r,qe$ in $T_{t,r,eq}$ represent translational, rotational and equilibrium temperature, respectively.
 $Z_{rot}$ is the rotational relaxation collision number which is related to the ratio of elastic collision frequency to inelastic collision frequency. The parameter $\sigma$ depends on the molecular potential, and $\omega_0$ and $\omega_1$ are set to have proper relaxation of heat flux. In this paper, these coefficients adopt the values $\sigma = 1 / 1.55$, $\omega_0 = 0.2354$, $\omega_1 = 0.3049$ for nitrogen\cite{liu2014unified}.
	
	\subsection{Unified gas-kinetic particle method}
	\subsubsection{General framework of the unified gas-kinetic particle method}\label{section210}
	Re-write the Rykov kinetic model equation in a more convenient form,
	\begin{equation}\label{simplefiedRykov}
		\frac{\partial f}{\partial t}+ \vec{v}\cdot\nabla_{\vec{x}} f = \frac{M^{*}-f}{\tau},
	\end{equation}
	where $M^*$ is defined as,
	\begin{equation}\label{newdist}
		M^{*} = \tilde{M}_t + \frac{\tilde{M}_{eq} - \tilde{M}_t}{Z_{rot}}
	\end{equation}
	Now the Rykov kinetic model has the BGK-type form with a different equilibrium distribution function. With a local constant collision time $\tau$, the integral solution of Eq.\eqref{simplefiedRykov} can be written as,
	\begin{equation}\label{integral-solution}
		f(\vec{x},t,\vec{v},\vec{\xi})=\frac{1}{\tau}\int_{0}^t e^{-(t-t')/\tau} M^{*}(\vec{x}',t',\vec{v},\vec{\xi}) dt'+e^{-t/\tau}f_0(\vec{x}-\vec{v}t),
	\end{equation}
	where $f_0$ is the initial distribution function at $t=0$ and $M^*$ is denfined as the convex combination of two modified equilibrium distribution function in Eq.\eqref{newdist}.
	The equilibrium distribution is integrated along the characteristics $\vec{x}'=\vec{x}+\vec{v}(t'-t)$.
	
	The UGKP is constructed on a discretized physical space $\sum_i \Omega_i\subset\mathcal{R}^3$ and discretized time $t^n\in\mathcal{R}^+$.
	The cell averaged conservative variables $\vec{W}_i=(\rho_i,\rho_i\vec{U}_i,\rho_iE_i,\rho_i{E_{rot,}}_i)$ on a physical cell $\Omega_i$ are defined as
	\begin{equation*}
		\vec{W}_i=\frac{1}{|\Omega_i|}\int_{\Omega_i}\vec{W}(\vec{x}) d\vec{x}.
	\end{equation*}
	The cell averaged macroscopic variables $\vec{W}_i$ are evolved by the macroscopic governing equations which can be obtained by taking moments of Eq.\eqref{simplefiedRykov}
	\begin{equation}\label{update-w}
		\vec{W}_i^{n+1}=\vec{W}_i^{n}-\frac{\Delta t}{|\Omega_i|} \sum_{l_s\in\partial\Omega_i} |l_s|\vec{F}_{s} + \vec{S}_i,
	\end{equation}
	where $l_s \in\partial\Omega_i$ is the cell interface with center $\vec{x}_s$ and outer unit normal vector $\vec{n}_s$. The flux function for the macroscopic variables at the cell interfaces are constructed by Eq. (\ref{integral-solution}),
	\begin{equation}\label{flux-w-ugks}
		\begin{aligned}
			\vec{F}_{s} &= \frac{1}{\Delta t}\int_{0}^{\Delta t}\int f(\vec{x}_s,t,\vec{v},\vec{\xi})\vec{v}\cdot \vec{n}_s \vec{\psi} d\Xi dt \\
			&=\frac{1}{\Delta t}\int_{0}^{\Delta t}\int\bigg[\frac1\tau\int_0^{t} e^{(t'-t)/\tau}M^*(\vec{x}'_s,t',\vec{v},\vec{\xi})dt'+
			e^{-t/\tau}f_0(\vec{x}_s-\vec{v}t)\bigg] \vec{v}\cdot \vec{n}_s \vec{\psi} d\Xi dt,
		\end{aligned}
	\end{equation}
	with the characteristics $\vec{x}'_s=\vec{x}_s+\vec{v}(t'-t)$.
	The equilibrium flux terms related to the Maxwellian distribution are denoted as $\vec{F}_{eq,s}$,
	\begin{equation}\label{Fg}
		\vec{F}_{eq,s}\defeq\frac{1}{\Delta t}\int_{0}^{\Delta t}\int\frac1\tau\int_0^{t} e^{(t'-t)/\tau}M^*(\vec{x}'_s,t',\vec{v}, \vec{\xi})dt' \vec{v}\cdot\vec{n}_s\vec{\psi} d\Xi dt,
	\end{equation}
	and the flux terms related to the initial distribution are $\vec{F}_{fr,s}$,
	\begin{equation}\label{Ff}
		\vec{F}_{fr,s}\defeq\frac{1}{\Delta t}\int_{0}^{\Delta t}\int e^{-t/\tau} f_0(\vec{x}_s-\vec{v}t) \vec{v}\cdot \vec{n}_s\vec{\psi} d\Xi dt.
	\end{equation}
	The source term $\vec{S}$ is
	\begin{equation*}
		\vec{S} = \int_{0}^{\Delta t } \int \frac{M^* - f}{\tau} \vec{\psi} d\Xi dt =  \int_{t^n}^{t^{n+1}} \vec{s} dt,
	\end{equation*}
	where $\vec{s}$ can be expressed as
	\begin{align*}
		\vec{s}=\left(0,0,0,0,\frac{\rho E_{rot}^{eq} - \rho E_{rot}}{Z_{rot}\tau}\right )^T.
	\end{align*}
	The equilibrium rotational energy $\rho E_{rot}^{eq}$ is determined under the assumption $T_r = T_t = T_{eq}$ such that
	\begin{equation}\label{equilibriumEr}
		\rho E_{rot}^{eq}  = \frac{\rho}{2\lambda_{eq}} \quad \text{and} \quad \lambda_{eq} = \frac{K_r + 3}{4} \frac{\rho}{\rho E - \frac12\rho(U^2 + V^2 + W^2)}.
	\end{equation}
	Here $K_r$ is the rotational degrees of freedom.
	
	In the UGKP method, the equilibrium flux term $\vec{F}_{eq,s}$ can be calculated analytically and the free streaming flux term  $\vec{F}_{fr,s}$ is calculated by the simulating particles. Specifically, the updates of macroscopic variables will become,
	\begin{equation*}
		\vec{W}^{n+1}_i=\vec{W}^n_i - \frac{\Delta t}{|\Omega_i|} \sum_{l_s\in\partial \Omega_i}|l_s| \vec{F}_{eq,s} +  \frac{\Delta t}{|\Omega_i|}\vec{W}_{fr,i} + \vec{S}_i.
	\end{equation*}
	where $\vec{W}_{fr,i}$	is the net free streaming flow of cell $i$ calculated by counting the particles passing through the cell interface during a time step. The detailed calculation method for $\vec{F}_{eq,s}$, $\vec{W}_{fr,i}$ and the update with source term $\vec{S}_i$ will be given in section \ref{section222}, section \ref{section223} and section \ref{update}, respectively.
	
	The particle dynamics in the UGKP method is constructed based on the Rykov kinetic model equation.
	The main idea of the UGKP method is to track particle trajectory until the collision happens.
	Once the particle collides with other particles, it will be merged into the macroscopic flow quantities,
	and get re-sampled from the updated macroscopic flow variables at the beginning of the next time step.
	
	\subsubsection{The construction of equilibrium flux}\label{section222}
	In this subsection, the construction of the equilibrium flux $\vec{F}_{eq,s}$ will be presented.
Recall that
	\begin{equation}\label{recallM}
		\begin{aligned}
			M^{*} &= \tilde{M}_t + \frac{\tilde{M}_{eq} - \tilde{M}_t}{Z_{rot}} \\
			&= M_t + \frac{{M}_{eq} - {M}_t}{Z_{rot}} + M_q,
		\end{aligned}
	\end{equation}
and
	\begin{equation*}
		M_q =  \frac{Z_{rot} - 1}{Z_{rot}}M_t^{+}  + \frac{1}{Z_{rot}}M_{eq}^+
	\end{equation*}
	as  the correction term for the heat flux.
	The leading order approximation \cite{xu2008multiple} gives,
	\begin{equation}\label{TdifferenceOrder}
		T_t - T_r = -\frac{2}{3} Z_{rot} \tau T_{eq} \nabla_{\vec{x}} \cdot \vec{U},
	\end{equation}
	from which $|T_{eq} - T_t| \sim O(\tau)$ can be estimated.
	The linearized Maxwell distribution function $M_t$ around the equilibrium temperature $T_{eq}$ is
	\begin{equation}
		M_t = M_{eq} + \frac{T_t - T_{eq}}{T_{eq}}M_{eq}\left[ \left(\frac{\vec{c}^2}{2RT_{eq}} - \frac{3}{2} \right) - \frac{3}{K_r} \left( \frac{\vec{\xi}^2}{2RT_{eq}} - \frac{K_r}{2} \right)  \right] + O(|T_{eq} - T_t|^2).
	\end{equation}
 	Therefore, the second term $({M}_{eq} - {M}_t)/Z_{rot}$ in Eq. (\ref{recallM}) is of order $\tau$. Since $M_q$ is also a high order term, only the first term of $M^*$, i.e., $M_t$, is expanded in the calculation of the equilibrium flux. The Maxwellian distribution $M_t$ is expanded around $\vec{x}_0$ as
	\begin{equation}\label{g-expansion1}
		\begin{aligned}
			M_t(\vec{x},t,\vec{v},\vec{\xi})=&M_t(\vec{x}_s,t^n,\vec{v},\vec{\xi})
			+(1-H[\bar{x}])\frac{\partial^l}{\partial x} M_t(\vec{x}_s,t^n,\vec{v},\vec{\xi})\bar{x}
			+H[\bar{x}]\frac{\partial^r}{\partial x}  M_t(\vec{x}_s,t^n,\vec{v},\vec{\xi})\bar{x}\\
			&+\frac{\partial}{\partial y}M_t(\vec{x}_s,t^n,\vec{v},\vec{\xi})\bar{y}
			+\frac{\partial}{\partial z}M_t(\vec{x}_s,t^n,\vec{v},\vec{\xi})\bar{z}
			+\frac{\partial}{\partial t}M_t(\vec{x}_s,t^n,\vec{v},\vec{\xi})(t-t^n) \\
			=&M_t(\vec{x}_s,t^n,\vec{v},\vec{\xi})\left[1+(1-H[\bar{x}])a^l\bar{x}+H[\bar{x}]a^r\bar{x}+b\bar{y}+c\bar{z}+A(t-t^n)\right],
		\end{aligned}
	\end{equation}
	where  $\bar{x}=x - x_s$, $\bar{y}=y - y_s$, and $\bar{z}=z - z_s$.
	The derivative functions of $M_t$, denoted as $a^l$, $a^r$, $b$, $c$, and $A$ have the following form
	\begin{align*}
		a^l&=a^l_1+a^l_2u+a^l_3v+a^l_4w+\frac12a^l_5\vec{v}^2 +\frac12 a^l_6\vec{\xi}^2 ,\\
		a^r&=a^r_1+a^r_2u+a^r_3v+a^r_4w+\frac12a^r_5\vec{v}^2 + \frac12a^r_6\vec{\xi}^2,\\
		&\qquad\qquad\qquad...\\
		A&=A_1+A_2u+A_3v+A_4w+\frac12A_5\vec{v}^2 + \frac12 A_6\vec{\xi}^2.
	\end{align*}
	The Heaviside function $H[x]$ is defined by
	\begin{equation}\nonumber
		H[x]=\left\{
		\begin{aligned}
		1 \quad x>0,\\
		0 \quad x\le0.
		\end{aligned}\right.
	\end{equation}
	The Maxwellian at $\vec{x}_0$ and its derivative functions can be obtained from the reconstructed macroscopic variables.
	In this paper, the van Leer limiter is used for reconstruction,
	\begin{equation}
		s=(\text{sign}(s_l)+\text{sign}(s_r))\frac{|s_l||s_r|}{|s_l|+|s_r|},
	\end{equation}
	where $s$, $s_l$, and $s_r$ are the slopes of macroscopic variables.
	The Maxwellian distribution $M_t$ at cell interface can be obtained from the macroscopic flow variables, which are evaluated by
	\begin{equation}
		\vec{W}_{s}=\int \vec{\psi} \left( M_t ^l H[\bar{u}]+M_t ^r(1-H[\bar{u}])\right)d\Xi,
	\end{equation}
	where $\bar{u}=\vec{u}\cdot\vec{n}_s$.
	
	The derivative functions $a^l,a^r,b,c,A$ are calculated from the spatial and time derivatives of $M_t$. Taking $a$ as an example,
	\begin{align*}
		a&=\frac{1}{M_t}\bigg(\frac{\partial M_t}{\partial x}\bigg),
	\end{align*}
	and
	\begin{align*}
		&a_6=4\frac{\lambda_r^2}{K_r}\left(2\frac{\partial \rho E_{rot}}{\partial x} - \frac12 \frac{K_r}{\lambda_r}\frac{\partial \rho}{\partial x} \right),\\
		&a_5 = \frac{4\lambda_t^2}{3}\left( B-2{U}R_1-2VR_2 -2WR_3\right),\\
		&{a}_{4} = 2\lambda_t R_3 - a_5 W,\\
		&{a}_{3} = 2\lambda_t R_2 - a_5 V,\\
		&{a}_{2} = 2\lambda_t R_1 - a_5 U,\\
		&a_1= \frac{\partial \rho}{\partial x} - {a}_2{U}- {a}_3{V}- {a}_4{W} - \frac12 a_5 (\vec{U}^2 + \frac{3}{2\lambda_t}) - \frac12 a_6\frac{K_r}{2\lambda_r},
	\end{align*}
	with the defined variables
	\begin{align*}
	\begin{aligned}
		B &= 2 \frac{\partial (\rho E - \rho E_{rot})}{\partial x} - (\vec{U}^2 + \frac{3}{2\lambda_t})\frac{\partial \rho}{\partial x} ,\\
		R_1 &= \frac{\partial \rho U}{\partial x}  - U \frac{\partial \rho}{\partial x} ,\\
		R_2 &= \frac{\partial \rho V}{\partial x} - V \frac{\partial \rho}{\partial x} ,\\
		R_3 &= \frac{\partial \rho W}{\partial x} -  W \frac{\partial \rho}{\partial x} ,\\
		\end{aligned}
	\end{align*}
	where the derivatives of macroscopic quantities are evaluated at $(\vec{x}_s,t^n)$.
The time derivatives of macroscopic variables are determined by the conservative requirements on the first order Chapman-Enskog expansion \cite{chapman1990mathematical},
	\begin{equation*}
		\left(\frac{\partial\vec{W}_s}{\partial t}\right)=-\int \left(a^l\bar{u}H[\bar{u}]+a^r\bar{u}(1-H[\bar{u}])+b\bar{v}+c\bar{w}\right)M_t \vec{\psi} d\Xi.
	\end{equation*}
	Once the Maxwellian distribution at cell interface and its derivative functions are determined, the equilibrium flux function Eq.\eqref{Fg} can be obtained using the expansion Eq.\eqref{g-expansion1} for the interface distribution function, which gives
	\begin{equation}\label{Fg1}
		\begin{aligned}
			\vec{F}_{eq,s}=\int \vec{v}\cdot \vec{n}_s \vec{\psi}\bigg\{ &C_1M^*(\vec{x}_s,t^n,\vec{v},\vec{\xi})
			+C_2\left[a^lH[\bar{u}]+a^r(1-H[\bar{u}])\right]\bar{u}  M_t(\vec{x}_s,t^n,\vec{v},\vec{\xi})\\
			+&C_2(b\bar{v}+c\bar{w})M_t(\vec{x}_s,t^n,\vec{v},\vec{\xi})
			+C_3AM_t(\vec{x}_s,t^n,\vec{v},\vec{\xi})\bigg\} d\Xi,
		\end{aligned}
	\end{equation}
	where the time integration related coefficients are
	\begin{equation*}
		\begin{aligned}
			C_1 &= 1 - \frac{\tau}{\Delta t} \left( 1 - e^{-\Delta t / \tau} \right) , \\
			C_2 &= -\tau + \frac{2\tau^2}{\Delta t} - e^{-\Delta t / \tau} \left( \frac{2\tau^2}{\Delta t} + \tau\right) ,\\
			C_3 &=  \frac12 \Delta t - \tau + \frac{\tau^2}{\Delta t} \left( 1 - e^{-\Delta t / \tau} \right) .
		\end{aligned}
	\end{equation*}
	Theoretically, the modified term $M_q$ in $M^*$ only contributes to the heat conduction coefficient in the energy flux. By following the treatment in \cite{may2007improved}, the calculation of equilibrium flux can be simplified. We can ignore the modified terms and correct the heat flux by modifying the derivatives of temperature.
	Based on the Chapman-Enskog expansion and the linearisation around translation temperature $T_t$ \cite{rykov1978macroscopic},  the heat fluxes $\vec{q}_t$ and $\vec{q}_r$ become
	\begin{equation}\label{heatflux}
		\begin{aligned}
			\vec{q}_t &= -\frac{15R}{4} \mu(T_t)(1 + 0.5(1- \omega _0) Z_{rot}^{-1})^{-1} \nabla_{\vec{x}}  T_t ,\\
			\vec{q}_r &= -R \mu(T_t)(\sigma + (1- \sigma )(1 - \omega _1) Z_{rot}^{-1})^{-1} \nabla_{\vec{x}} T_r ,
		\end{aligned}
	\end{equation}
	where $\mu(T_t) = \tau p_t$ and the pressure $p_t$ is related to the translational temperature only through $p_t = \rho R T_t$.
We can modify the computed coefficients in the expansion of Maxwellian to get the above heat fluxes by re-scaling the translational and rotational temperature gradients, such as changing $a_5 = 2\partial _x \lambda_t$ and $a_6 = 2 \partial_x \lambda_r$ to
	\begin{align*}
		\tilde{a}_6 &=  \frac{3}{2}(\sigma + (1- \sigma )(1 - \omega _1) Z_{rot}^{-1})^{-1}a_6 \\
		\tilde{a}_5 &=  (1 + 0.5(1- \omega _0) Z_{rot}^{-1})^{-1}a_5.
	\end{align*}
	Thus, only few additional floating point operations are needed for each spatial slope reconstruction to correct the heat flux, and the final form of $\vec{F}_{eq,s}$ becomes
	\begin{equation}\label{eqflux}
		\begin{aligned}
			\vec{F}_{eq,s}=\int \vec{v}\cdot \vec{n}_s \vec{\psi}\bigg\{ &C_1\left( M_t(\vec{x}_s,t^n,\vec{v},\vec{\xi}) + \frac{M_{eq}(\vec{x}_s,t^n,\vec{v},\vec{\xi}) - M_t(\vec{x}_s,t^n,\vec{v},\vec{\xi})}{Z_{rot}} \right) \\
			+&C_2\left[a^lH[\bar{u}]+a^r(1-H[\bar{u}])\right]\bar{u}  M_t(\vec{x}_s,t^n,\vec{v},\vec{\xi})\\
			+&C_2(b\bar{v}+c\bar{w})M_t(\vec{x}_s,t^n,\vec{v},\vec{\xi})
			+C_3AM_t(\vec{x}_s,t^n,\vec{v},\vec{\xi})\bigg\} d\Xi,
		\end{aligned}
	\end{equation}
with the above scaled coefficients.

	\subsubsection{The evolution of particles}\label{section223}
	The simulation particle $P_k(m_k,\vec{x}_k,\vec{v}_k, e_k, t_{f,k}, \omega_k, \kappa_k)$ is represented by its mass $m_k$,
	position coordinate $\vec{x}_k$, velocity coordinate $\vec{v}_k$,  free streaming time $t_{f,k}$ and internal energy $e_k$. $\omega_k$ and $\kappa_k$ are the weights coming from the Rykov kinetic model.
	Recall that the evolution of particles follows the integral form of the Rykov model,
	\begin{equation}\label{integral-solution1}
		f(\vec{x},t,\vec{v},\vec{\xi})=\frac{1}{\tau}\int_{0}^t e^{-(t-t')/\tau} M^{*}(\vec{x}',t',\vec{v},\vec{\xi}) dt'+e^{-t/\tau}f_0(\vec{x}-\vec{v}t).
	\end{equation}
	The Maxwellian distribution function  $M^*$ around the interface can be expanded as
	\begin{equation*}
		M^*(\vec{x}',t',\vec{v},\vec{\xi})=M^*(\vec{x},t,\vec{v},\vec{\xi})+\nabla_{\vec{x}} M^*(\vec{x},t,\vec{v},\vec{\xi})\cdot(\vec{x}'-\vec{x})+\partial_t M^*(\vec{x},t,\vec{v},\vec{\xi})t'+O((\vec{x}'-\vec{x})^2,t'^2).
	\end{equation*}
	The integral solution becomes
	\begin{equation}\label{particle}
		f(\vec{x},t,\vec{v},\vec{\xi})=(1-e^{-t/\tau})M^+(\vec{x},t,\vec{v},\vec{\xi})+e^{-t/\tau}f_0(\vec{x}-\vec{v}t).
	\end{equation}
	A first order approximation of $M^+$ can be expressed as
	\begin{equation}\label{1st-particle}
		M^+(\vec{x},t,\vec{v},\vec{\xi})={M}^*(\vec{x},t,\vec{v},\vec{\xi}),
	\end{equation}
	and the second order expansion gives
	\begin{equation}\label{2nd-particle}
		M^+(\vec{x},t,\vec{v},\vec{\xi})=M^*(\vec{x},t,\vec{v},\vec{\xi})		+\frac{e^{-t/\tau}(t+\tau)-\tau}{1-e^{-t/\tau}}(\partial_tM^*(\vec{x},t,\vec{v},\vec{\xi})+\vec{v}\cdot\nabla_{\vec{x}}M^*(\vec{x},t,\vec{v},\vec{\xi})).
	\end{equation}
	Above $M^+$ is named as the hydrodynamic distribution function with analytical formulation.
	For UGKP method, the approximation (\ref{1st-particle}) for $M^+$ is used for a simple particle-sampling algorithm \cite{bird1994molecular}.
	The particle evolution equation Eq.\eqref{particle} means that the simulation particle has a probability of $e^{-t/\tau}$ to free stream, and has a probability of $(1-e^{-t/\tau})$ to collide with other particle and the post-collided velocity follows the velocity distribution $M^+(\vec{x},t,\vec{v},\vec{\xi})$.
	The time for the free streaming to stop and follow the distribution $M^+$ is called the first collision time $t_c$.
	The cumulative distribution function of the first collision time is
	\begin{equation}\label{tc-distribution}
		F(t_c<t)=1-\exp(-t/\tau),
	\end{equation}
	from which $t_c$ can be sampled as $t_c=-\tau\ln(\eta)$ with $\eta$ generated from a uniform distribution $U(0,1)$. For a particle $P_k$, the free streaming time can be given as,
	\begin{equation}\label{freetime}
		t_{f,k} =
		\begin{cases}
			-\tau \ln (\eta)&\mbox{if $	-\tau \ln (\eta) < \Delta t$} ,\\
			\Delta t &\mbox{if $	-\tau \ln (\eta) > \Delta t$, }
		\end{cases}
	\end{equation}
	where $\Delta t$ is the time step. In a numerical time step from $t^{n}$ to $t^{n+1}$, all simulating particles in UGKP method can be categorized into two groups: the \textbf{collisionless particle} \index{collisionless particle} $P^f$ and the \textbf{collisional particle} \index{collisional particle} $P^c$.
	The categorization is based on the relation between the free streaming time $t_{f}$ and the time step $\Delta t$.
	More specifically, the collisionless particle is defined as the particle whose free streaming time $t_f$ greater than or equal to the time step $\Delta t$,
	and the collisional particle is defined as the particle whose free streaming time $t_f$ smaller than $\Delta t$.
	For the collisionless particle, its trajectory is fully tracked during the whole time step.
	For collisional particle, the particle trajectory is tracked till $t_f$.
	Then the particle's mass, momentum, and energy are merged into the macroscopic quantities in that cell and the simulation particle gets eliminated.
	Those eliminated particles will get re-sampled once the updated macroscopic quantities $\vec{W}^{n+1}$ are obtained.
	As shown in Eq.\eqref{particle}, the re-sampled particles follow the hydrodynamic distribution $M^+$ and therefore they are defined as \textbf{hydro-particle} \index{hydro-particle} $P^h$. The macroscopic quantities corresponding to the hydro-particles are defined as \textbf{hydro-quantities} \index{hydro-quantities} $\vec{W}^h$.
	The hydro-particles will be sampled at the beginning of each time step and become the candidates for collisionless/collisional particles
	again in the next time step evolution according to their newly-sampled $t_f$.
		
	At the beginning of each step, we need to sample particles from $M^*$ defined in Eq.(\ref{newdist}).
	For cell $\Omega_i$ with hydro quantities $\vec{W}^h_{i}=\left(\rho_{i}^h,\rho_{i}^h U_{i}^h,\rho_{i}^h V_{i}^h,\rho_{i}^h W_{i}^h,\rho_{i}^h E_{i}^h, \rho_{i}^h E_{rot,i}^h\right)^T$, using the stratification for variance reduction,
	hydro-particles can be sampled from the modified Maxwellian distribution $\tilde{M}_t$ with a total mass of $({(Z_{rot} - 1)}/(Z_{rot}))\rho_i^h |\Omega_i| $ and the modified Maxwellian distribution $\tilde{M}_{eq}$ with a total mass of $\rho_i^h |\Omega_i| /Z_{rot}$, respectively.
	
	Taking $\tilde{M}_t$ as an example, the reduced distribution function for rotational variable $\vec{\xi}$ can be written as
	\begin{equation*}
		\begin{aligned}
			G_t &= \int \tilde{M}_t d\vec{\xi} = G_m(\lambda_t)\left[  1 +\frac{4\lambda_t^2\vec{q}_t \cdot \vec{c}}{15\rho} \left( 2\lambda_t \vec{c}^2 - 5 \right) \right] ,\\
			R_t &= \int  \vec{\xi}^2 \tilde{M}_t d\vec{\xi} = \frac{K_r}{2\lambda_r}G_m(\lambda_t) \left[ 1 + \frac{4\lambda_t^2\vec{q}_t \cdot \vec{c}}{15\rho} \left( 2\lambda_t \vec{c}^2 - 5\right) + \frac{8(1 - \sigma)\lambda_t \lambda_r \vec{q}_r \cdot \vec{c}}{K_r \rho}\right]
		\end{aligned}
	\end{equation*}
	with
	\begin{equation*}
		G_m(\lambda) = \rho\left( \frac{\lambda}{\pi} \right)^{\frac{3}{2}}e^{-\lambda \vec{c}^2} .
	\end{equation*}
Following the idea of importance sampling, the hydro quantities $\vec{W}^h_i$ from the moments of the above specified distributions can be rewritten as,
	
	\begin{equation}\label{importance}
		\begin{aligned}
			\vec{W}^h_i&=\int\left(\begin{array}{c}
			G_t \\
			\vec{v} G_t \\
			\frac{\vec{v}^2}{2} G_t+	\frac{1}{2}R_t \\
			\frac12 R_t
			\end{array}\right)_i d\vec{v} = \int\left(\begin{array}{c}
			\frac{G_t}{G_m} G_m \\
			\vec{v} \frac{G_t}{G_m} G_m \\
			\left(\frac{\vec{v}^2}{2}\frac{G_t}{G_m}+	\frac{1}{2}\frac{R_t}{G_m} \right)G_m \\
			\frac12 \frac{R_t}{G_m} G_m
			\end{array}\right)_i d\vec{v} \\
			&\approx \sum \left(\begin{array}{c}
			\omega_k \frac{m_p}{|\Omega|_i}  \\
			\omega_k \frac{m_p \vec{v}_k}{|\Omega|_i} \\
			\omega_k\frac{m_p\vec{v}_k^2}{2|\Omega|_i}+	\kappa_k\frac{m_p e_k}{2|\Omega|_i} \\
			\kappa_k\frac{m_p e_k}{2|\Omega|_i}
			\end{array}\right) .
		\end{aligned}
	\end{equation}

	In order to recover the gas distribution function on the microscopic level, the sampled particles $P_k$, $k=1,...,N_i$, follows
	\begin{equation}\label{sample1}
	\begin{aligned}
	&m_p=\frac{Z_{rot} - 1}{Z_{rot}}\frac{\rho_{i}^h |\Omega|_i}{N_i}, \quad \vec{x}_{k}\sim \text{U}(\Omega_i), \quad e_k	=\frac{K_r }{2{\lambda_{r,i}}}, \\
	&\vec{v}_k=(-\ln (\vec{\eta_1})/{\lambda_{t,i}})^{1/2}\cos(2\pi\vec{\eta}_2) +\vec{U}_{i}, \quad \vec{\eta}_{1,2}\sim \text{U}(0,1)^3,
	\end{aligned}
	\end{equation}
	where $\text{U}(\Omega_i)$ is the uniform distribution on $\Omega_i$ and $\text{U}(0,1)^3$ is the uniform distribution on $(0,1)^3$.
	The addition weights $\omega_k$ and $\kappa_k$ are required in order to make the macroscopic quantities of the sampled particles $P_k$ consistent with the macro quantities $\vec{W}_{i}^h$. As shown in Eq. (\ref{importance}), $\omega_k$ and $\kappa_k$ are determined by the coefficients $G_t/G_m$ and $R_t/G_m$, respectively, i.e.
	\begin{equation}\label{sample2}
	\begin{aligned}
	\omega_k &= 1 +\frac{4\lambda_t^2\vec{q}_t \cdot \vec{c}_k}{15\rho} \left( 2\lambda_t \vec{c}_k^2 - 5 \right),\\
	\kappa_k &= 1 + \frac{4\lambda_t^2\vec{q}_t \cdot \vec{c}_k}{15\rho} \left( 2\lambda_t \vec{c}_k^2 - 5\right) + \frac{8(1 - \sigma)\lambda_t \lambda_r \vec{q}_r \cdot \vec{c}_k}{K_r\rho}.
	\end{aligned}
	\end{equation}

	Now all the quantities of $P_k$ are determined. The particle will take free streaming for a period of ${t_f}_k$,
	\begin{equation}\label{stream}
		\vec{x}_{k}^{n+1} = \vec{x}_{k}^{n} + \vec{v}_k t_{f,k}.
	\end{equation}
	The net free streaming flow of cell $i$ at the next step can be calculated by counting the particles passing through the cell interface, which can be written as,
	\begin{equation}\label{particleevo}
		\vec{W}_{fr,i}=\frac{1}{\Delta t}\left(  \sum_{k\in P_{\partial \Omega_i^{+}}} \vec{W}_{P_k} - \sum_{k\in P_{\partial \Omega_i^{-}}} \vec{W}_{P_k}\right),
	\end{equation}
	where $\vec{W}_{P_k}=\left( \omega_km_k,\omega_km_k \vec{v}_k,\frac12 m_k \left(\omega_k\vec{v}_k^2+\kappa _ke_k\right),\frac12 m_k\kappa _ke_k \right)^T$, $P_{\partial \Omega_i^{+}}$ is the index set of the particles streaming into cell $\Omega_i$ during a time step, and $P_{\partial \Omega_i^{-}}$ is the index set of the particles streaming out of cell $\Omega_i$.

	\subsubsection{The update of macroscopic flow variables with source term}\label{update}
	Since the source term in the rotational energy, $\rho E_{rot}$ can be updated using a semi-implicit scheme.
    Based on $\vec{W}^*$ as defined by
	\begin{equation}\label{updateW}
		\vec{W}^*_i=\vec{W}^n_i+\frac{\Delta t}{|\Omega_i|} \left(\sum_{l_s\in\partial \Omega_i}|l_s| \vec{F}_{eq,s} + \vec{W}_{fr,i}\right),
	\end{equation}
	 $\rho ^{n+1}, (\rho \vec{U} )^{n+1}$ and $(\rho E)^{n+1}$ can be updated. Then, from the Eq.(\ref{equilibriumEr}), the equilibrium rotational energy $(\rho E_{rot}^{eq})^{n+1}$ can be obtained as well, from which the source term for rotational energy can be approximated as,
	\begin{equation}\label{source1}
		S = \frac{\Delta t}{2}\left( \frac{2(\rho E_{rot}^{eq})^{n+1} - (\rho E_{rot})^* - (\rho E_{rot})^{n+1}}{(Z_{rot}\tau)^*} \right),
	\end{equation}
	thus
	\begin{equation}\label{source2}
		(\rho E_{rot})^{n+1} =\left( 1 + \frac{\Delta t}{2(Z_{rot}\tau)^*} \right)^{-1} \left( (\rho E_{rot})^* + \frac{\Delta t}{2}\left( \frac{2(\rho E_{rot}^{eq})^{n+1} - (\rho E_{rot})^*}{(Z_{rot}\tau)^*} \right) \right) .
	\end{equation}

	In UGKP, the evolution of microscopic particle is coupled with the evolution of macroscopic flow variables. The composition of the particles during time evolution in the
	UGKP method is illustrated in Fig.\ref{ugkp} \cite{zhu2019unified}. The algorithm of UGKP method for diatomic gas can be summarized as following:
	\begin{enumerate}
		\item
		Sample the particle quantities $(m_k,\vec{x}_k,\vec{v}_k, e_k, \omega_k, \kappa_k)$ by Eq.\eqref{sample1} and Eq.\eqref{sample2} for each newly added particle $P_k$ from the hydro-quantities $\vec{W}^h$. For the first step, $\vec{W}^h = \vec{W}^{n = 0}$ as shown in Fig. \ref{ugkp1}.
		
		\item
		Sample free streaming time $t_{f,k}$ by Eq.(\ref{freetime}) for each particle $P_k$ and they will be classified into collisionless particles (white circles in Fig. \ref{ugkp2}) and collisional particles (solid circles in Fig. \ref{ugkp2}). Then, stream the particles by Eq.\eqref{stream}.
		
		\item
		Calculate the net free streaming flow $\vec{W}_{fr}$ by Eq.\eqref{particleevo}, and evaluate the equilibrium flux $\vec{F}^{eq}$ by Eq.\eqref{eqflux}.
		
		\item
		Update total macroscopic flow variables $\vec{W}$ by Eq.\eqref{updateW} and Eq.\eqref{source2}. Calculate the macro-quantities of collisionless particles $\vec{W}^p$ by collecting the macro-quantities of collisionless particles, from which the hydro-quantities $\vec{W}^h$ are obtained by $\vec{W}^h=\vec{W}-\vec{W}^p$. [Shown in Fig. \ref{ugkp3}]
		
		\item
		Keep collisionless particles and remove collisional particles and then go to step 1. Fig. \ref{ugkp4} is the initial state in Fig. \ref{ugkp1} for the next time step.
	\end{enumerate}

	\begin{figure}[h!]
		\centering
		\begin{subfigure}[b]{0.24\textwidth}
			\includegraphics[width=\textwidth]{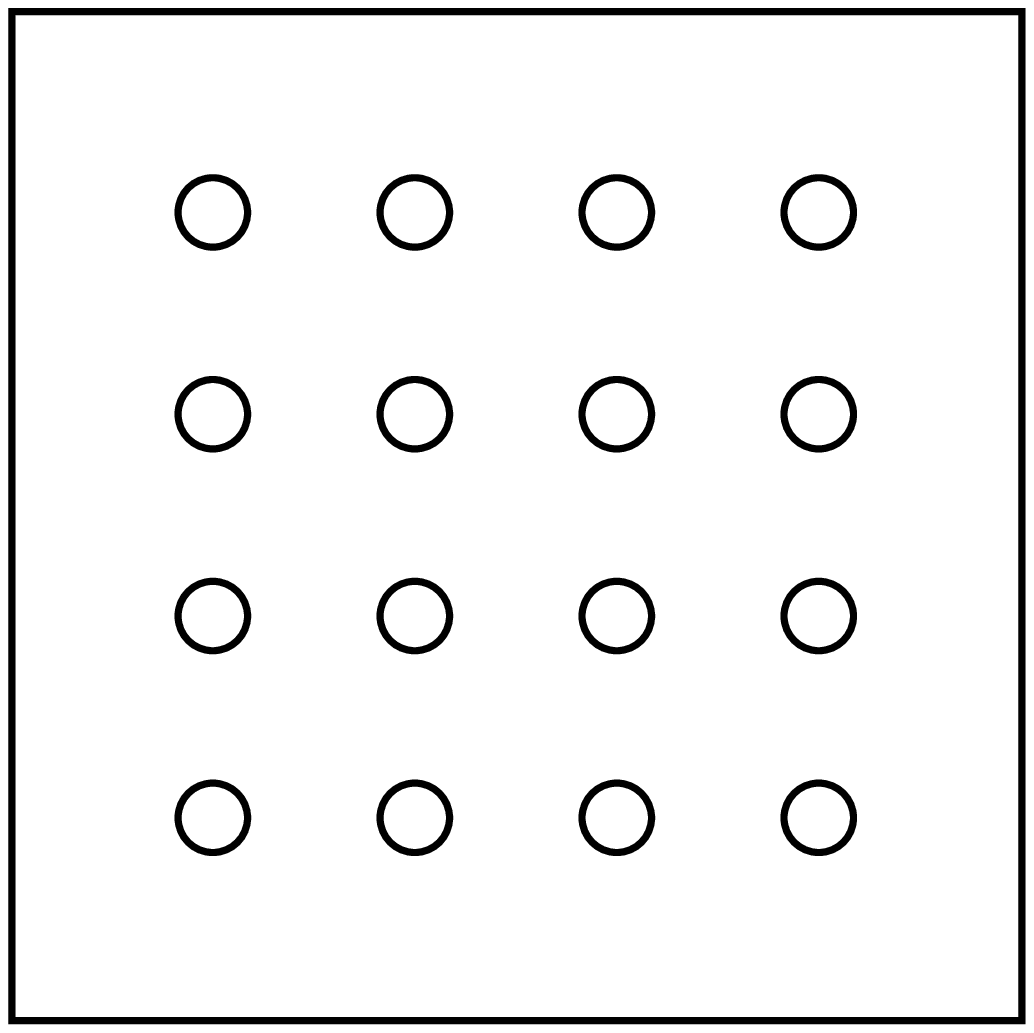}
			\caption{}
			\label{ugkp1}
		\end{subfigure}
		\begin{subfigure}[b]{0.24\textwidth}
			\includegraphics[width=\textwidth]{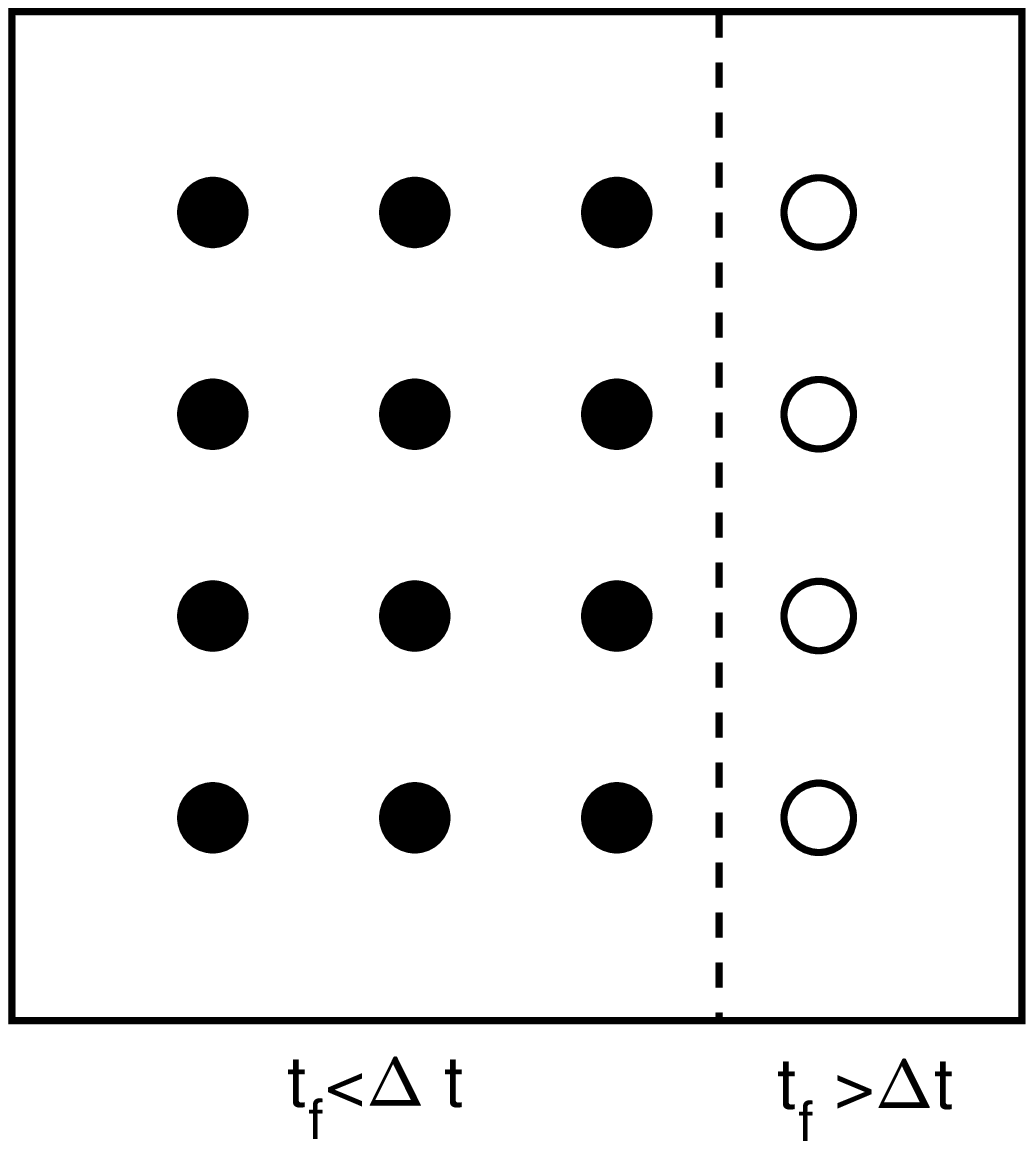}
			\caption{}
			\label{ugkp2}
		\end{subfigure}
		\begin{subfigure}[b]{0.24\textwidth}
			\includegraphics[width=\textwidth]{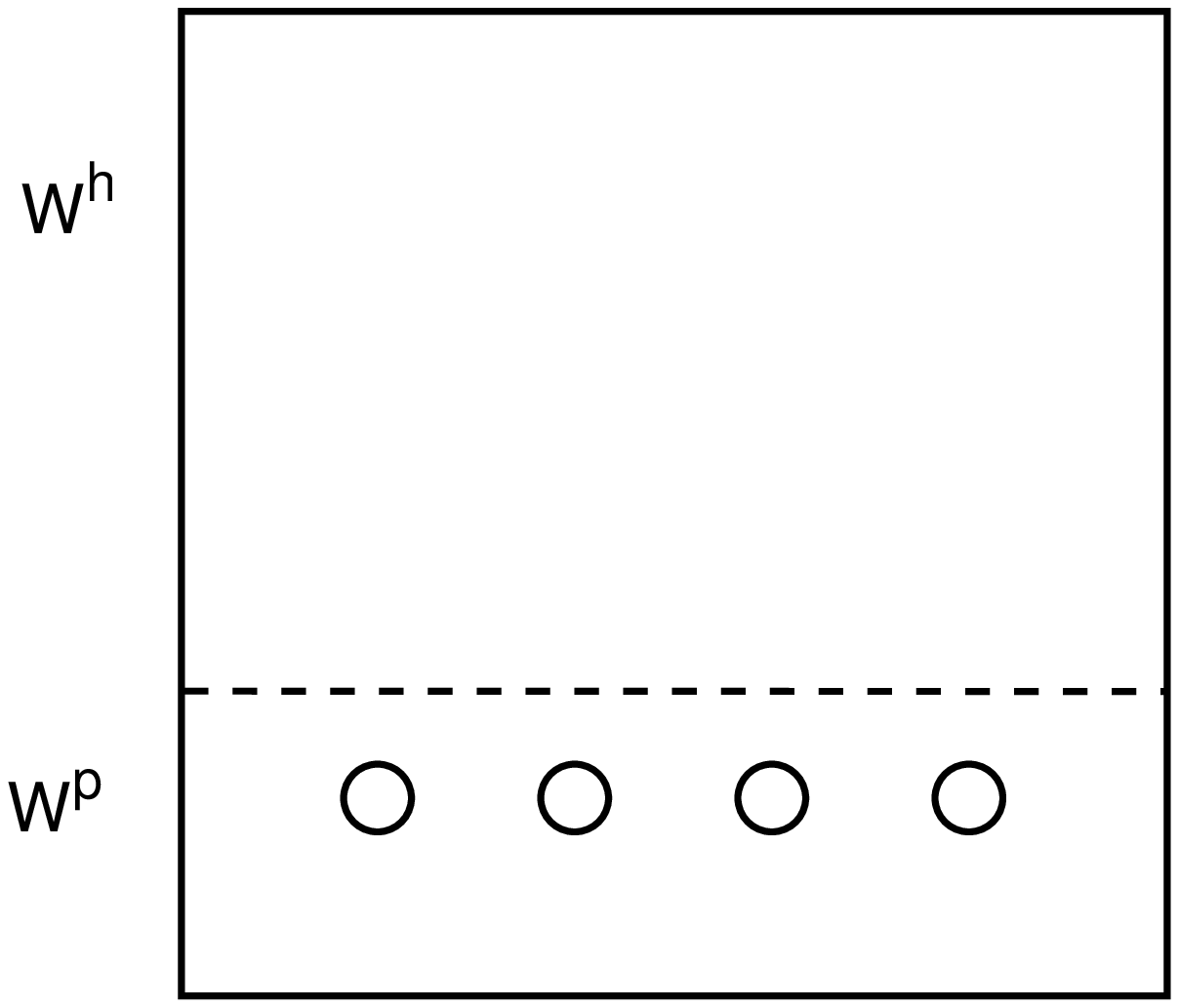}
			\caption{}
			\label{ugkp3}
		\end{subfigure}
		\begin{subfigure}[b]{0.24\textwidth}
			\includegraphics[width=\textwidth]{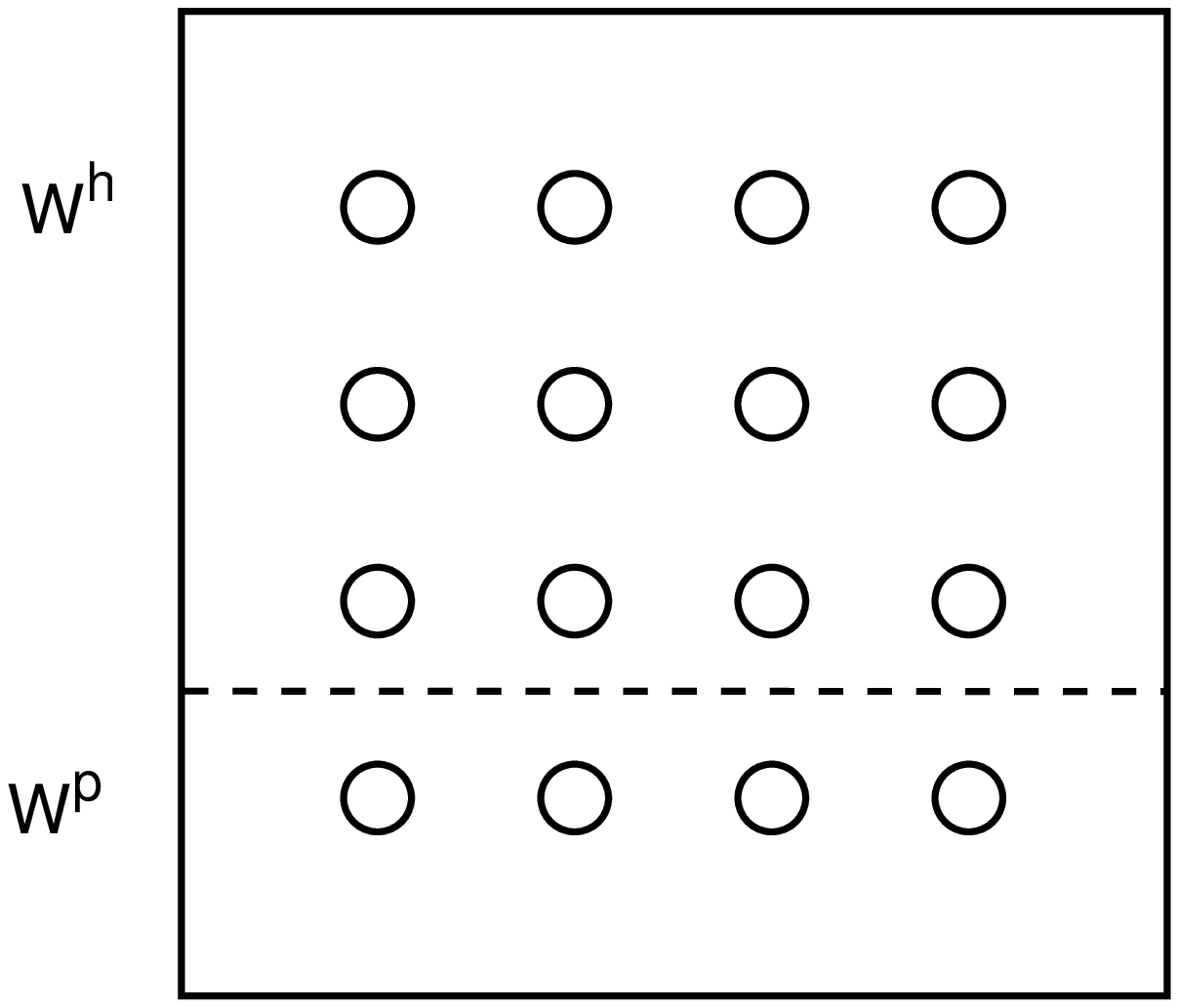}
			\caption{}
			\label{ugkp4}
		\end{subfigure}
		\caption{Diagram to illustrate the composition of the particles during time evolution in the UGKP method. (a) Initial field; (b) classification of the collisionless particles (white circle) and collisional particles (solid circle) according to the free transport time $t_f$; (c) update solution at macroscopic level;(d) update solution at microscopic level.}
	\label{ugkp}
	\end{figure}

	\subsection{Unified gas-kinetic wave-particle method}\label{method2}
	The UGKWP method improves UGKP method mainly in the following two aspects:
	\begin{itemize}
		\item The free transport terms in numerical flux contributed by the collisional 	hydro-particles are evaluated analytically;
		\item Only collisionless hydro-particles are sampled.
	\end{itemize}
	Firstly, since the distribution of the hydro-quantities $\vec{W}^{h}$ is known as $M^+$, the flux contributed by the free transport of collisional hydro-particles can be evaluated analytically,
	\begin{equation*}
		\sum_s |l_s| \vec{F}_{fr,s}=\sum_s |l_s| \vec{F}_{fr,s}^h+\vec{W}_{fr,i}^p,
	\end{equation*}
	where $\vec{F}_{fr,s}^h$ is the free transport flux contributed by the hydro-quantities \cite{liu2020unified},
	\begin{equation}\label{Ff1}
		\vec{F}_{fr,s}^h=\int  \left[C_4M^+(\vec{x}_s,t^n,\vec{v},\vec{\xi}) + C_5 \vec{v} \cdot \nabla_{\vec{x}} M_t(\vec{x}_s,t^n,\vec{v},\vec{\xi})\right]  \vec{v}\cdot\vec{n}_s \vec{\psi} d\Xi .
	\end{equation}
	where $M^+$ is given in Eq.\eqref{2nd-particle}, $C_4 = \frac{\tau}{\Delta t} \left(1 - e^{-\Delta t / \tau}\right) -  e^{-\Delta t / \tau}$, and $C_5 = \tau  e^{-\Delta t / \tau} - \frac{\tau^2}{\Delta t}(1 -  e^{-\Delta t / \tau}) +\frac12\Delta t e^{-\Delta t / \tau} $.
	
	Secondly, since the numerical flux contributed by the streaming of collisional particles can be evaluated by $\vec{F}_{fr,s}^h$ analytically, only the collisionless hydro-particle \index{collisionless hydro-particles} will be sampled.
	Based on the cumulative distribution function of the first collision time Eq.\eqref{tc-distribution},
	the collisionless hydro-particles are sampled from a $M^+(\vec{W}^{n+1}_i)$ with the total mass of $e^{-\Delta t/\tau_i}|\Omega_i| \rho^h_i$. Then, the net free streaming flux contributed by the streaming of all left collisionless and collisional particles can be calculated by Eq.(\ref{particleevo}),
	\begin{equation*}
		\vec{W}_{fr,i}^p=\frac{1}{\Delta t}\left(  \sum_{k\in P_{\partial \Omega_i^{+}}} \vec{W}_{P_k} - \sum_{k\in P_{\partial \Omega_i^{-}}} \vec{W}_{P_k}\right).
	\end{equation*}
	Therefore, the evolution of macroscopic flow variables in Eq.(\ref{updateW}) now becomes
	\begin{equation}\label{updateWP}
		\vec{W}_i^*=\vec{W}_i^n+\frac{\Delta t}{|\Omega_i|} \left(\sum_{l_s\in\partial \Omega_i}|l_s| \vec{F}_{eq,s}+\sum_{l_s\in\partial \Omega_i}|l_s| \vec{F}_{fr,s}^h + \vec{W}_{fr,i}^p\right),
	\end{equation}
	from which  $\rho ^{n+1}, (\rho \vec{U} )^{n+1}$ and $(\rho E)^{n+1}$ can be obtained. Following the same calculations in Eq.(\ref{source1}) and Eq.(\ref{source2}), $(\rho E_{rot})^{n+1}$ can be obtained as well. The composition of the particles during time evolution in the
	UGKWP method is illustrated in Fig.\ref{ugkwp} \cite{liu2020unified}.
	The algorithm of UGKWP method for diatomic gases can be summarized as following:
	\begin{enumerate}
		\item
		Sample the particle quantities $(m_k,\vec{x}_k,\vec{v}_k, e_k, \omega_k, \kappa_k)$ by Eq.\eqref{sample1} and Eq.\eqref{sample2} for each newly added collisionless hydro-particle $P_k$ from the hydro-quantities $e^{-\Delta t / \tau}\vec{W}^h$. These particles are all defined as collisionless particles which have $t_f = \Delta t$. For the first step, $W^h = W^{n = 0}$ as shown in Fig. (\ref{ugkwp1}).
		
		\item
		Sample free streaming time $t_{f,k}$ by Eq. (\ref{freetime}) for particles $P_k$ from $W^p$. These particles are classified into collisionless particles (white circles in Fig. \ref{ugkwp2}) and collisional particles (solid circles in Fig. (\ref{ugkwp2})). Then, stream all the particles by Eq.\eqref{stream}. For the first step, $W^p = 0$.
		
		\item
		Calculate the net free streaming flow $\vec{W}_{fr}$ by Eq.\eqref{particleevo}, and evaluate the equilibrium flux $\vec{F}^{eq}$ and free transport flux $\vec{F}^h_{fr,s}$ by  Eq.\eqref{eqflux} and Eq.\eqref{Ff1}, respectively.
		
		\item
		Update total flow variables $\vec{W}$ by Eq.\eqref{updateWP} and Eq.\eqref{source2}. Calculate the macro-quantities of collisionless particles $\vec{W}^p$ by collecting the macro-quantities of collisionless particles, from which  the hydro-quantities $\vec{W}^h$ are obtained by $\vec{W}^h=\vec{W}-\vec{W}^p$. [Shown in Fig. \ref{ugkwp3}]
		
		\item
		Keep collisionless particles and remove collisional particles. Then, go to step 1. Fig. \ref{ugkwp4} is the initial state of next time step in Fig. \ref{ugkwp1}.
	\end{enumerate}
	\begin{figure}
		\centering
		\begin{subfigure}[b]{0.24\textwidth}
			\includegraphics[width=\textwidth]{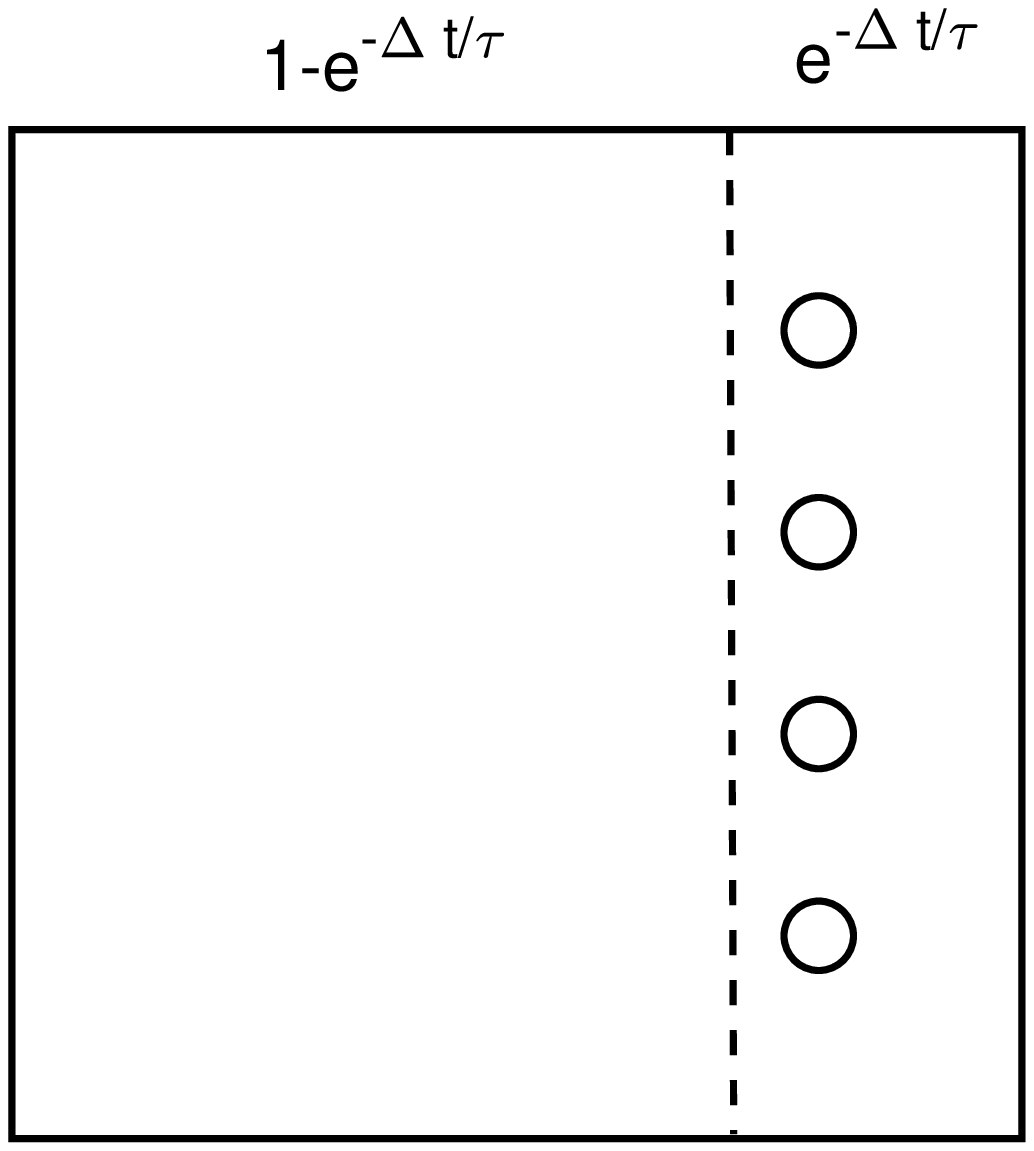}
			\caption{}
			\label{ugkwp1}
		\end{subfigure}
		\begin{subfigure}[b]{0.24\textwidth}
			\includegraphics[width=\textwidth]{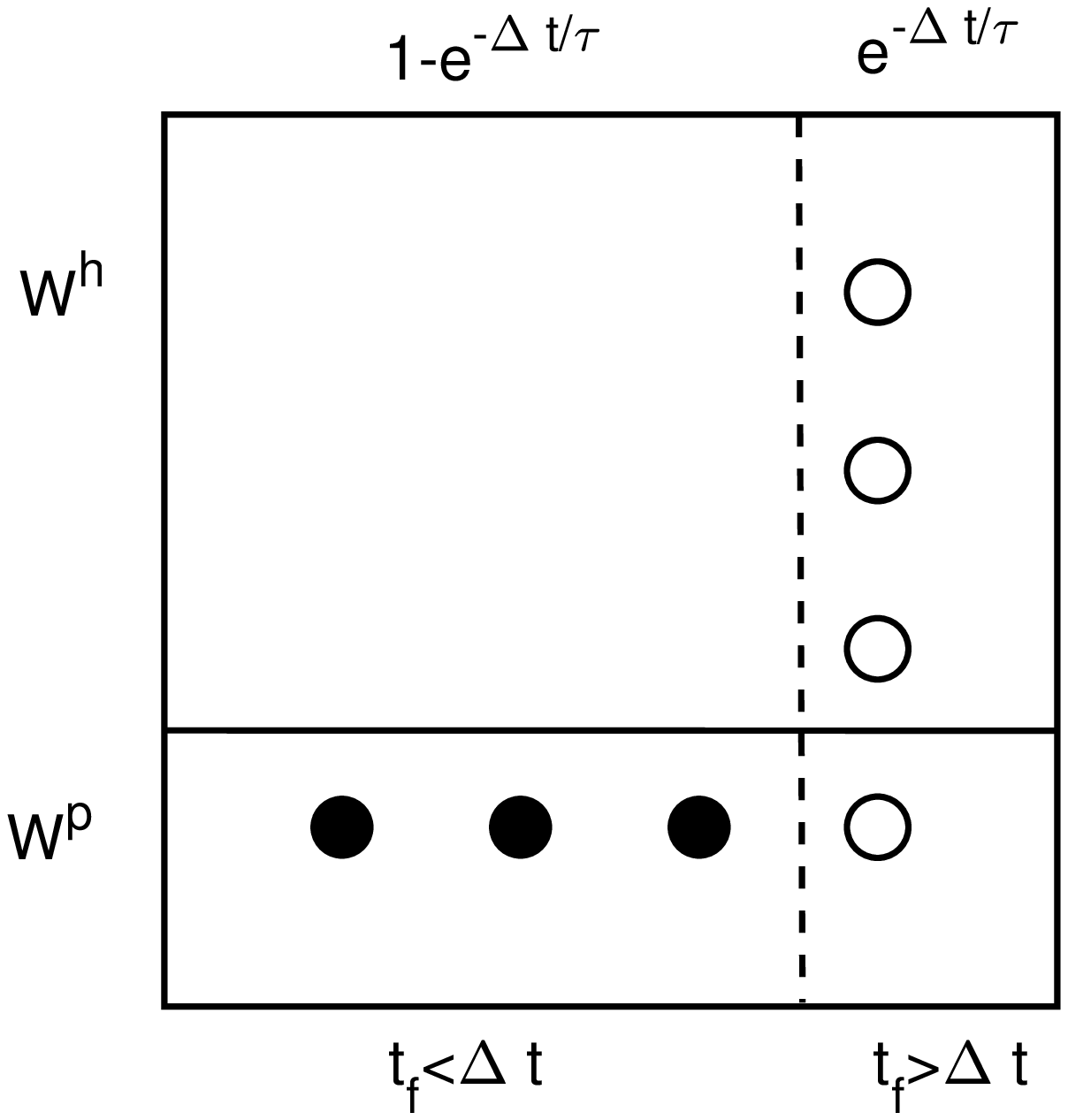}
			\caption{}
			\label{ugkwp2}
		\end{subfigure}
		\begin{subfigure}[b]{0.24\textwidth}
			\includegraphics[width=\textwidth]{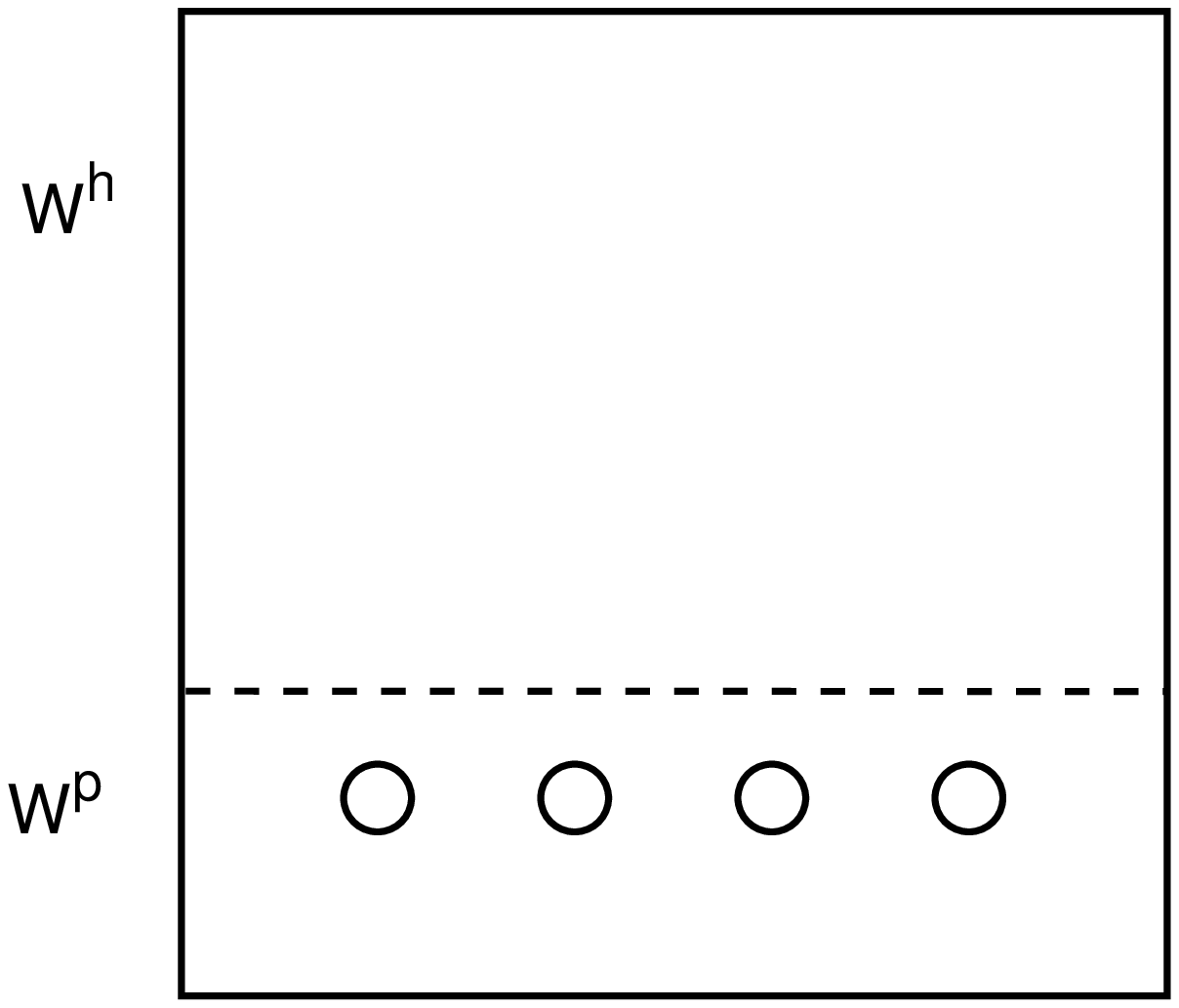}
			\caption{}
			\label{ugkwp3}
		\end{subfigure}
		\begin{subfigure}[b]{0.24\textwidth}
			\includegraphics[width=\textwidth]{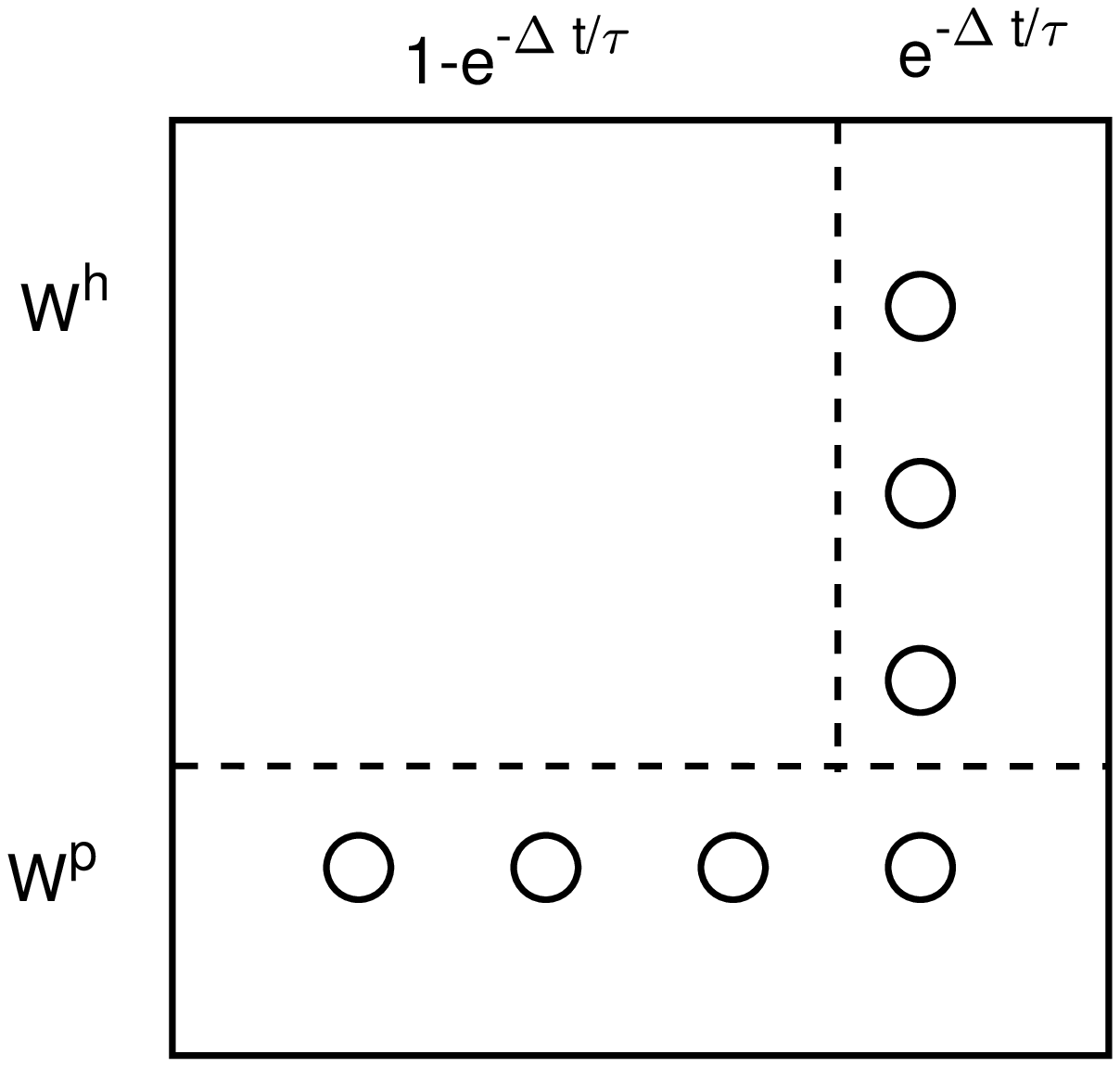}
			\caption{}
			\label{ugkwp4}
		\end{subfigure}
		\caption{Diagram to illustrate the composition of the particles during time evolution in the UGKWP method. (a) Initial field; (b) classification of the collisionless and collisional particles for $\vec{W}^p$; (c) update on macroscopic level; (d) update on the microscopic level.}
		\label{ugkwp}
	\end{figure}

	\section{Analysis and discussion}\label{discussion}

	\subsection{Asymptotic behavior in the continuum regime}
		
	In this section, we are going to analyze the asymptotic behavior of the UGKWP method with diatomic relaxation in continuum regime.
	For simplicity, the following analysis is based on two-dimensional case. Following the Chapman-Enskog procedure, one can show that the macro description of the Rykov model can be written as \cite{rykov1978macroscopic},
	\begin{equation}\label{macro}
		\begin{aligned}
		&\frac{\partial \rho}{\partial t} + \frac{\partial (\rho U)}{\partial x} +\frac{\partial (\rho V)}{\partial y} = 0 ,\\
	 	&\frac{\partial (\rho U)}{\partial t} + \frac{\partial (\rho U^2 + p_t)}{\partial x} +\frac{\partial (\rho UV)}{\partial y} =  \frac{\partial \tau_{xx}}{\partial x} + \frac{\partial \tau_{yx}}{\partial y} ,\\
		&\frac{\partial (\rho V)}{\partial t}+ \frac{\partial \rho UV}{\partial x} +\frac{\partial (\rho V^2 + p_t)}{\partial y} = \frac{\partial \tau_{xy}}{\partial x} + \frac{\partial \tau_{yy}}{\partial y} ,\\
		&\frac{\partial (\rho E)}{\partial t} + \frac{\partial (\rho E U + p_tU)}{\partial x} +\frac{\partial  (\rho E V + p_tV)}{\partial y} = \frac{\partial (U \tau_{xx} + V \tau_{xy} + q_x)}{\partial x} + \frac{\partial (U \tau_{yx} + V \tau_{xx} + q_y)}{\partial y} ,\\
		&\frac{\partial (\rho E_r)}{\partial t} + \frac{\partial (\rho E_r U)}{\partial x} +\frac{\partial  (\rho E_r V)}{\partial y} = \frac{\partial q_{rx}}{\partial x} + \frac{\partial q_{ry}}{\partial y} + \frac{\rho E_r^{eq} - \rho E_r}{Z_{rot}\tau}. \\
		\end{aligned}
	\end{equation}
	Here the viscous and heat conduction terms are
	\begin{equation}\label{macrostress}
		\begin{aligned}
		&\tau_{xx} = \tau p_t \left[ 2 \frac{\partial U}{\partial x} - \frac{2}{3} \left(\frac{\partial U}{\partial x} +\frac{\partial V}{\partial y}   \right) \right] ,\\	
		&\tau_{yy} = \tau p_t \left[ 2 \frac{\partial V}{\partial y} - \frac{2}{3} \left(\frac{\partial U}{\partial x} +\frac{\partial V}{\partial y}   \right) \right] ,\\
		&\tau_{xy} = \tau_{yx} = \tau p_t  \left(\frac{\partial U}{\partial y} +\frac{\partial V}{\partial x}   \right) \\
		&(q_x, q_y)^T = \vec{q}_r +  \vec{q}_t ,\\
		&(q_{rx},q_{ry})^T = \vec{q}_r ,
		\end{aligned}
	\end{equation}
	where $\vec{q}_r$ and $\vec{q}_t$ are defined in Eq.\eqref{heatflux}. The pressure $p_t = \rho R T_t$ is only related to the translational temperature.
	\begin{proposition}[Asymptotic-preserving property]
	Consider a well resolved flow region with $M_t^l=M_t^r$ and $\nabla_{\vec{x}}^l M_t=\nabla_{\vec{x}}^r M_t$, for fixed time $\Delta t$, and small $\tau$, the scheme is asymptotically equivalent, up to $O(\tau^2)$, a first order scheme for the system \eqref{macro} and \eqref{macrostress}.
	\end{proposition}
	\begin{proof}
	For cell $i$, the total mass of the sampled collisionless hydro-particles is
	$m^h=e^{-\Delta t/\tau_i} |\Omega_i| \rho_i^h$, and the total mass of the collisionless and collisional particles $m^p$ is proportional to $m^h$,
	i.e. $m^p\sim O(e^{-\Delta t/\tau_i})$.
	Therefore, the numerical flux contribution by collisionless and collisional particle streaming is $\vec{W}_{fr,i}^p\sim O(e^{-\Delta t/\tau_i})$.

	In the free transport flux of hydro-quantities $\vec{F}_{fr,s}^h$ given by Eq.\eqref{Ff1},
	the hydrodynamic distribution function $M^+(\vec{x}_0,t,\vec{v},\vec{\xi})$ become
	\begin{equation}\label{glimit}
		M^+(x_0, y_0 , t, u, v , \vec{\xi})=M^*-\tau
		\left(\partial_t M_t + u \partial _x M_t + v \partial _y M_t \right)+O(\tau^2) .\\
	\end{equation}
	Substituting Eq.\eqref{glimit} into Eq.\eqref{Ff1}, the total flux $\vec{F}_{an,s}$ of the macroscopic variables becomes
	\begin{equation}\label{fluxlimit}
		\begin{aligned}
			\vec{F}_{an,s} =& \vec{F}_{eq,s} + \vec{F}_{fr,s}^h\\
			=&\int u \bigg\{ (C_1 + C_4)M^* + (C_2 - \tau C_4 + C_5)(u \partial _x M_t + v 	\partial _y M_t) \\
			&+ (C_3 - \tau C_4)\partial_t M_t \bigg\}\vec{\psi}dudvd\vec{\xi}+O(\tau^2) \\
			=& \int u \bigg\{ M^* - \tau (u \partial _x M_t + v \partial _y M_t + \partial_t M_t ) + \frac12 \Delta t\partial_t M_t \bigg\}\vec{\psi}dudvd\vec{\xi} +O(\tau^2).
		\end{aligned}
	\end{equation}
	If  $O(\tau^2)$ terms are neglected, Eq.\eqref{fluxlimit} becomes
	\begin{equation}\label{numF}
	\begin{aligned}
		\vec{F}_{an,s} &=  \int u \bigg\{ M^* - \tau (u \partial _x M_t + v \partial _y M_t + \partial_t M_t ) + \frac12 \Delta t\partial_t M_t \bigg\}\vec{\psi}dudvd\vec{\xi}  \\
		&= \left(\begin{array}{c}
		\rho U \\
		\rho U^2 + p_t + \tau_{xx} \\
		 \rho UV +\tau_{xy}\\
		\rho E + p_t U + U\tau_{xx} + V\tau_{xy} + q_x \\
		\rho E_r + \rho E_r U + q_{rx}
		\end{array}\right)
		+  \frac12\Delta t\left(\begin{array}{c}
		\frac{\partial \rho U}{\partial t}\\
		\frac{\partial \rho U^2}{\partial t} + \frac{\partial p_t}{\partial t}\\
		\frac{\partial \rho UV}{\partial t}\\
		\frac{\partial \rho E}{\partial t} +\frac{\partial p_t U}{\partial t}   \\
		\frac{\partial \rho E_r}{\partial t} + \frac{\partial \rho E_r U}{\partial t}
		\end{array}\right)
	\end{aligned}
	\end{equation}
	It can be observed that the numerical flux is consistent with the flux in system \eqref{macro} and \eqref{macrostress}. Therefore, in the continuum regime, the UGKWP method converges to Eq.\eqref{macro}, which is a first order scheme for the system \eqref{macro} and \eqref{macrostress}.
	\end{proof}

	For the limiting Euler system with the absence of viscous and heat conduction terms, we can have the following proposition.
	\begin{proposition}
	Consider a well resolved flow region with $M_t^l=M_t^r$ and $\nabla_{\vec{x}}^l M_t=\nabla_{\vec{x}}^r M_t$, for fixed time $\Delta t$, in the limit $\tau \to 0$, the scheme becomes a second order scheme for the limiting Euler system with the absence of viscous and heat conduction terms in  system \eqref{macro} and \eqref{macrostress}.
	\end{proposition}
	\begin{proof}
	As $\tau \to 0$, following directly from Eq.\eqref{numF}, we have
 	\begin{equation*}
 	\begin{aligned}
 	\vec{F}_{an,s} = \left(\begin{array}{c}
 	\rho U \\
 	\rho U^2 + p_t \\
 	\rho UV \\
	 \rho E + p_t U\\
	 \rho E_r + \rho E_r U
 	\end{array}\right)
 	+  \frac12\Delta t\left(\begin{array}{c}
 	\frac{\partial \rho U}{\partial t}\\
 	\frac{\partial \rho U^2}{\partial t} + \frac{\partial p_t}{\partial t}\\
 	\frac{\partial \rho UV}{\partial t}\\
 	\frac{\partial \rho E}{\partial t} +\frac{\partial p_t U}{\partial t}   \\
 	\frac{\partial \rho E_r}{\partial t} + \frac{\partial \rho E_r U}{\partial t}
 	\end{array}\right) .
 	\end{aligned}
 	\end{equation*}
 	Combine with the semi-implicit update of source term, it can be observed that this is exactly a second order scheme for the limiting Euler system.
	\end{proof}

	In the limit of total equilibrium state with $Z_{rot} = 1$, both the translational temperature and the rotational temperature converges  to the equilibrium temperature as $\tau \to 0$. From Eq.\eqref{TdifferenceOrder}, the pressure $p_t$ can be rewritten as,
	\begin{equation}\label{EOS}
		p_t = p + p_t - p = p - \frac{4}{15}Z_{rot}\tau p\left( \frac{\partial U}{\partial x} +\frac{\partial V}{\partial y} \right).
	\end{equation}
	When $Z_{rot} = 1$, the second term on the right hand side of Eq. (\ref{EOS}) is exactly the bulk viscosity for rotational degrees of freedom in NS equations.
	
	In the continuum regime with $\Delta t \gg \tau$, for a fixed particle mass $m_p$, the number of sampled collisionless hydro-particles in cell $i$ is $e^{-\Delta t/\tau_i} |\Omega_i| \rho_i^h/m_p$. And the total simulation particle number $N_p$ in such regime decreases exponentially, $N_p\sim O(e^{-\Delta t/\tau})$. Therefore, the computational cost of UGKWP in continuum regime becomes comparable to hydrodynamic NS solvers,
such as recovering GKS for the NS solution \cite{xu2001}.

	\section{Numerical tests}\label{numericaltest}

	\subsection{Rotational relaxation in a homogeneous gas}
	For a diatomic homogeneous gas with different initial translational  $T_t$ and rotational $T_r$ temperature, the system will finally evolve to an equilibrium state with $T_{eq} = T_t = T_r$. The relaxation rate is related to the rotational collision frequency. Since there is no free transport phenomenon in homogeneous case, the governing equation becomes,
	\begin{equation}\label{homogeneous}
		\frac{\partial f}{\partial t} = \frac{\tilde{M}_t-f}{\tau} +\frac{\tilde{M}_{eq}-\tilde{M}_t}{Z_{rot}\tau}.
	\end{equation}
	Multiplying Eq.(\ref{homogeneous}) with $\xi ^2$ and integrating over the velocity and rotational energy space, the time evolution equation of rotational energy becomes
	\begin{equation*}
	\frac{\partial T_r}{\partial t} = \frac{T_{eq} - T_r}{Z_{rot}\tau},
	\end{equation*}
	from which the analytical solution can be obtained,
	\begin{equation}
		T_r(t) = T_{eq} - \left( T_{eq} - T_r(0) \right)e^{-\frac{t}{Z_{rot}\tau}}.
	\end{equation}
	In the computation, the variable hard sphere (VHS) model with $\omega = 1$ (Maxwell molecule) and  $\omega = 0.72$ (Nitrogen molecule) are used and the relaxation time is approximated as $\tau = \mu / p_t$. Fig. \ref{relaxation} shows the UGKWP solutions at $Z_{rot} = 3$ and $5$ for Maxwell molecule and Nitrogen molecule along with the analytical solutions. The computational results show that the UGKWP solutions for Maxwell molecule match better with the analytical solutions while the solutions for Nitrogen molecule relax a little bit faster. The mean collision time (m.c.t) here is calculated by the equilibrium temperature $T_{eq}$, which is used to normalize $t$.
		\begin{figure}
		\centering
		\includegraphics[width=0.48\textwidth]{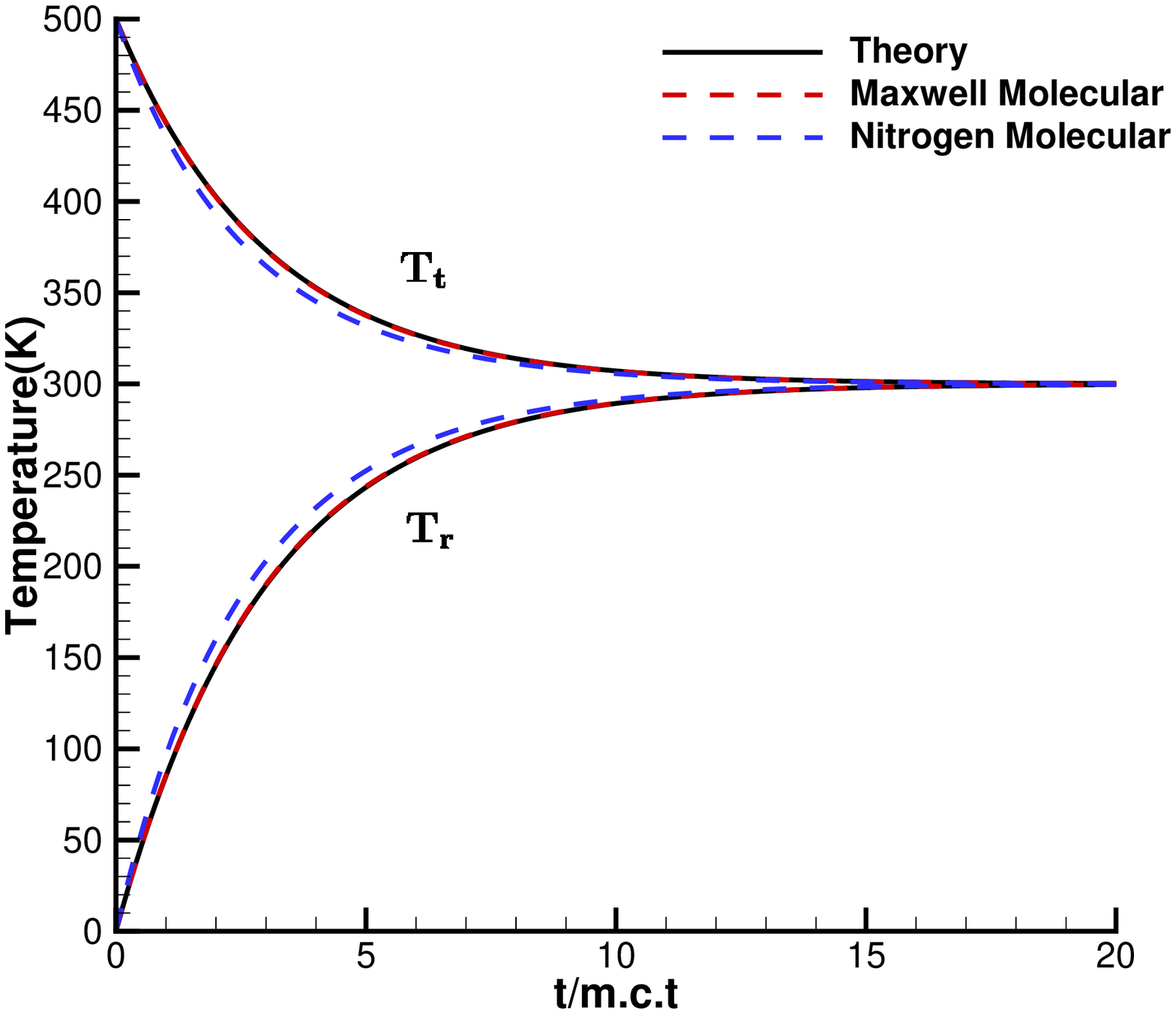}
		\includegraphics[width=0.48\textwidth]{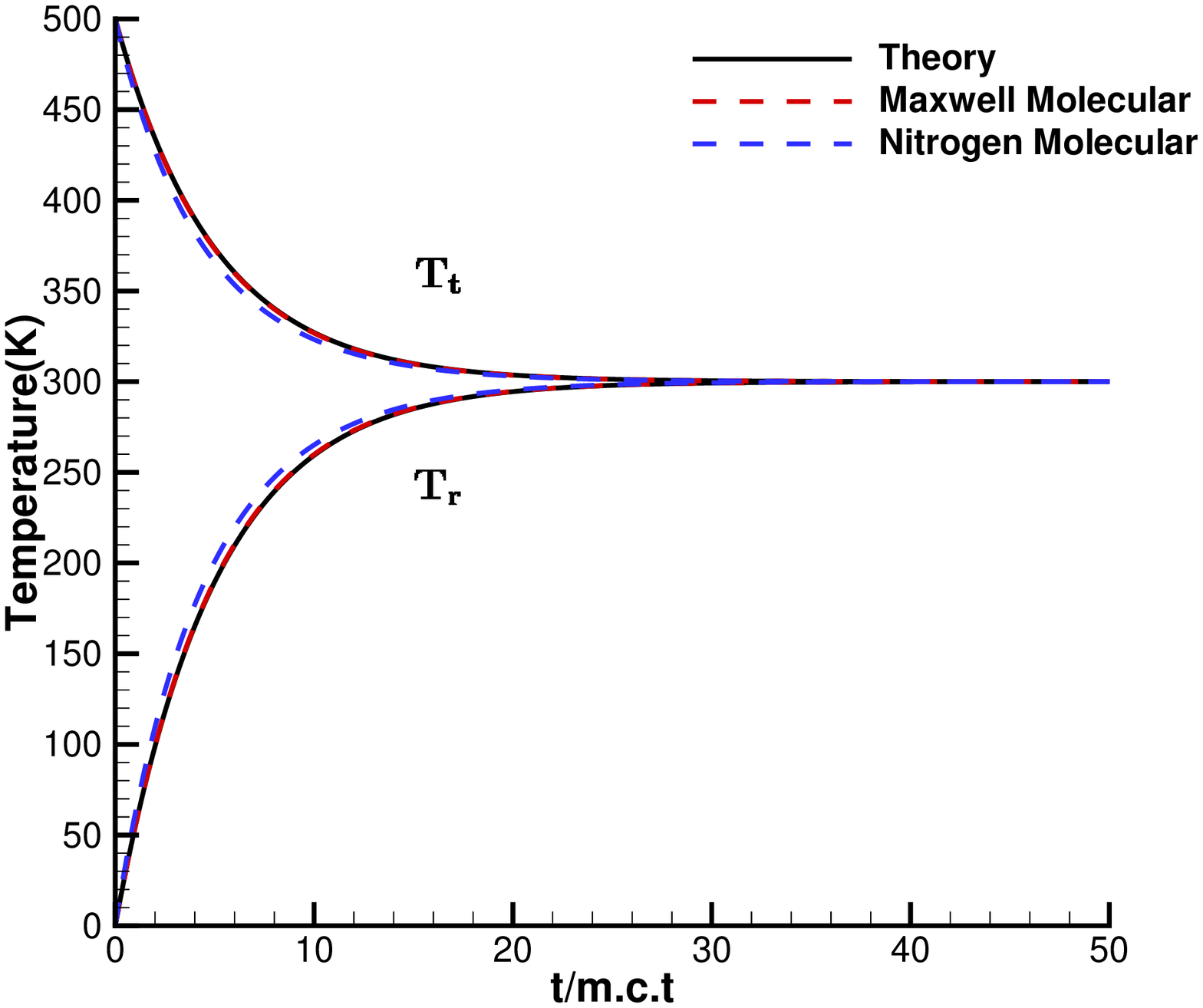}
		\caption{Rotational relaxation in a homogeneous gas. Left: $Z_{rot} = 3$. Right: $Z_{rot} = 5$}
		\label{relaxation}
	\end{figure}
	
	\subsection{1D shock tube problem}
	In 1D shock tube problem, two limiting cases of $Z_{rot} = 1$ and $Z_{rot} \to \infty$ are tested to validate the rotational relaxation. Theoretically, the solution for  $Z_{rot} = 1$ is a limiting solution for the diatomic gas with $T_t = T_r = T_{eq}$ and $Z_{rot} = \infty$ is a limiting solution for the monatomic gas. The computational domain is $[0,1]$ and initial condition is
	\begin{align*}
	(\rho, u, T_t,T_r)=\left\{\begin{array}{ll}
	(1.0,0,2.0,2.0) & x \leq 0.5 , \\
	(0.125,0,1.6,1.6) & x>0.5 .
	\end{array}\right.
	\end{align*}
	The viscous coefficient is given as
	\begin{equation}\label{vhs-vs1}
	\mu=\mu_{ref}\left(\frac{T}{T_0}\right)^{\omega},
	\end{equation}
	with the temperature dependency index $\omega=0.72$, and the reference viscosity
	\begin{equation}\label{vhs-vs2}
	\mu_{ref}=\frac{15\sqrt{\pi}}{2(5-2\omega)(7-2\omega)}\mathrm{Kn}.
	\end{equation}
	The Knudsen number is $\mathrm{Kn} = 0.0001$, and the cell size is fixed as $\Delta x = 0.01$. These two test cases are calculated up to $t = 0.12$. As shown in Fig. \ref{Sod}, the solutions of UGKS and UGKWP agree well in these two limiting cases.
	\begin{figure}
		\centering
		\includegraphics[width=0.48\textwidth]{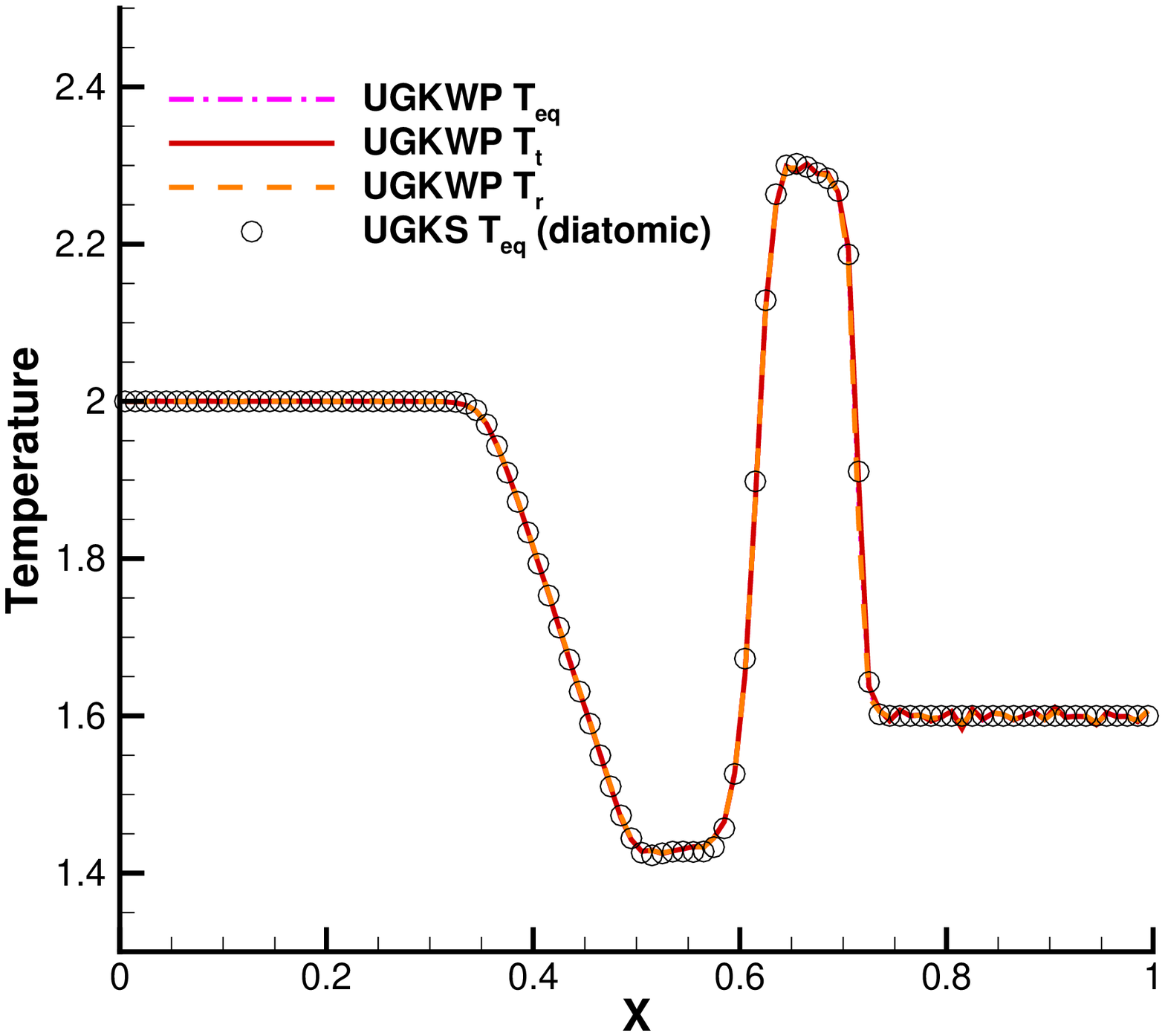}{a}
		\includegraphics[width=0.48\textwidth]{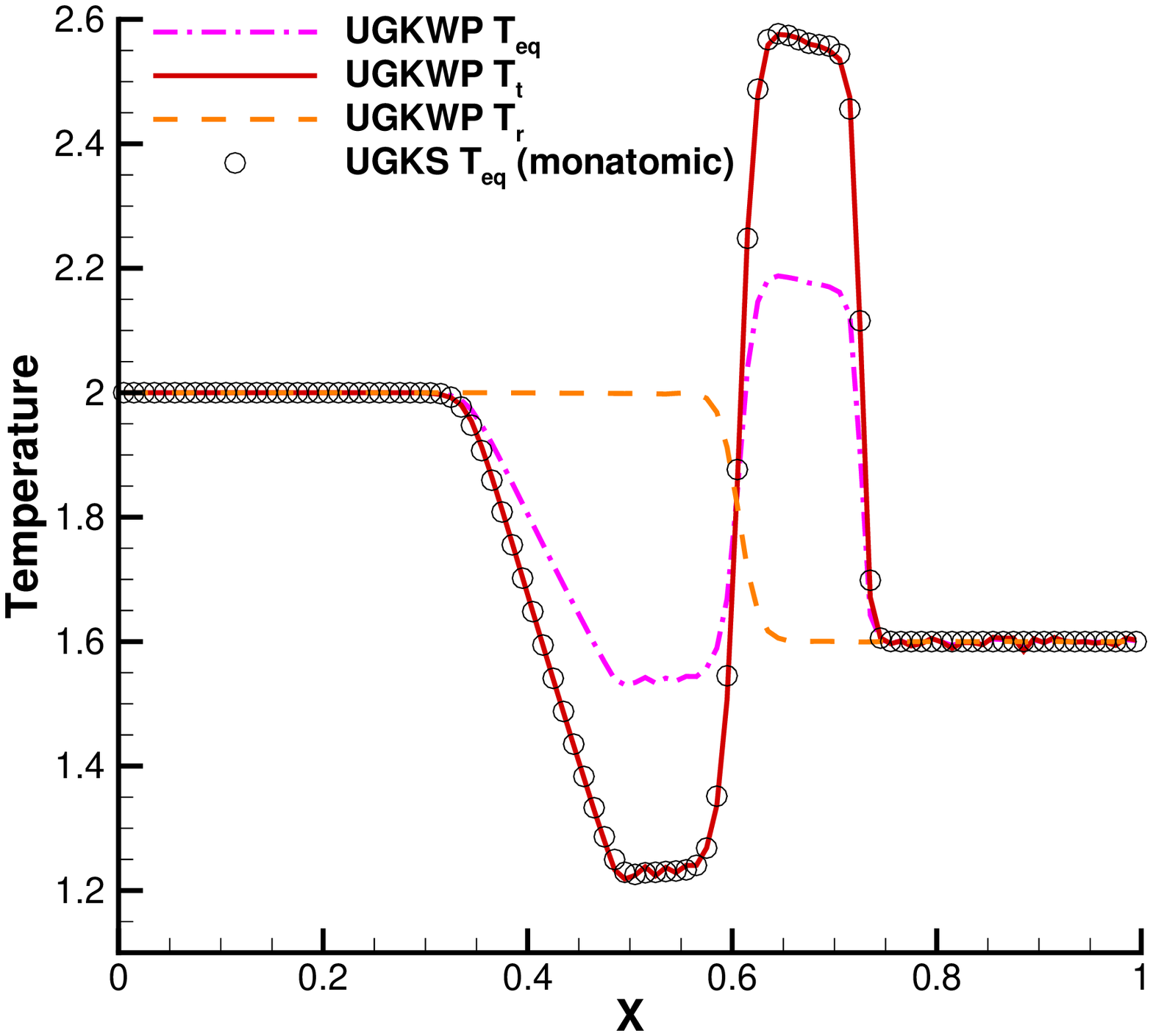}{b}
		\caption{The UGKWP results at different $Z_{rot}$. (a)$Z_{rot} = 1$; (b)$Z_{rot} = 100000$.}
		\label{Sod}
	\end{figure}
	
	\subsection{Normal shock}
	To demonstrate the accuracy of UGKWP method in capturing the highly non-equilibrium flow, the one dimensional shock wave is studied. For the nitrogen gas, the viscous coefficient follows Eq.\eqref{vhs-vs1} and Eq.\eqref{vhs-vs2} with the temperature dependency index $\omega=0.72$. In this calculation, the reference length is the upstream mean free path, and the computational domain is [-25,25] with 100 cells.
	The upstream ($x\le0$) and downstream ($x>0$) is connected by the Rankine-Hugoniot condition. The collision rotation number used in the UGKWP is $Z_{rot} = 2.4$. In order to reduce the statistical noise, $5\times10^{3}$ simulation particles are used in each cell. The normalized temperature and density from UGKWP and DSMC \cite{liu2014unified} at $\mathrm{M} = 1.53, 4.0, 5.0 ,7.0$ are plotted in Fig. \ref{shockdsmc}. As analyzed before, since the Rykov model reduces to Shakhov model at large $Z_{rot}$, the early rising of the temperature occurs at high Mach number.

	\begin{figure}
		\centering
		\includegraphics[width=0.48\textwidth]{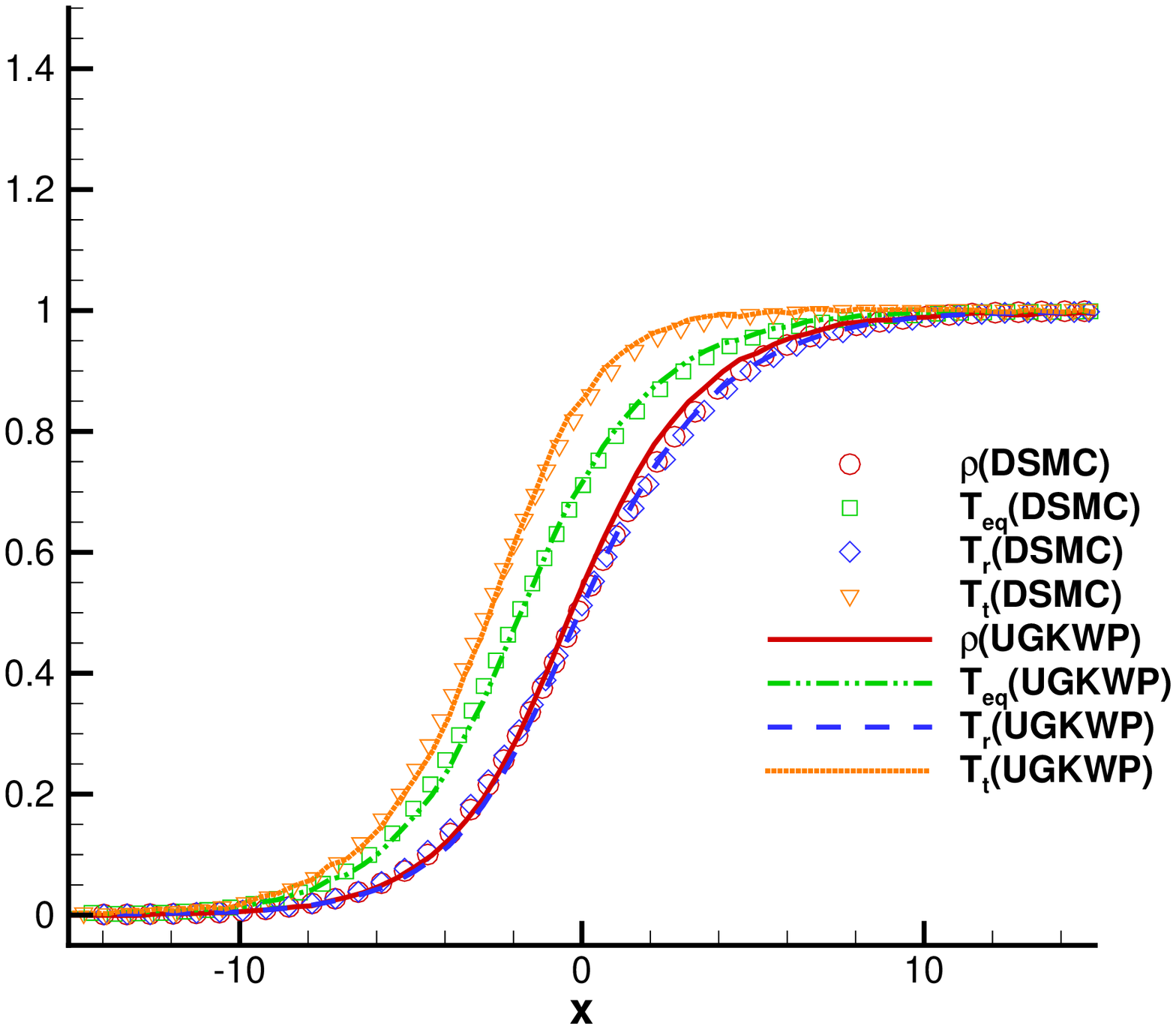}{a}
		\includegraphics[width=0.48\textwidth]{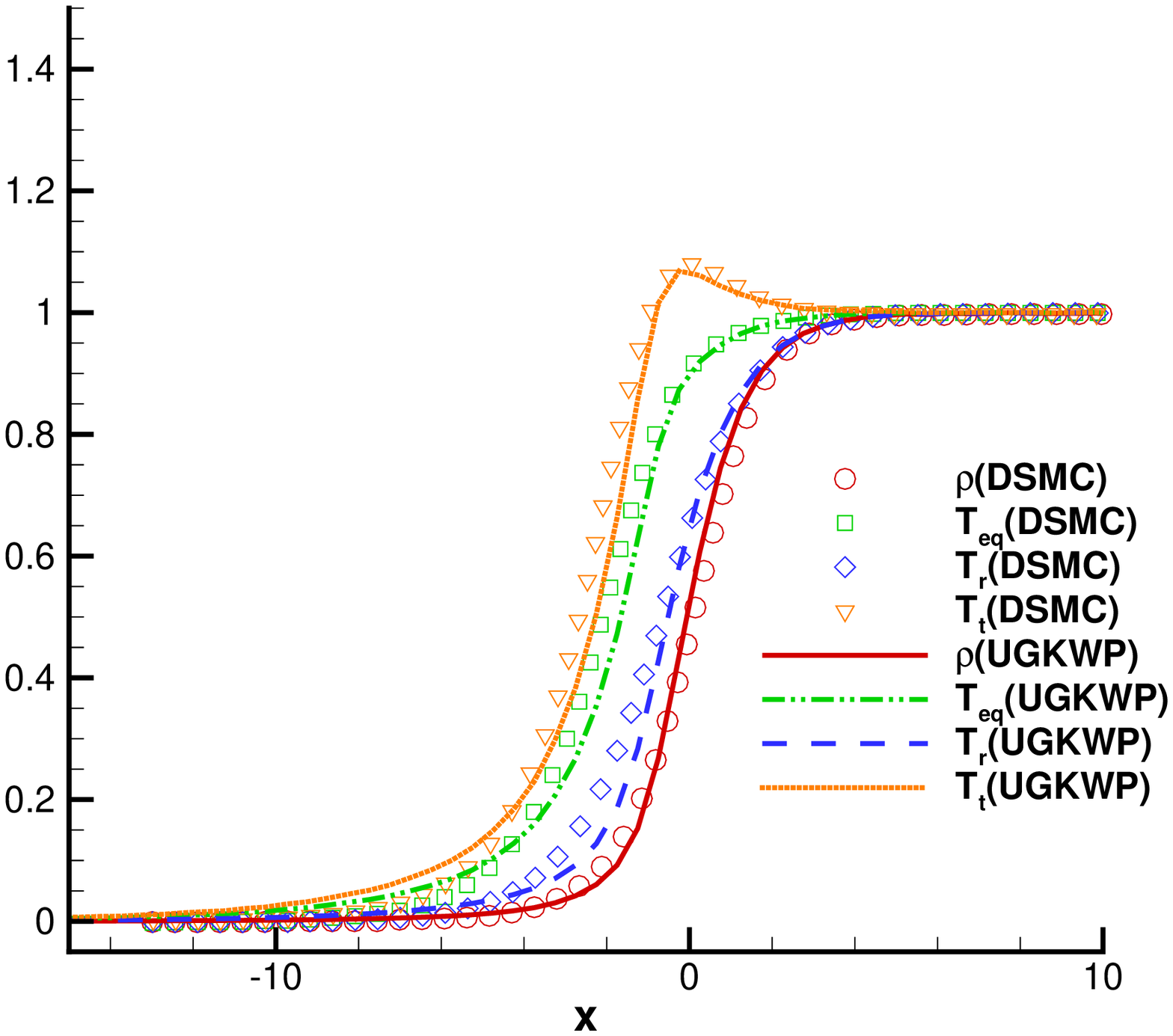}{b}
		\includegraphics[width=0.48\textwidth]{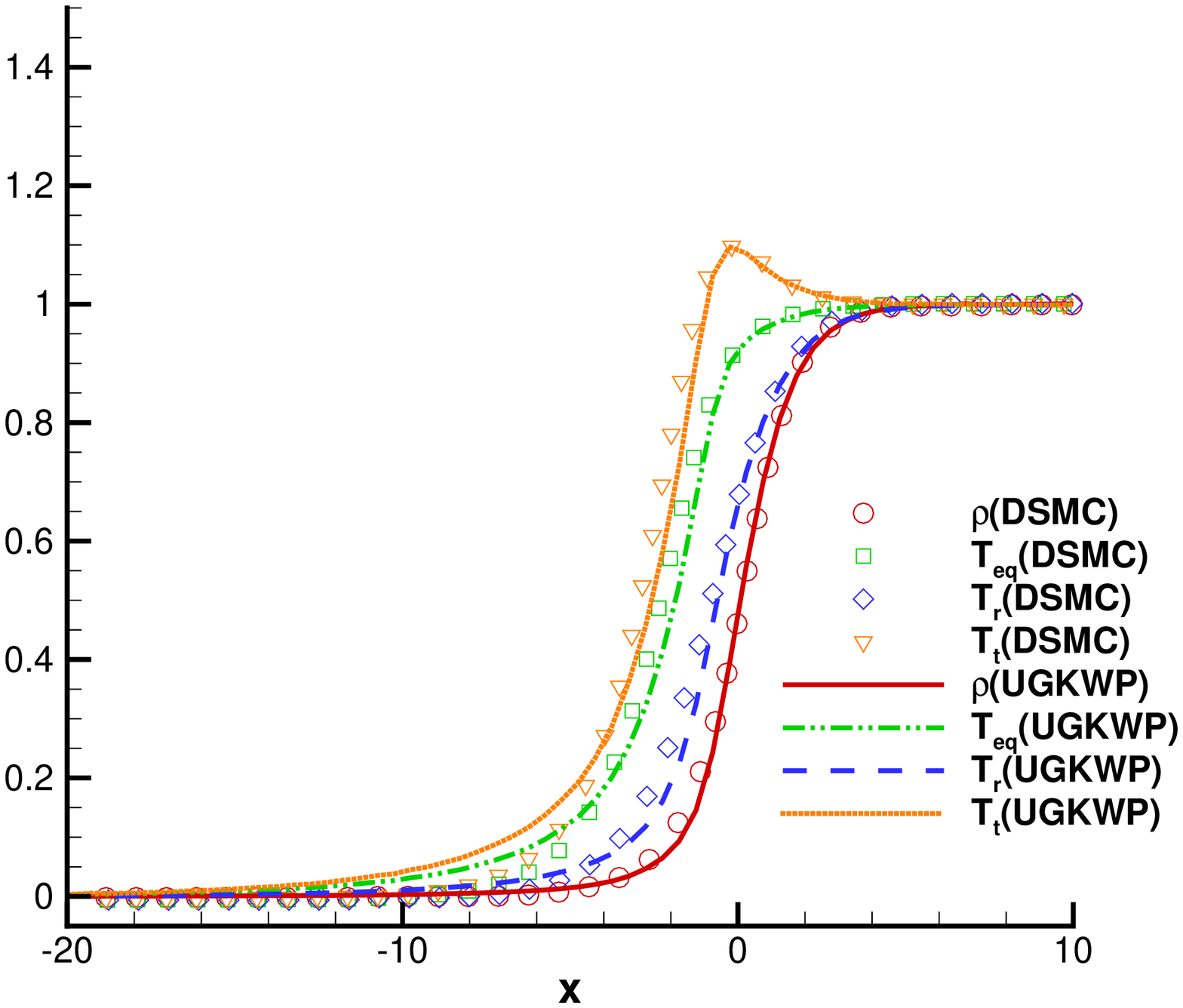}{c}
		\includegraphics[width=0.48\textwidth]{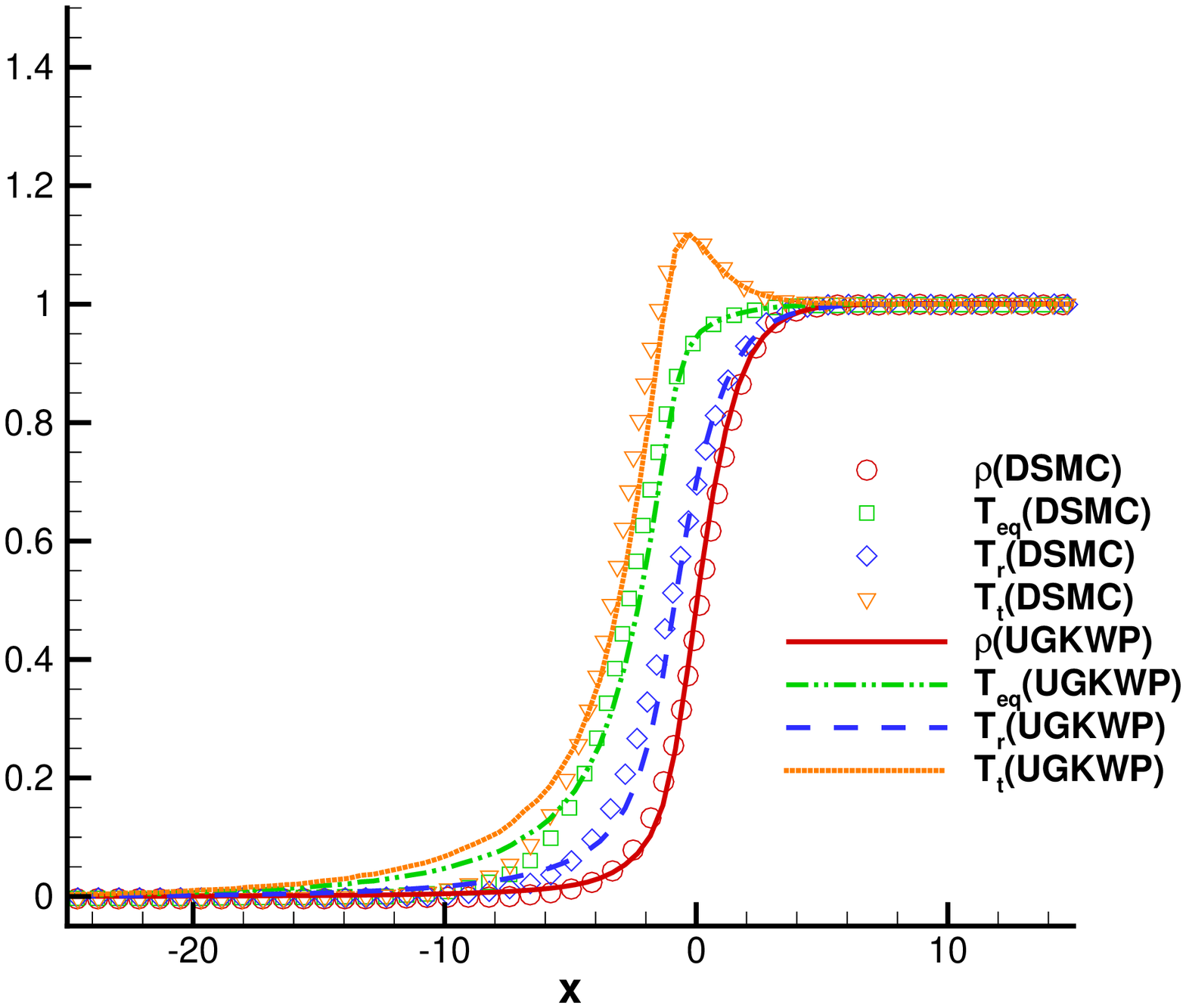}{d}
		\caption{Comparison of UGKWP and DSMC results of nitrogen shock wave at different Mach number in nitrogen. (a)$\mathrm{M} = 1.53$; (b)$\mathrm{M} = 4.0$; (c)$\mathrm{M} = 5.0$; (d)$\mathrm{M} = 7.0$}
		\label{shockdsmc}
	\end{figure}
		 Next we compare UGKWP results with experiment data for nitrogen shock waves with upstream Mach number $\mathrm{M}=7$ and $\mathrm{M}=12.9$. The rotaional collision number is given by
	\begin{equation}\label{Zrot}
	Z_{rot} = \frac{Z_{rot}^{\infty}}{1 + (\pi^2/2)\sqrt{T^*/T_t} + (\pi^2/4 + \pi)(T^*/T_t)},
	\end{equation}
	where $Z_{rot}^{\infty} = 18.0$ and $T^* = 91.5K$ are used in the computation. Fig. \ref{shock} shows the good agreement \cite{robben1966experimental}.

	\begin{figure}
		\centering
		\includegraphics[width=0.48\textwidth]{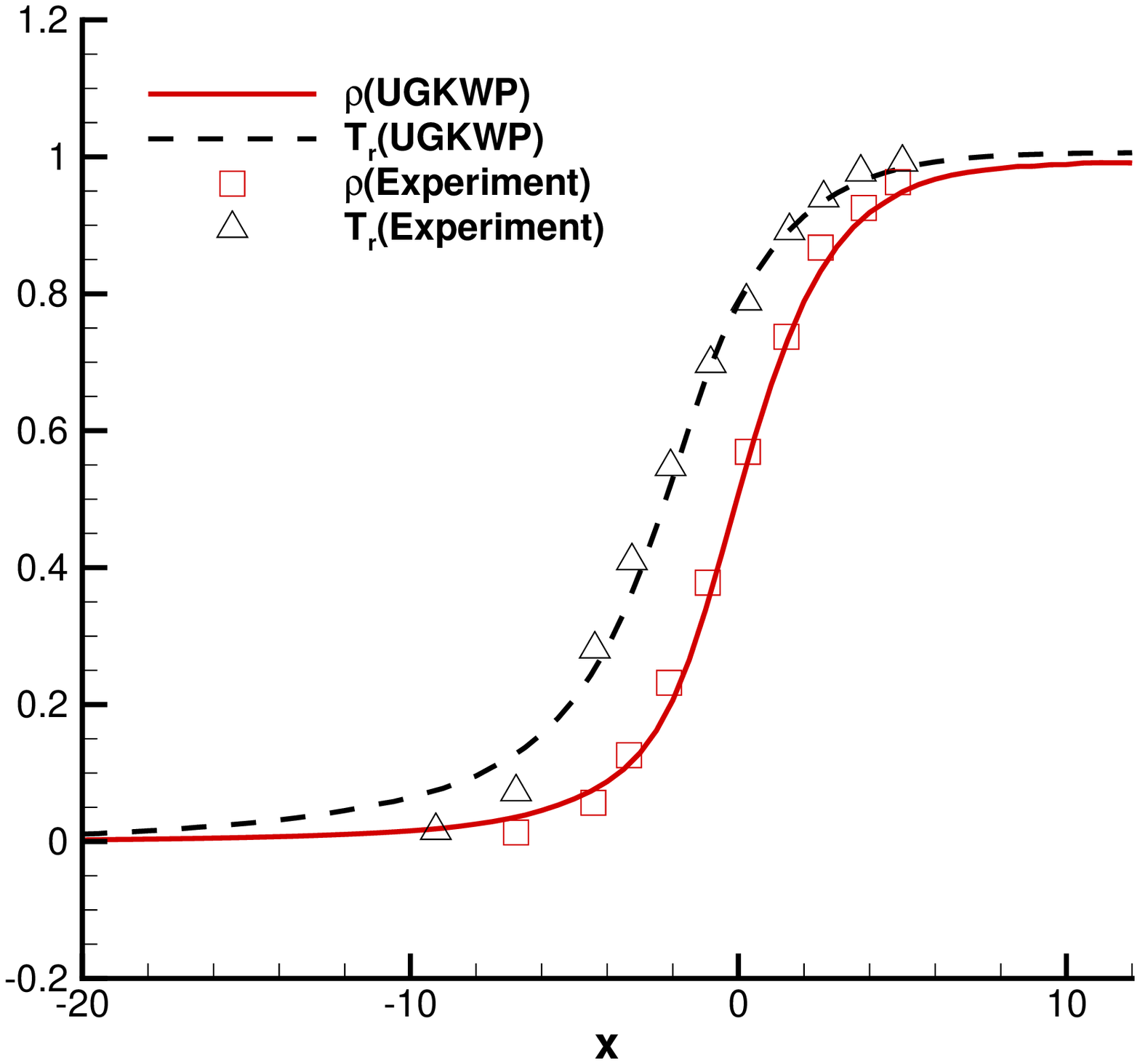}{a}
		\includegraphics[width=0.48\textwidth]{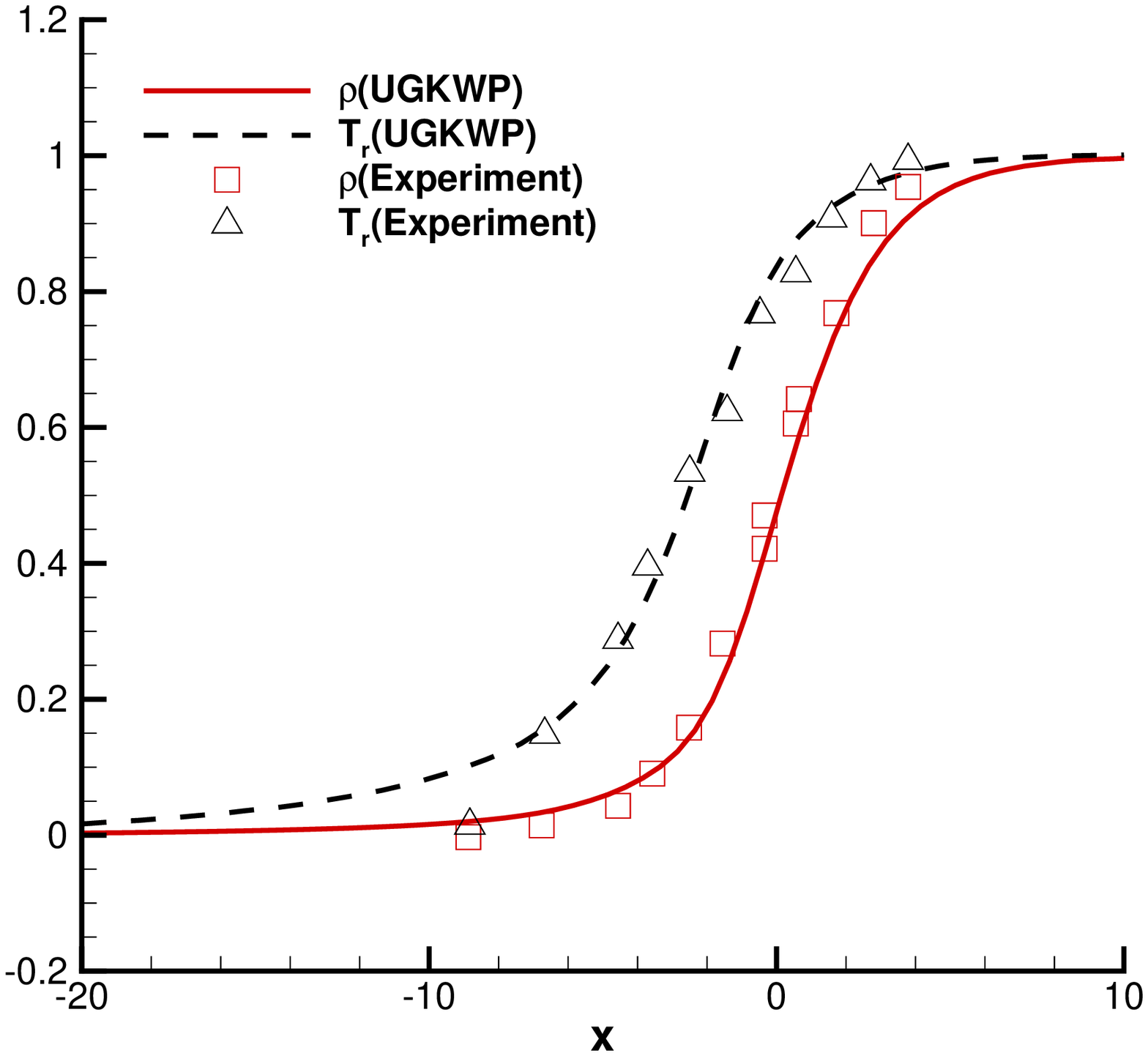}{b}
		\caption{Normalized density and rotational temperature profile of normal shock wave at $\mathrm{M}=7$ (left) and $\mathrm{M}=12.9$ (right). The experiment results are shown in symbol, and the UGKWP solution are shown in line. $x$ is normalized by the mean free path based on the sonic temperature.}
		\label{shock}
	\end{figure}

	\subsection{Planar Fourier flow}
	In this case, we consider the flow driven by the temperature gradient. Consider the nitrogen gas between two parallel plates with a distance $L$. The temperatures at the bottom and top are fixed with values $T_0 = 4/3$ and $T_1 = 2/3$. We set up the simulation as a $1$D problem in the $x$-direction. The computational domain is [0,1] with 20 cells and each cell has a maximum number of 150 particles. The initial density and Mach number of the gas inside the channel are $1$ and $0$. Diffuse boundary conditions are adopted at both plates.
	Fig. \ref{planar} shows the density and translational temperature computed by UGKWP, which are in good agreement with the DSMC results\cite{wu2015kinetic}.
	\begin{figure}
		\centering
		\includegraphics[width=0.48\textwidth]{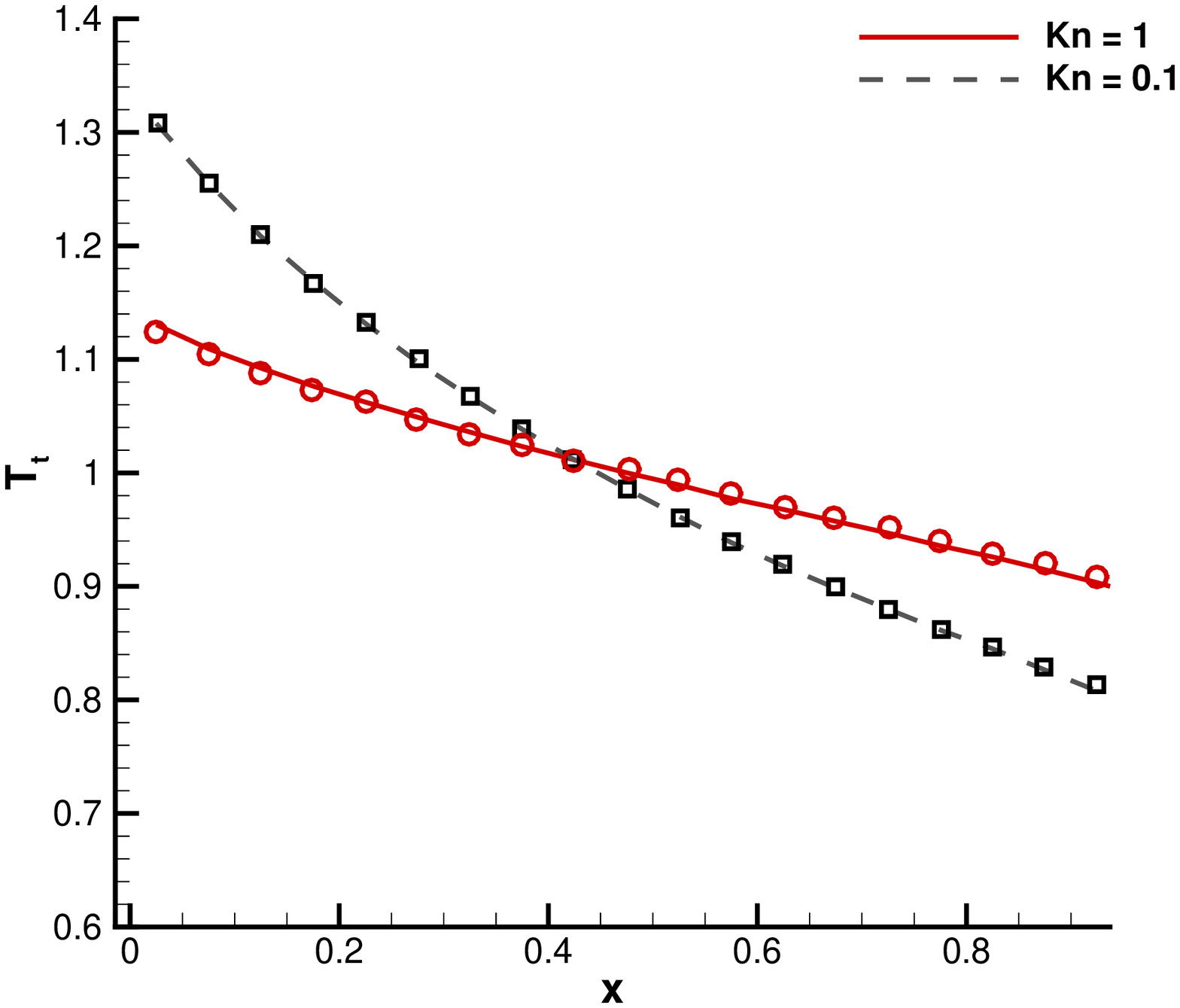}{a}
		\includegraphics[width=0.48\textwidth]{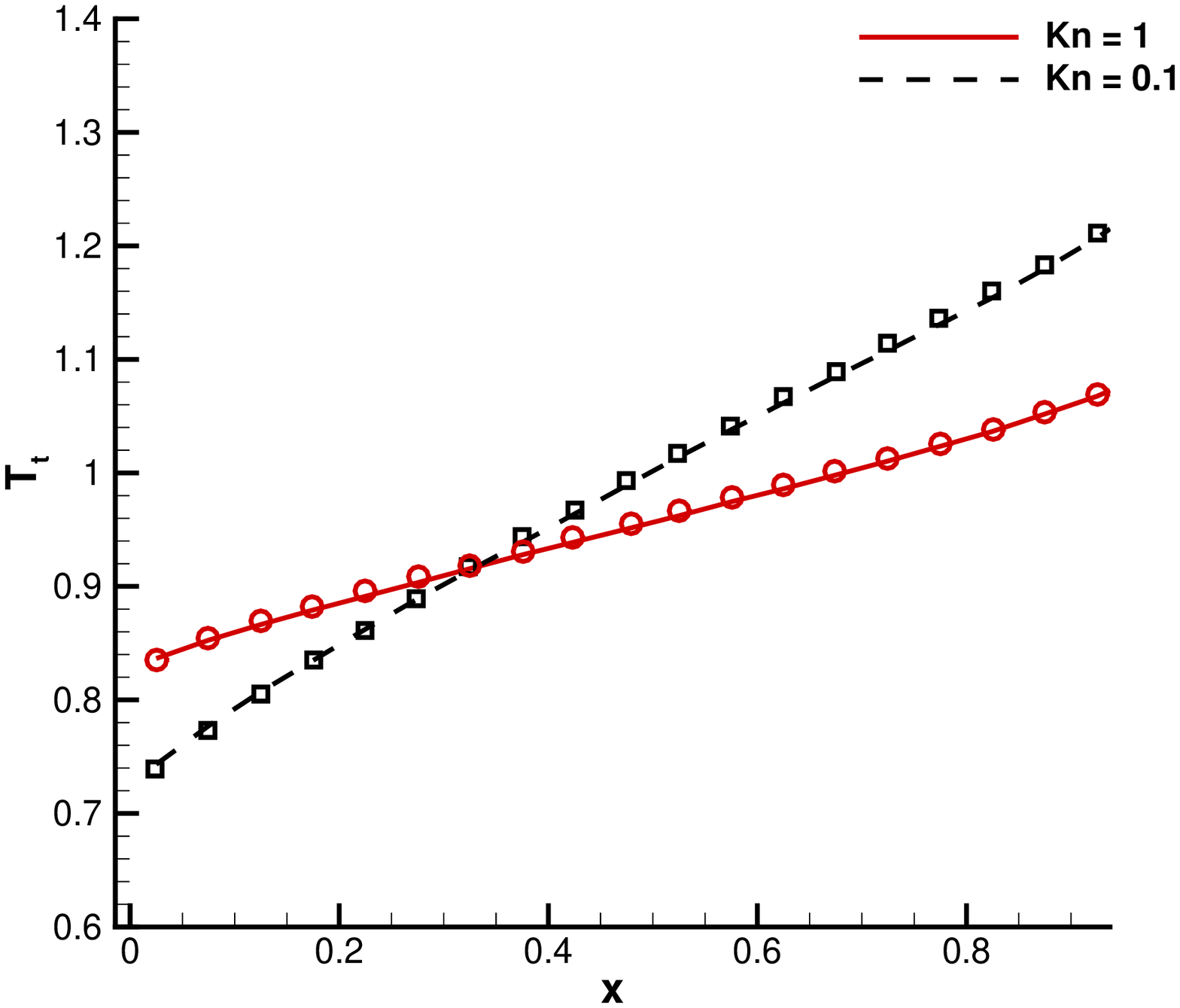}{b}
		\caption{(a)Density and (b)translational temperature profile of planar Fourier flow at $Kn = 0.1$ (squre) and $Kn = 1$ (circle). The DSMC results are shown in symbol, and the UGKWP solution are shown in line.}
		\label{planar}
	\end{figure}
	\subsection{Flow around a blunt circular cylinder}
	Next we calculate the hypersonic nitrogen gas flow passing over a blunt circular cylinder at Mach number $\mathrm{M} = 5.0$ and Knudsen number $\mathrm{Kn}=0.1$. The cylinder has a radius $R = 0.01m$  and the computational domain is divided with $64 \times 150$ cells.
	For nitrogen gas, the molecular number density $n$ is $n=1.2944\times 10^{21}$ /m$^3$. The viscosity coefficient at upstream is $\mu = 1.65788 \times 10^{-5}$ Ns/m$^2$. The cylinder has a surface with constant temperature $T_w = 273$K, and diffusive boundary condition is adopted here. The rotational collision number $Z_{rot}$ is calculated by Eq.(\ref{Zrot}).
 	The dimensionless quantities are used with respect to the reference length as the cylinder radius $L_{ref}=R$, the reference velocity $U_{ref}=\sqrt{2RT_\infty}$, the reference time $t_{ref}=L_{ref}/U_{ref}$, the reference density $\rho_{ref}= \rho_\infty$, and the reference temperature $T_{ref}=T_\infty$.
	The distribution of density, velocity, temperature, and rotational temperature are shown in Fig. \ref{blunt2d}. Fig. \ref{blunt} shows the comparisions between UGKWP results and DSMC solutions \cite{liu2014unified}. Reasonable agreement have been achieved.

	\begin{figure}
		\centering
		\includegraphics[width=0.48\textwidth]{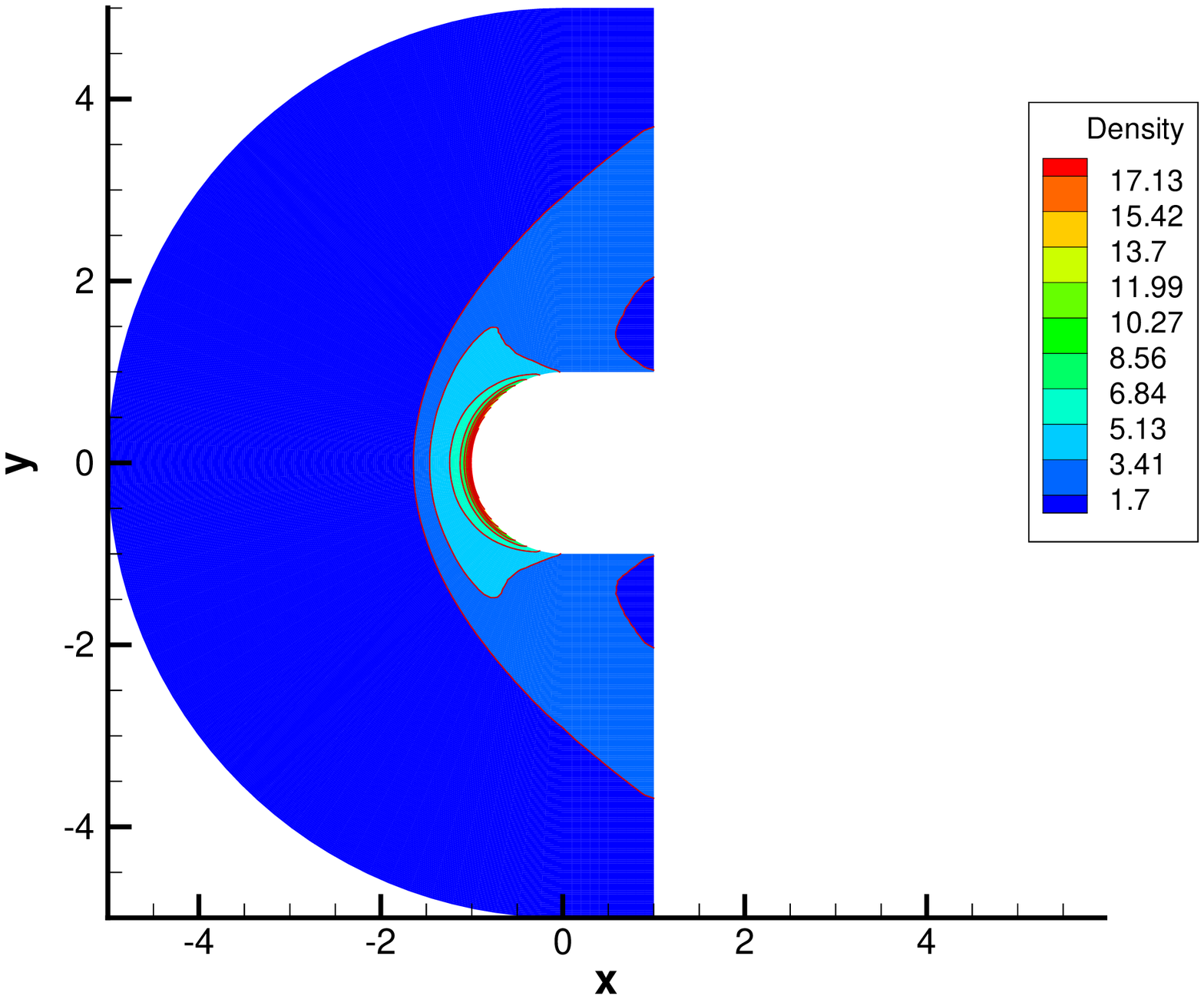}{a}
		\includegraphics[width=0.48\textwidth]{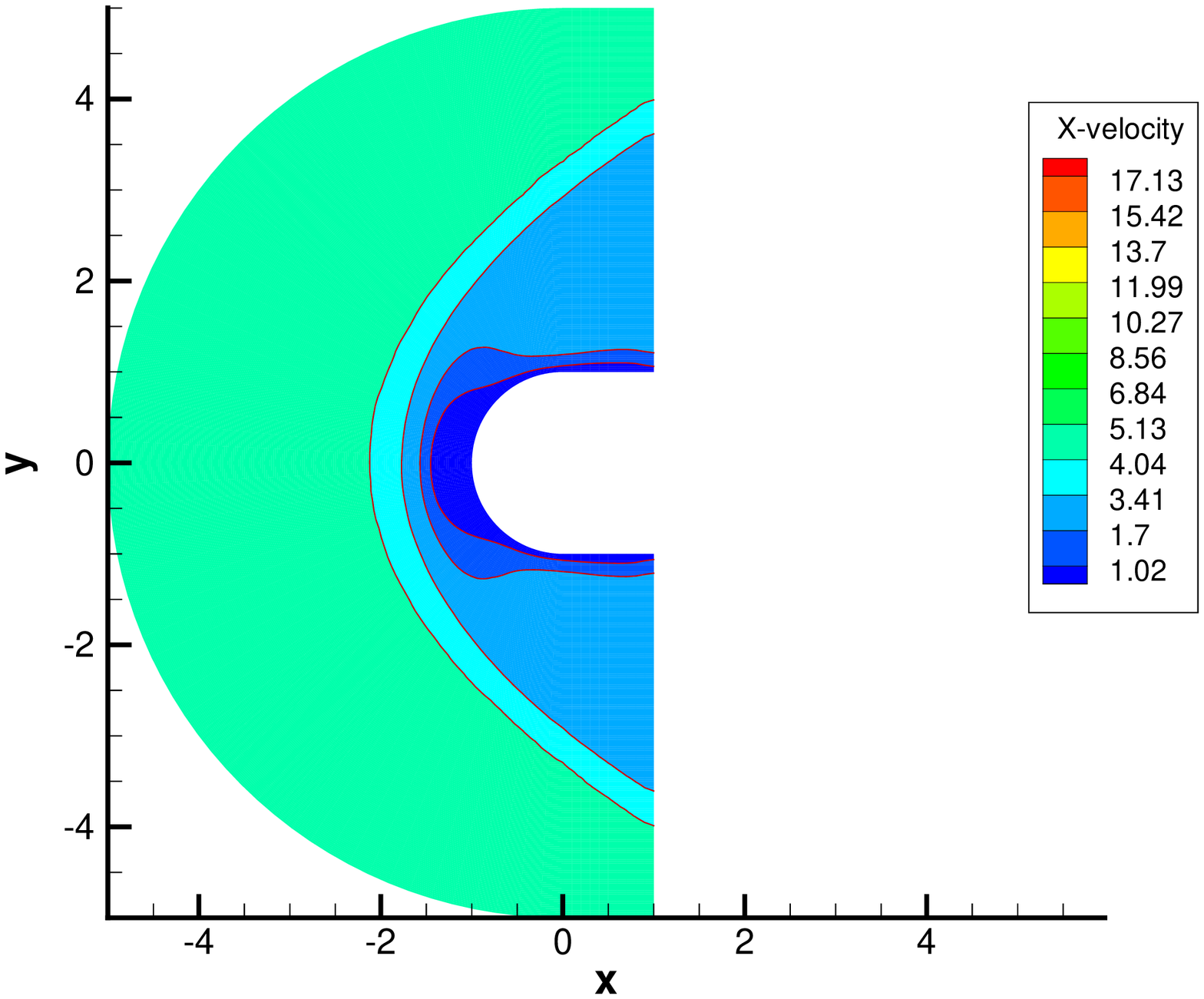}{d}
		\includegraphics[width=0.48\textwidth]{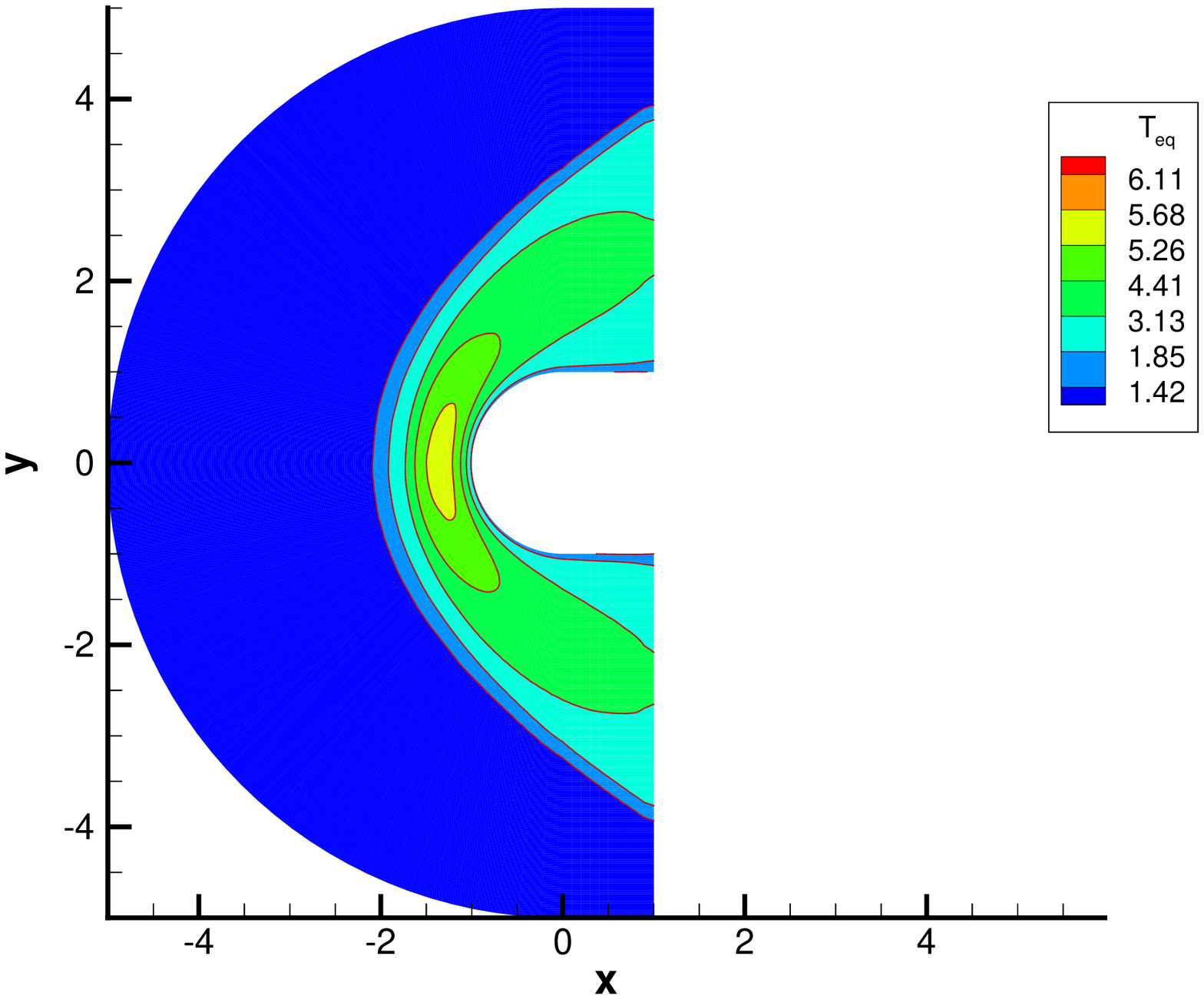}{b}
		\includegraphics[width=0.48\textwidth]{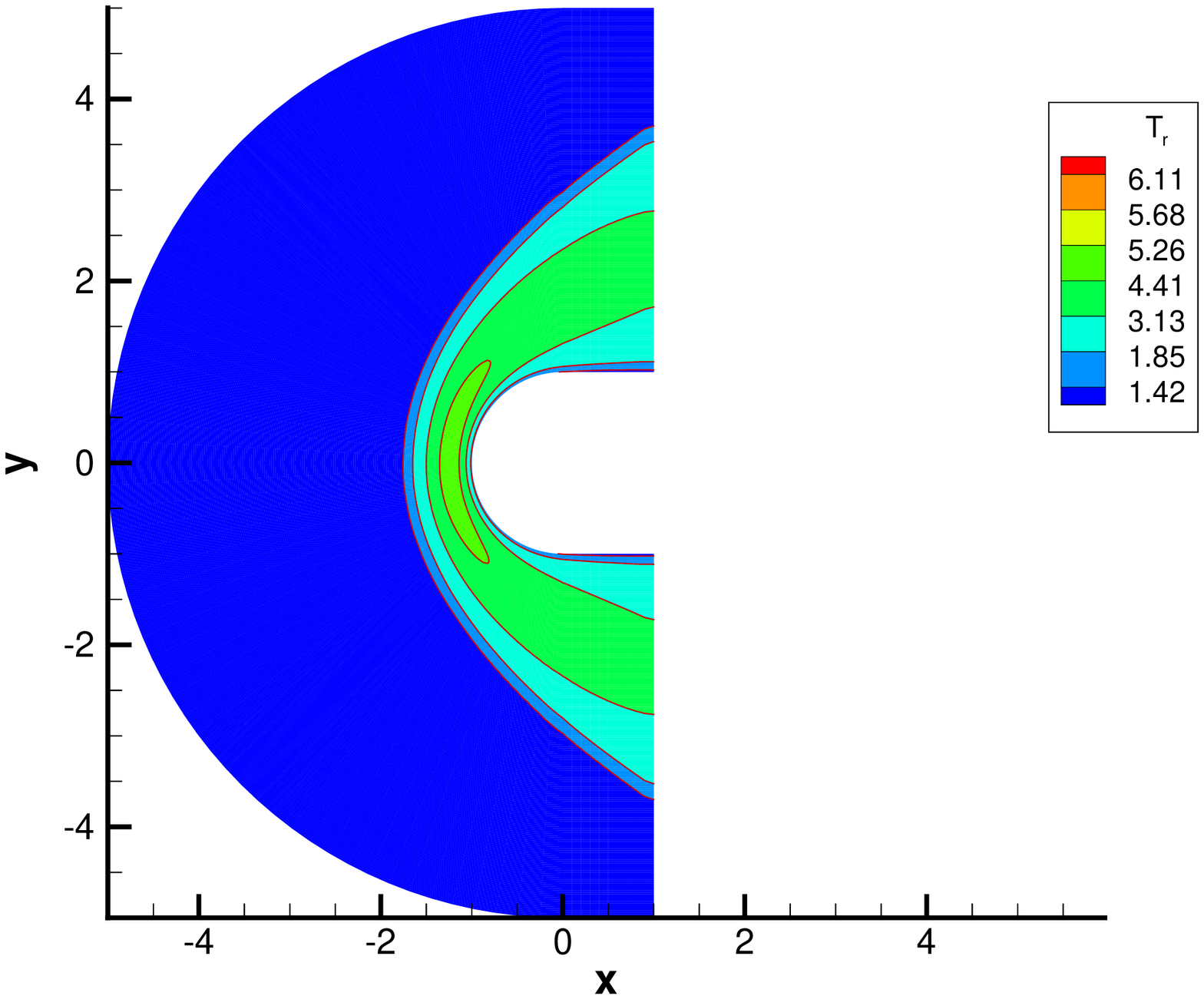}{c}
		\caption{(a)Density and (b)x direction velocity (c)temperature (d) rotational temperature contour for $Kn = 0.1$ and $\mathrm{M} = 5.0$. }
		\label{blunt2d}
	\end{figure}

		\begin{figure}
		\centering
		\includegraphics[width=0.48\textwidth]{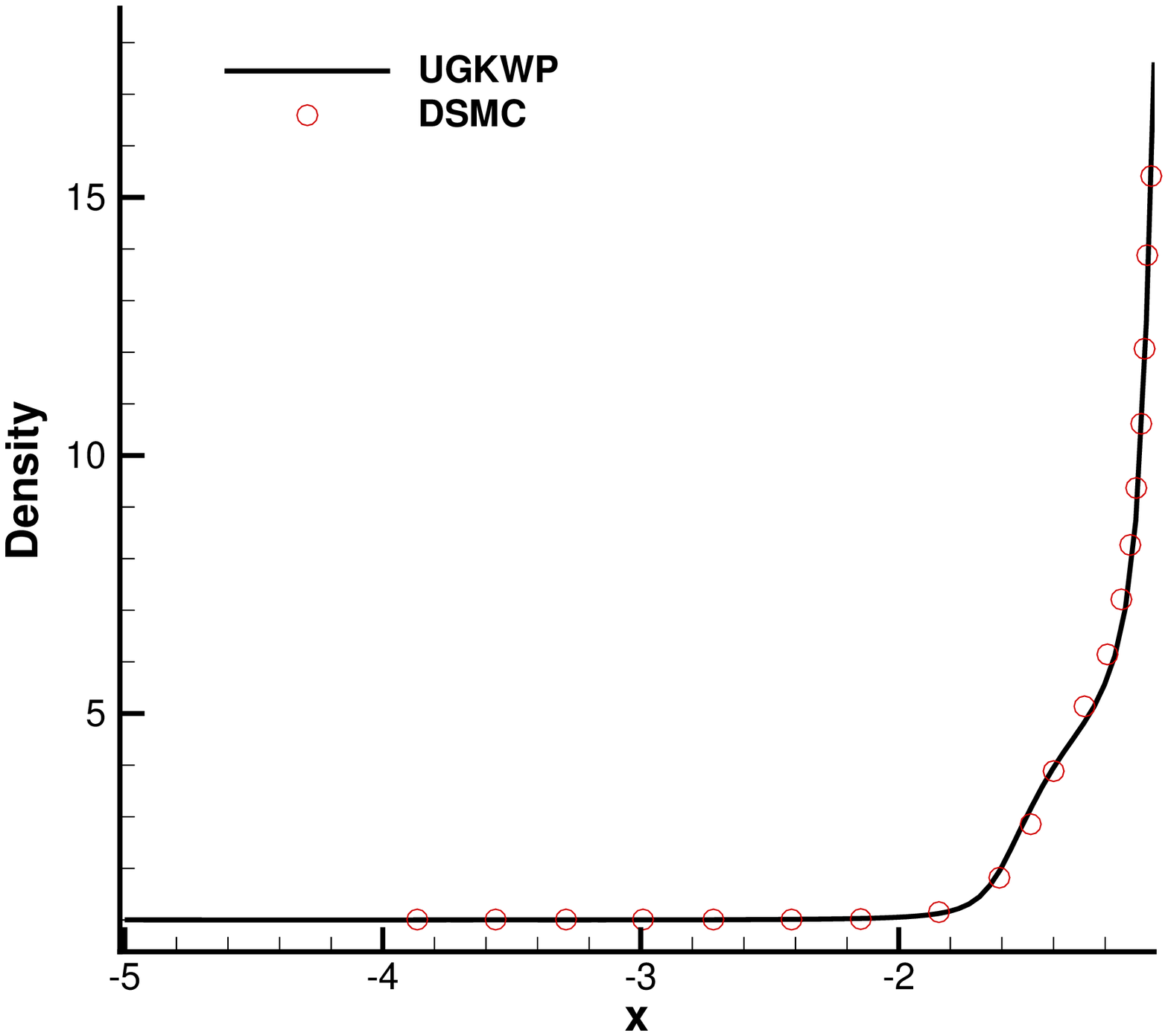}{a}
		\includegraphics[width=0.48\textwidth]{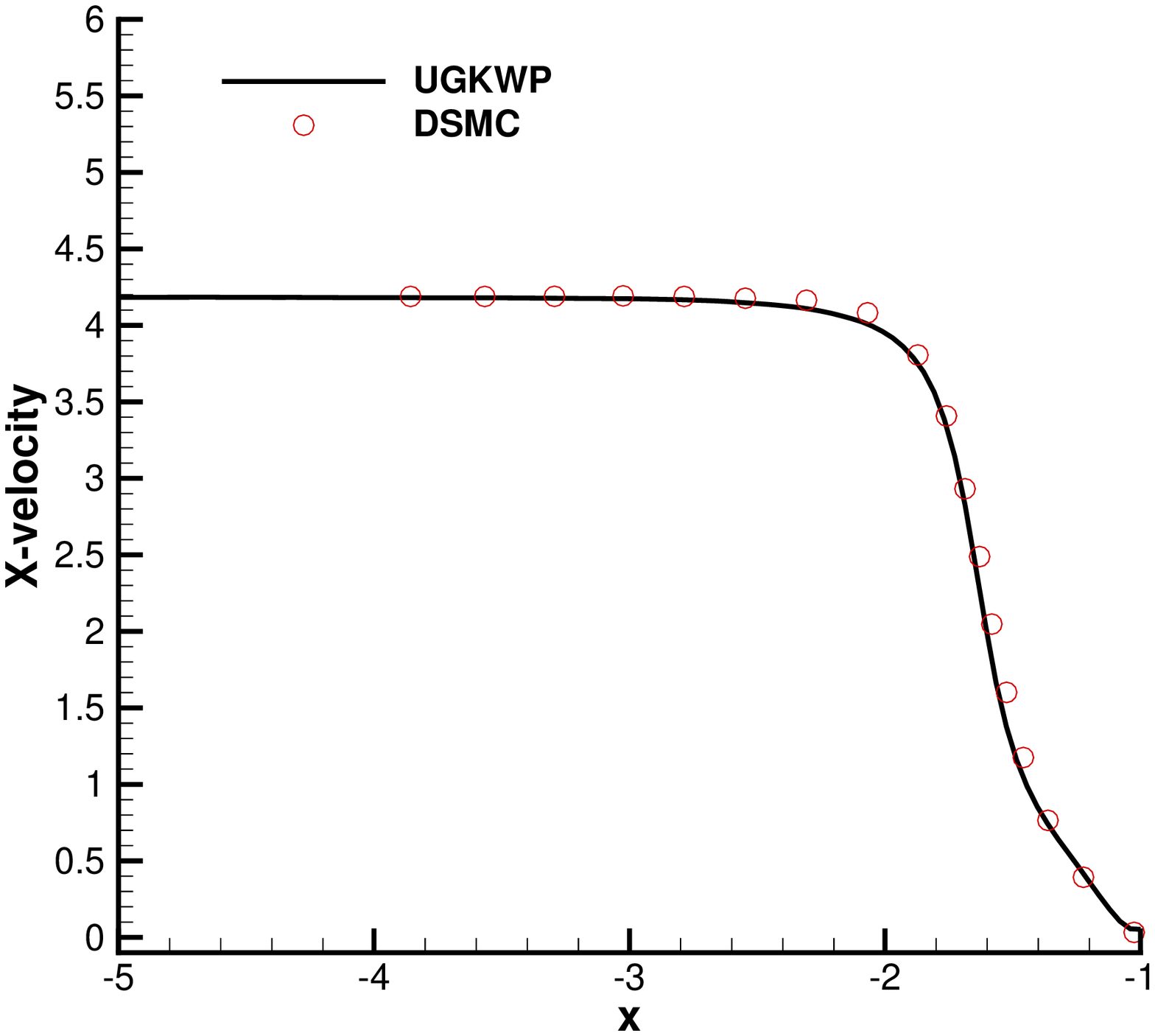}{d}
		\includegraphics[width=0.48\textwidth]{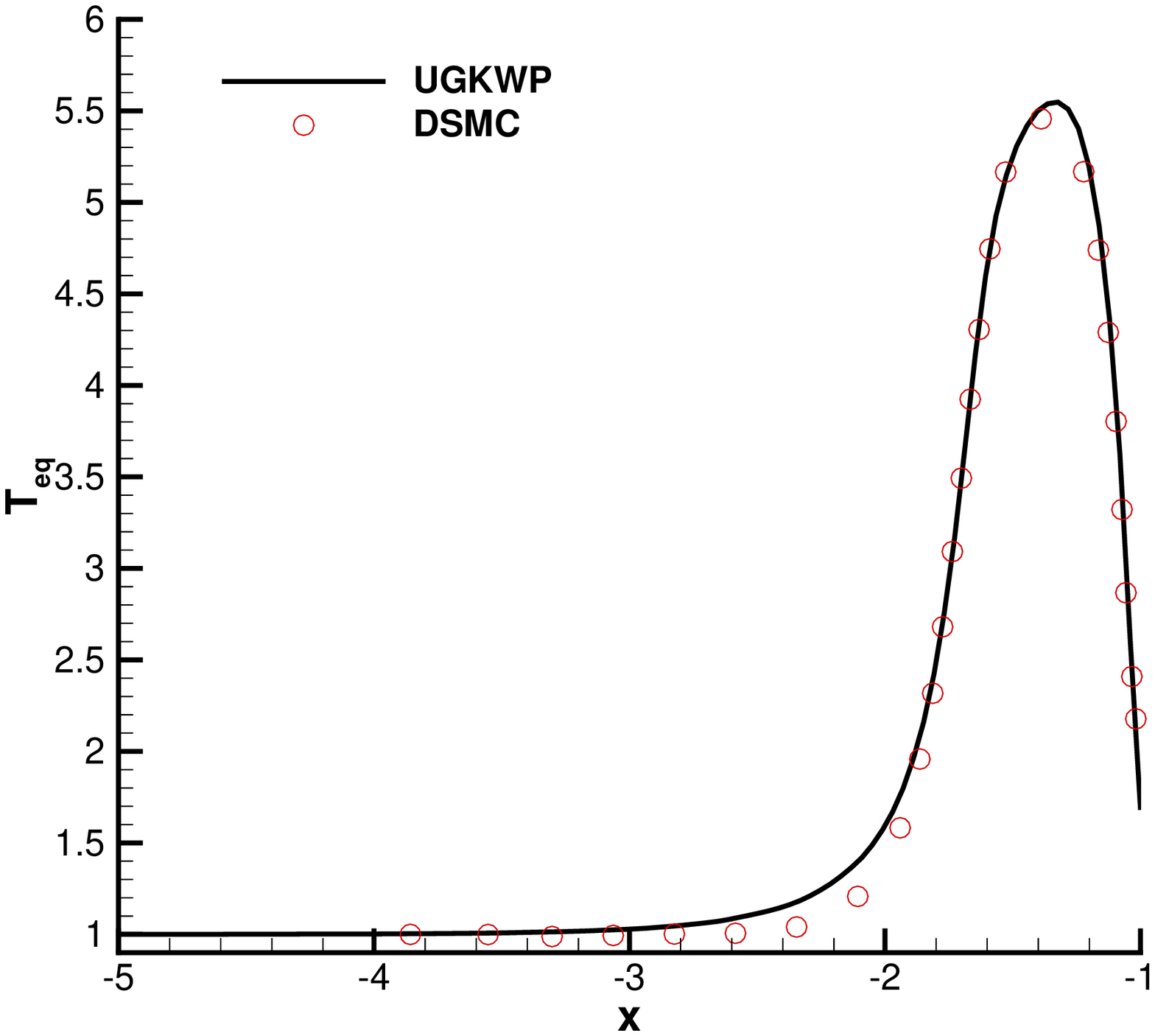}{b}
		\includegraphics[width=0.48\textwidth]{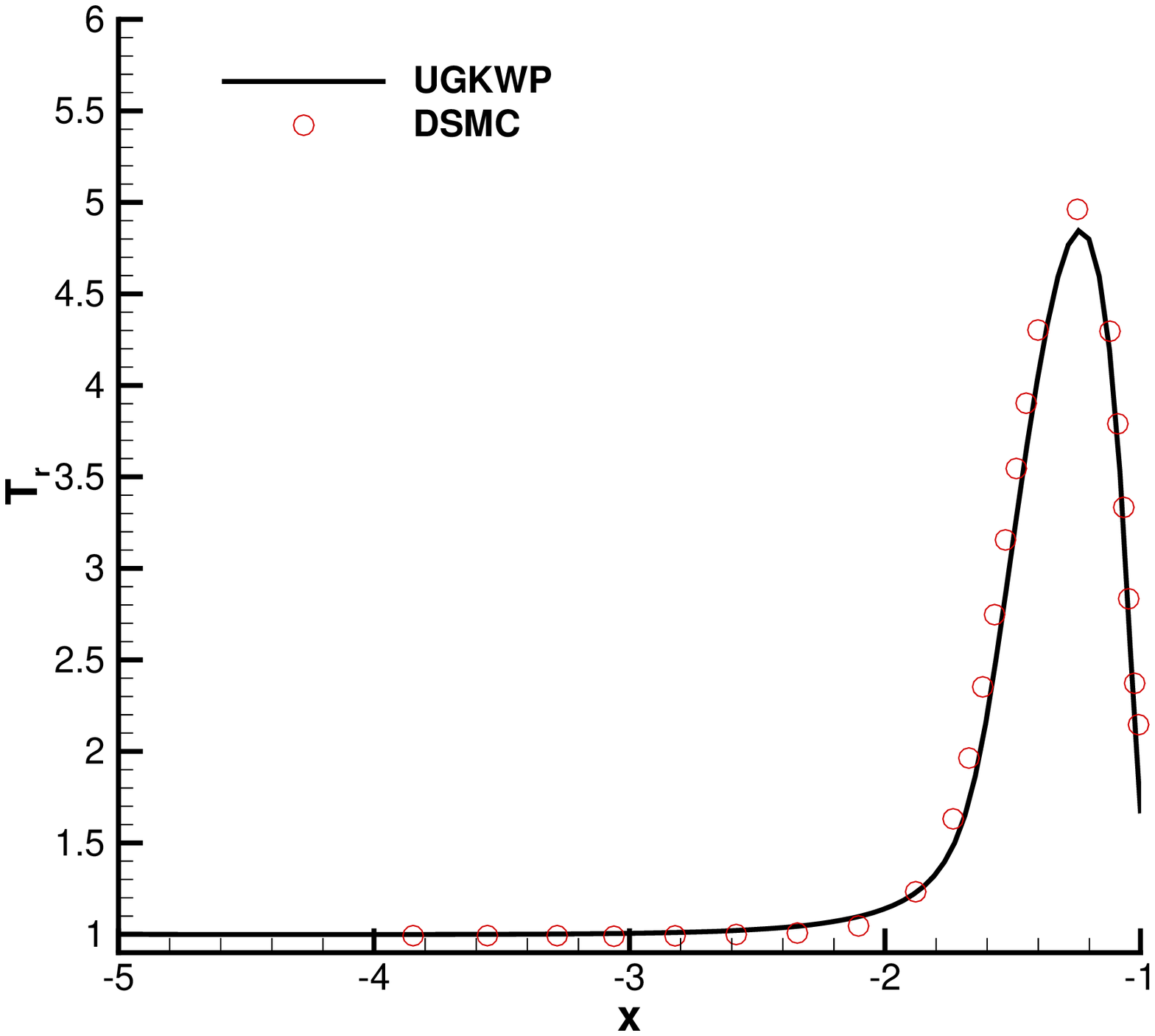}{c}
		\caption{(a) Density, (b) x direction velocity, (c) temperature, (d) rotational temperature profile along stagnation line for $\mathrm{M}=5.0$ and $\mathrm{Kn}=0.1$. The DSMC results are shown in symbol, and the UGKWP solutions are shown in line.}
		\label{blunt}
	\end{figure}
	\subsection{Flow passing a flat plate}
	Following the experiment conducted by Tsuboi and Matsumoto \cite{zhu2019unified}, we simulate the hypersonic rarefied gas flow over a flat plate using UGKWP for nitrogen gas. The case is run 34, where the nozzle exit Mach number is $\mathrm{M} = 4.89$, the nozzle exit pressure is $P_e = 2.12\mathrm{Pa}$, the stagnation pressure is $P_0 = 983\mathrm{Pa}$ and the nozzle exit temperature is $T_e = 116\mathrm{K}$. The stagnation temperature is $T_0 = 670\mathrm{K}$, which is used as a reference temperature to determine the viscosity coefficient,
	\begin{equation*}
	\mu = \mu_{ref} \left( \frac{T_t}{T_0}\right)^{\omega} .
	\end{equation*}
	The reference viscosity is defined as
	\begin{equation*}
	\mu_{ref} = \frac{5\sqrt{2 \pi R T_{ref}}}{16} \frac{\rho_{ref}}{l_{mfp}},
	\end{equation*}
	where $\rho_{ref} = 6.15 \times 10 ^{-5}$kg m $^{-3}$ is the reference density, $l_{mfp} = 0.78mm$ is the mean free path and $T_{ref} = 116$K is the reference temperature. The flat plate has a constant wall temperature of $290$K and the diffusive boundary condition is adopted at the plate. In this case, the relaxation collision number $Z_{rot}$ is set to be $3.5$.
	
	In this study, $59 \times 39$ grid points are used above the plate and $44 \times 25$ grid points are used below the plate, which has the same configuration as that used in UGKS \cite{liu2014unified}. The contours of the density, equilibrium temperature, rotational temperature and translational temperature are shown in Fig.\ref{flatplate2d}. The temperature distribution along the vertical line above the flat plate at $x = 5$mm and $x = 20$mm are shown in Fig. \ref{flatplatecomp}, which show good agreement with the experiment measurements.
	\begin{figure}
		\centering
		\includegraphics[width=0.48\textwidth]{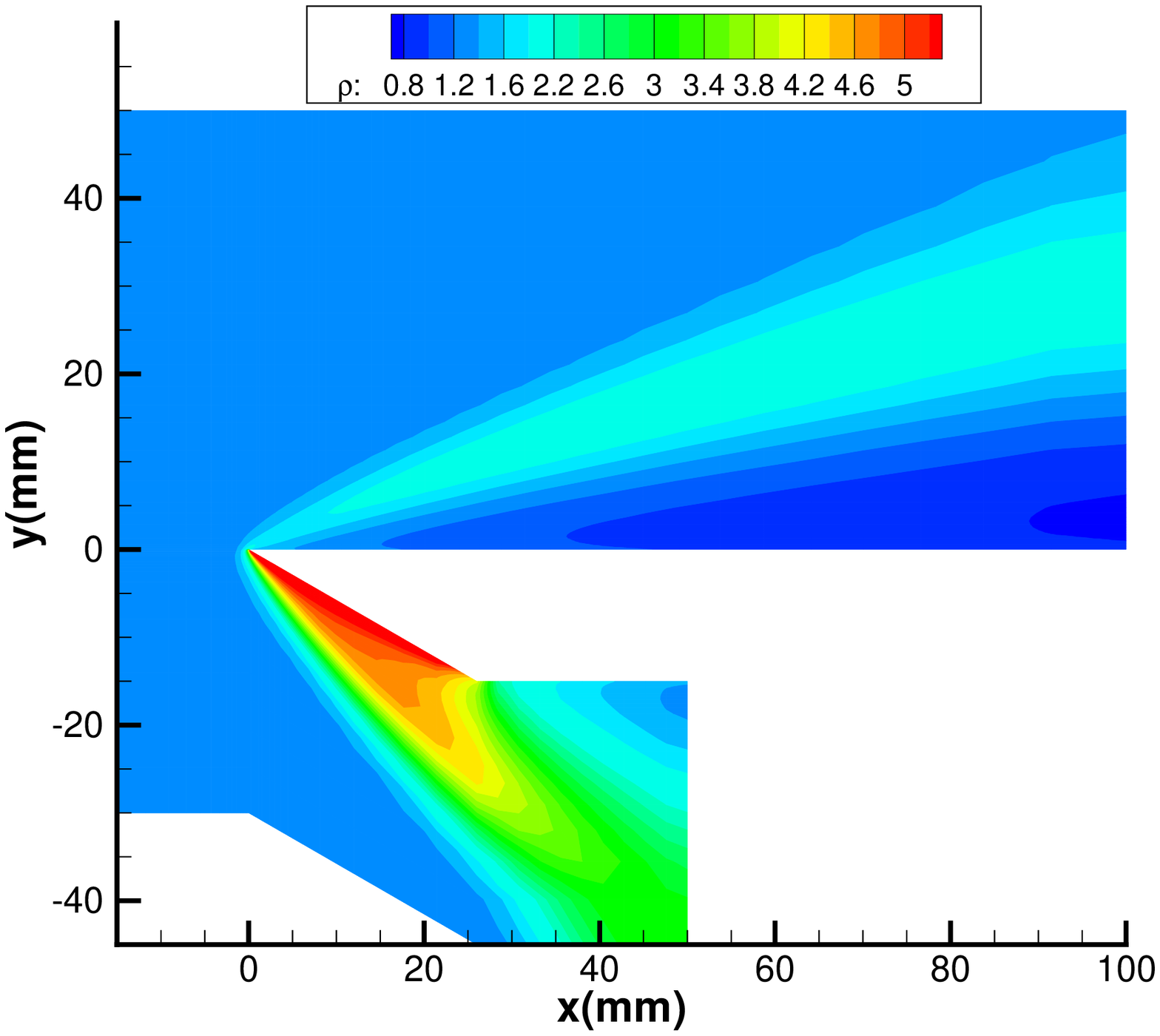}{a}
		\includegraphics[width=0.48\textwidth]{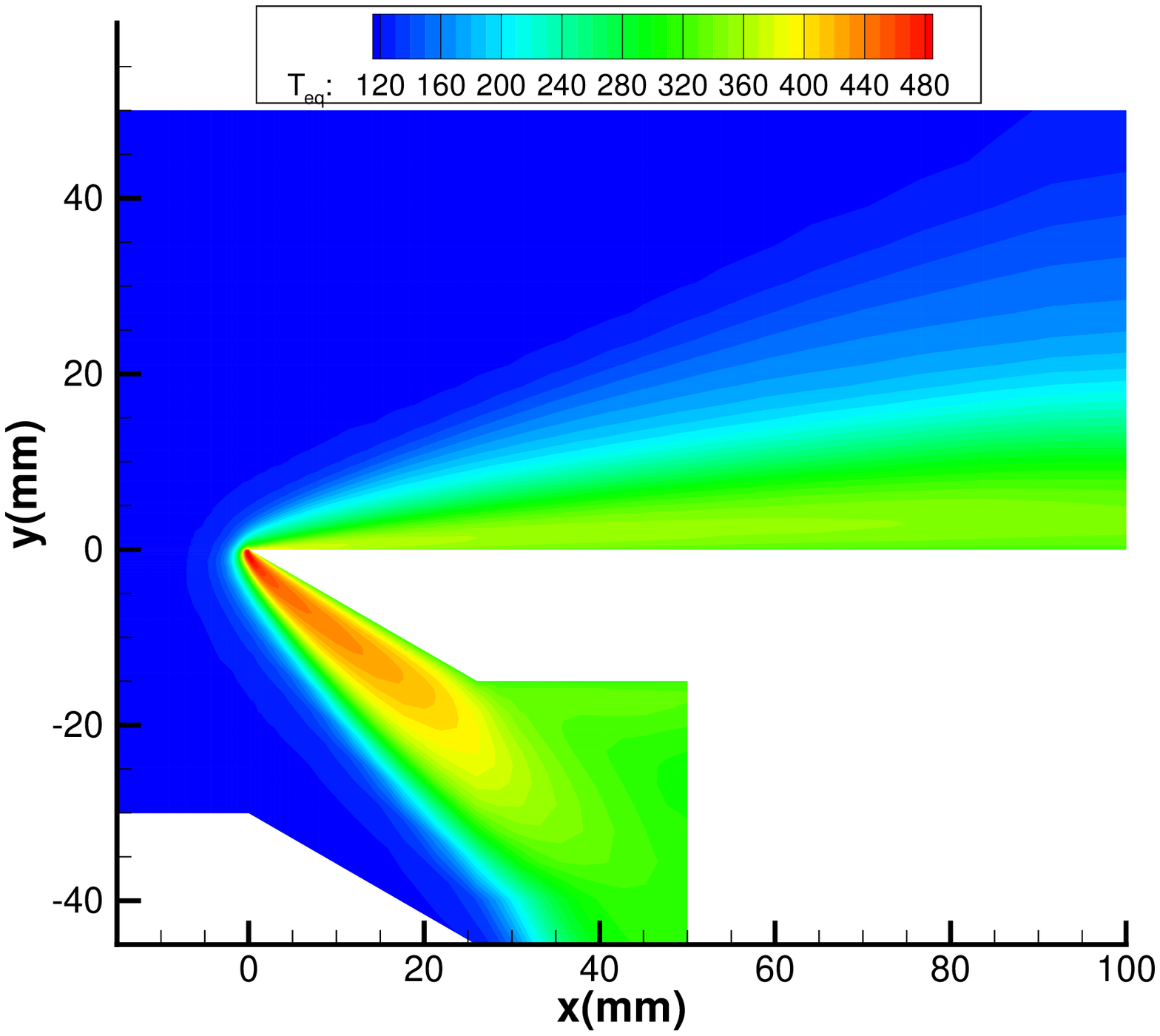}{b}
		\includegraphics[width=0.48\textwidth]{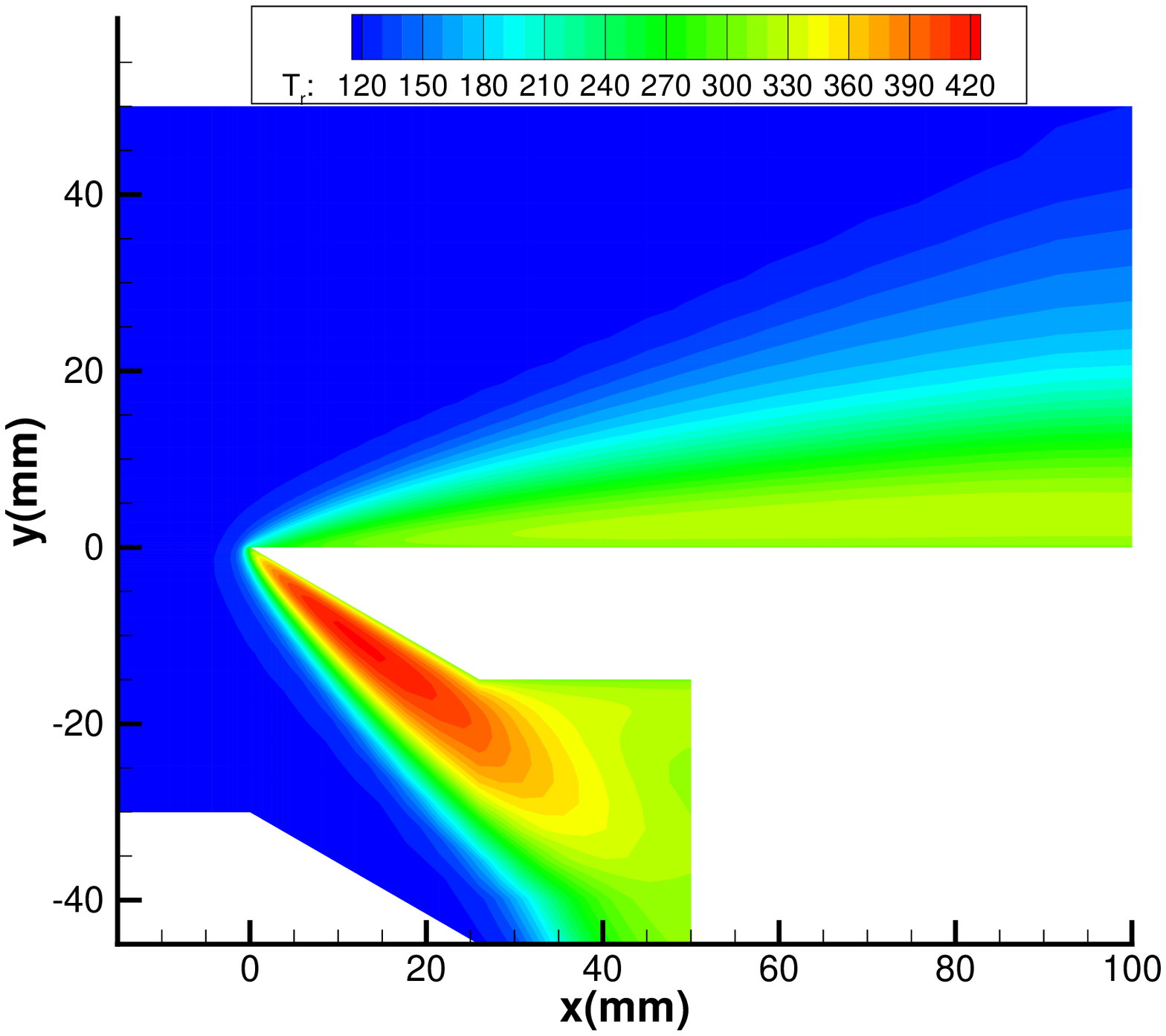}{c}
		\includegraphics[width=0.48\textwidth]{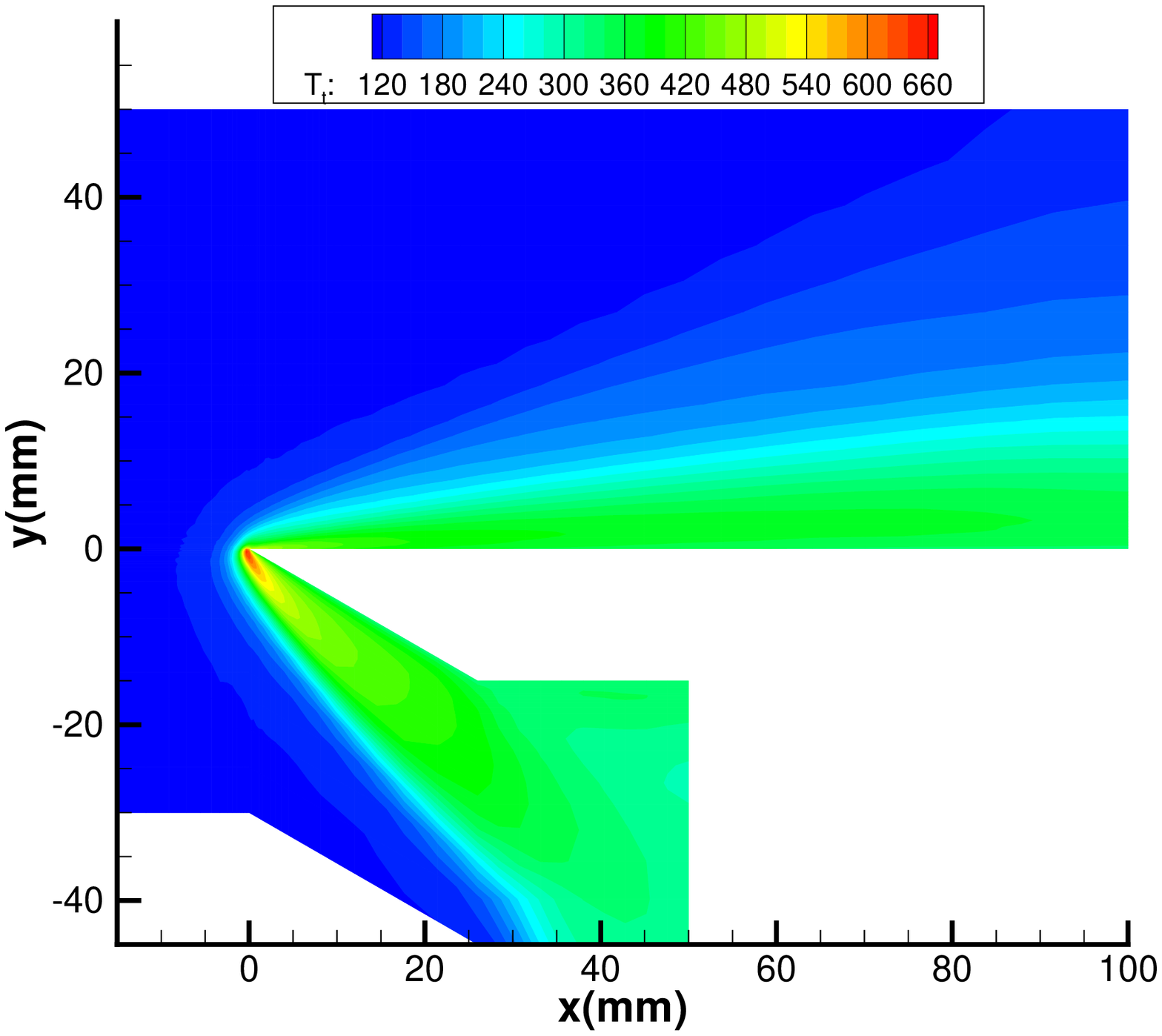}{d}
		\caption{(a)Density and (b)temperature (c)rotational temperature (d) translational temperature contour for the hypersonic flow passing a flat plate}
		\label{flatplate2d}
	\end{figure}
	\begin{figure}
		\centering
		\includegraphics[width=0.48\textwidth]{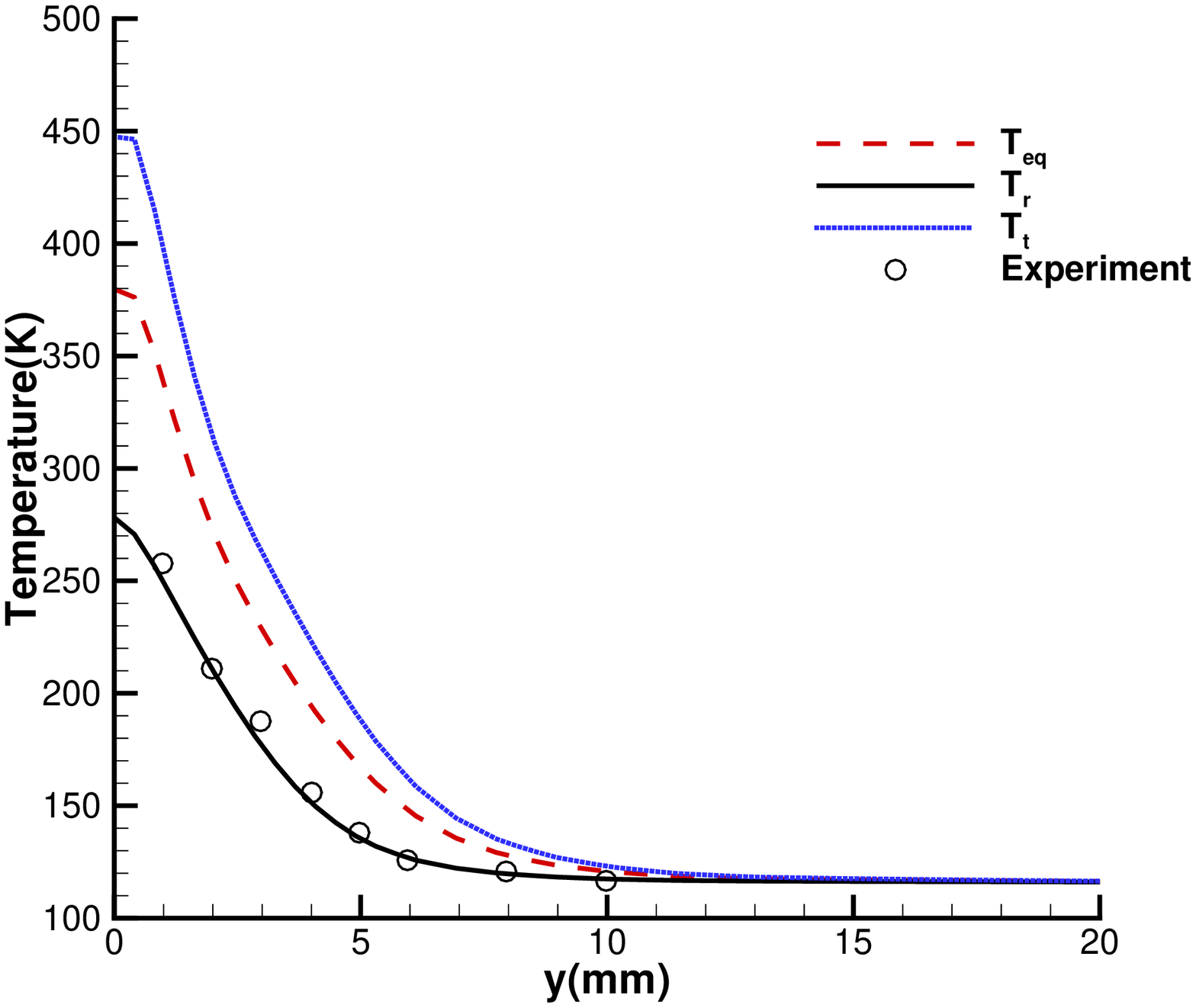}{a}
		\includegraphics[width=0.48\textwidth]{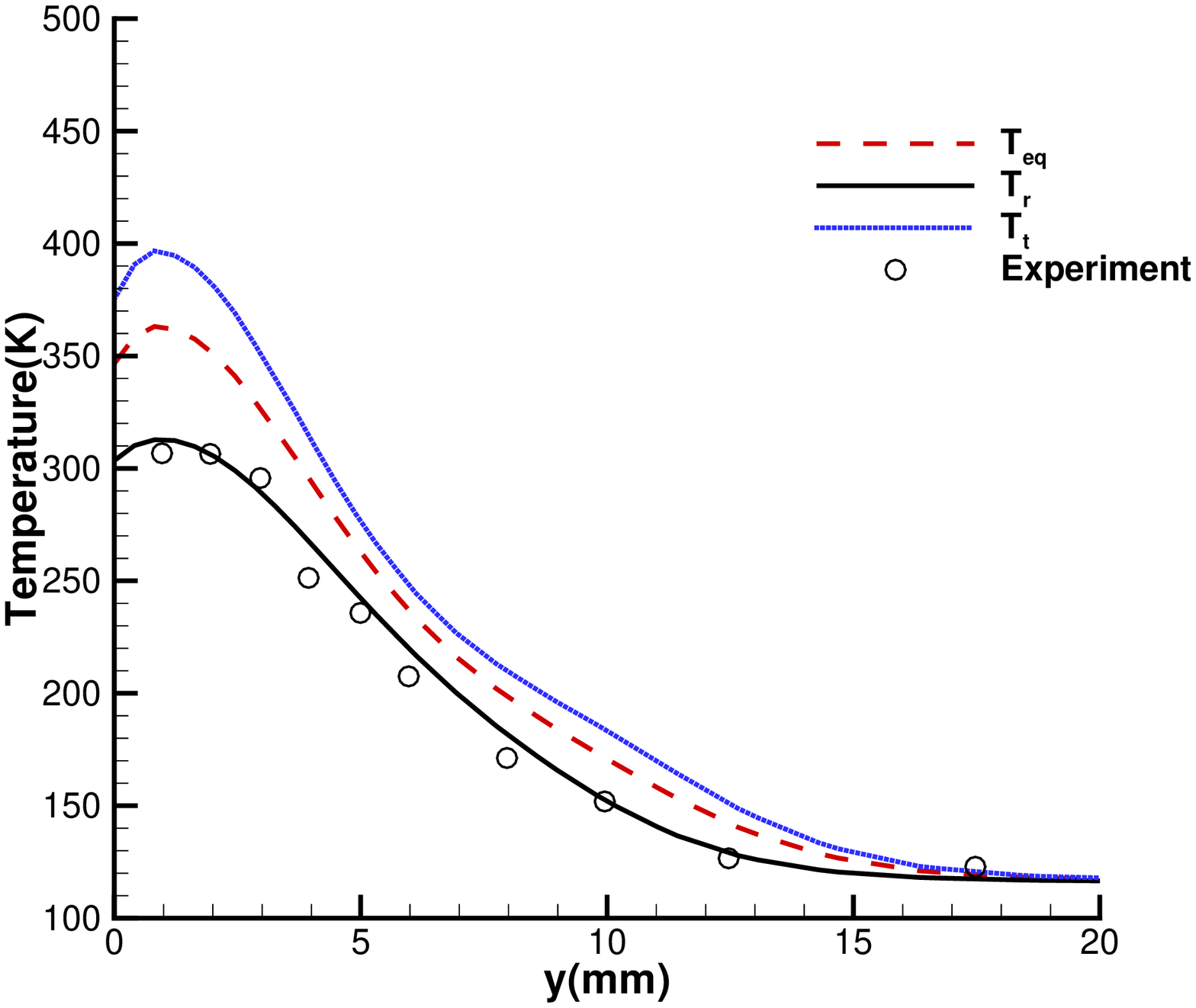}{b}
		\caption{Temperature profiles along vertical lines at (a) x = 5mm and (b) x = 20 mm. The experiment results are shown in symbol, and the UGKWP solutions are shown in line}
		\label{flatplatecomp}
	\end{figure}

	\subsection{Flow passing a sphere}
	The three dimensional case is about $\mathrm{M} = 4.25$ nitrogen gas flow passing through a sphere at $\mathrm{Kn}=0.031$ and $\mathrm{Kn}=0.121$ in the transition regime. The radius of sphere is $10^{-3}$m and the surface mesh of the sphere is divided into $6$ blocks with $16\times16$ mesh points in each block with a minimum surface spacing $6.255\time 10^{-5}$m. Diffusive wall boundary condition with a constant temperature $T_w = 302$ K is imposed on the surface. The computational domain is composed of $29700$ hexahedra with growth rate $1.1$ and smallest cell height $5 \times 10^{-5}$ m. The inflow is diatomic nitrogen gas with molecular mass
	$m = 4.65 \times 10^{-26} $ kg and diameter $d = 4.17 \times 10^{-10}$ m. The upstream flow temperature is set to be $T_{\infty} = 65$ K. The reference viscosity is given
	by the variable hard sphere (VHS) model with $\omega = 0.74$.
	For the case of $\mathrm{Kn}=0.031$, the time-averaging starts from $2500$ steps and continues for $13000$ steps with an initial field computed
	by 1000 steps GKS. The total computation takes $13500$ time steps, and runs $3.1$h on a workstation with (Dual CPU) Intel
	Xeon Platinum $8168$ at $2.70$ GHz with $48$ cores, which is around $85$ times faster than implicit-UGKS \cite{jiang2019implicit}. The distribution of density, velocity, temperature, and rotational temperature are shown in Fig. \ref{sphere}.
	For the case of $\mathrm{Kn}=0.121$, the time-averaging starts from $2500$ steps and continues for $17000$ steps with an initial field computed
	by 1000 steps GKS. The total computational time steps are $19500$ and the computation runs $32.1$h on the same work station with $48$ cores, which is around $10$ times faster than implicit-UGKS. The distribution of density, velocity, temperature, and rotational temperature are shown in Fig. \ref{sphere2}.
	Fig. \ref{sphereerror} shows the relative error with drag coefficient (Air) given by experiment\cite{wendtJF}.
	Apart from the above two cases, we also compute a hypersonic  case with $\mathrm{M} = 10$ and $\mathrm{Kn}=0.01$.
The distribution of density, velocity, temperature, and rotational temperature are shown in Fig. \ref{sphere3}.
In comparison with the previous supersonic Mach number $4.25$ cases, the computation for Mach $10$ case needs only $1.46$ hour with a personal $48$ cores workstation.
	The comparison of the drag coefficients between UGKWP results and the experiment measurements  are shown in Table. \ref{compDragError}.
	\begin{figure}
		\centering
		\includegraphics[width=0.48\textwidth]{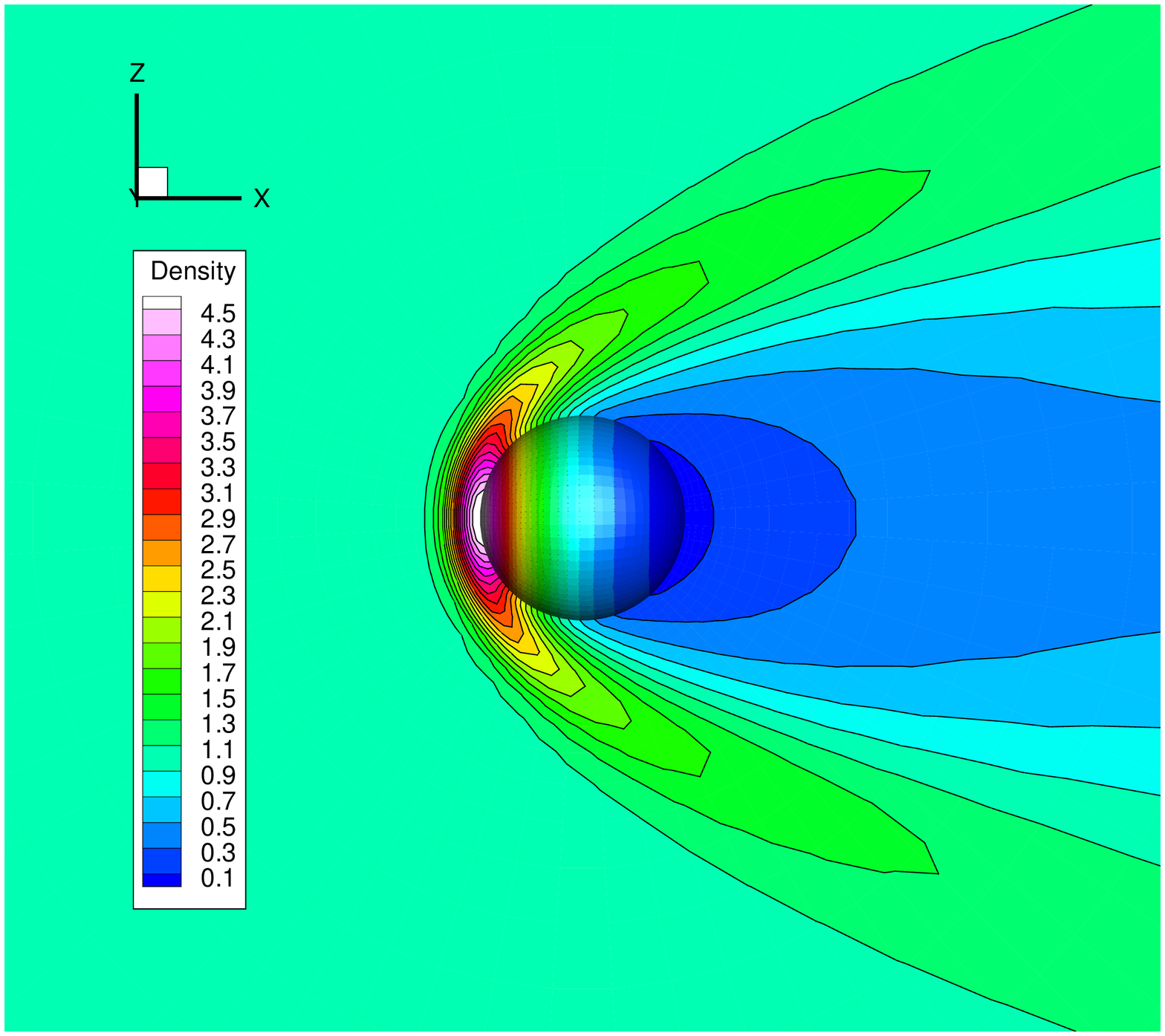}{a}
		\includegraphics[width=0.48\textwidth]{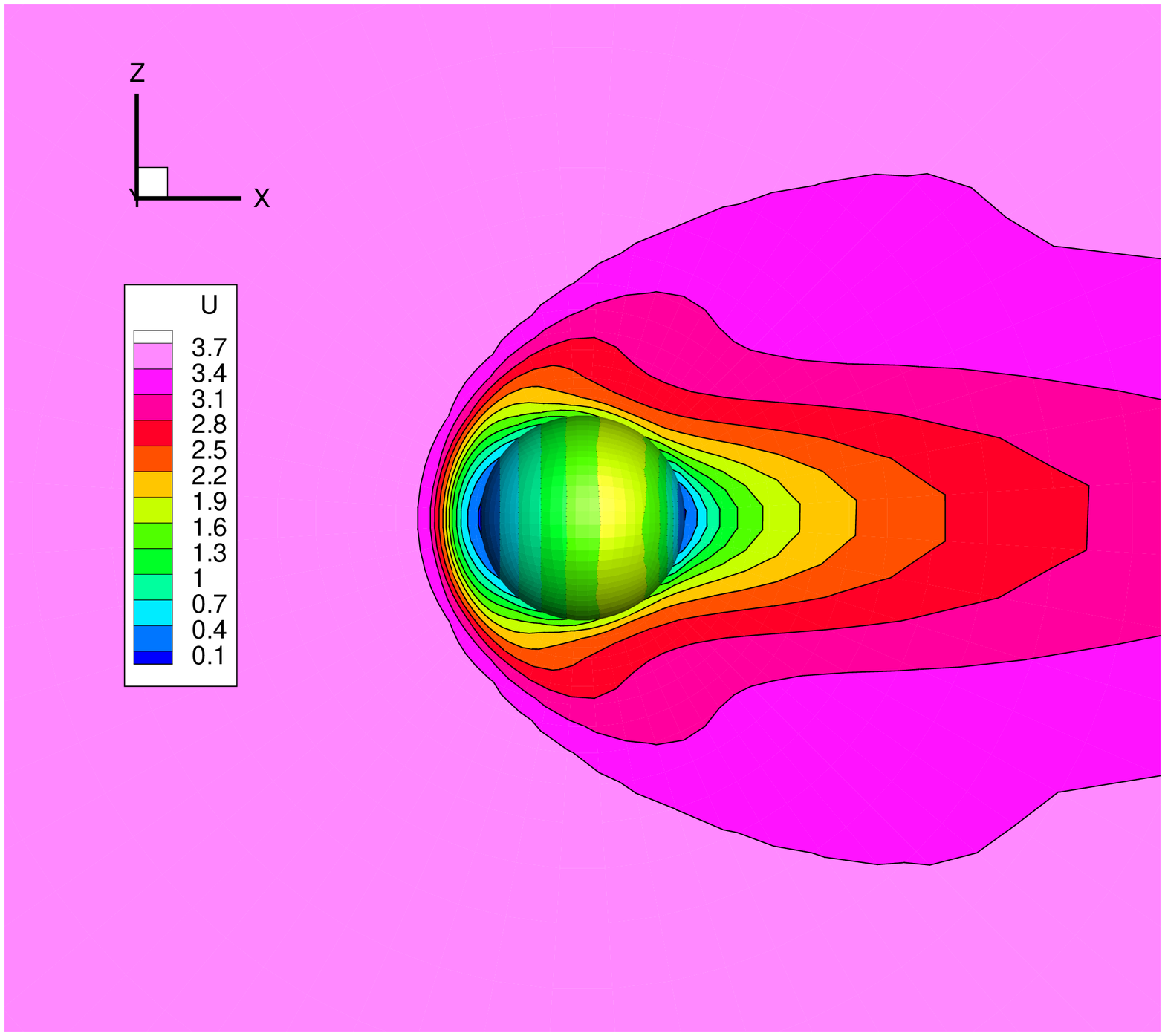}{b}
		\includegraphics[width=0.48\textwidth]{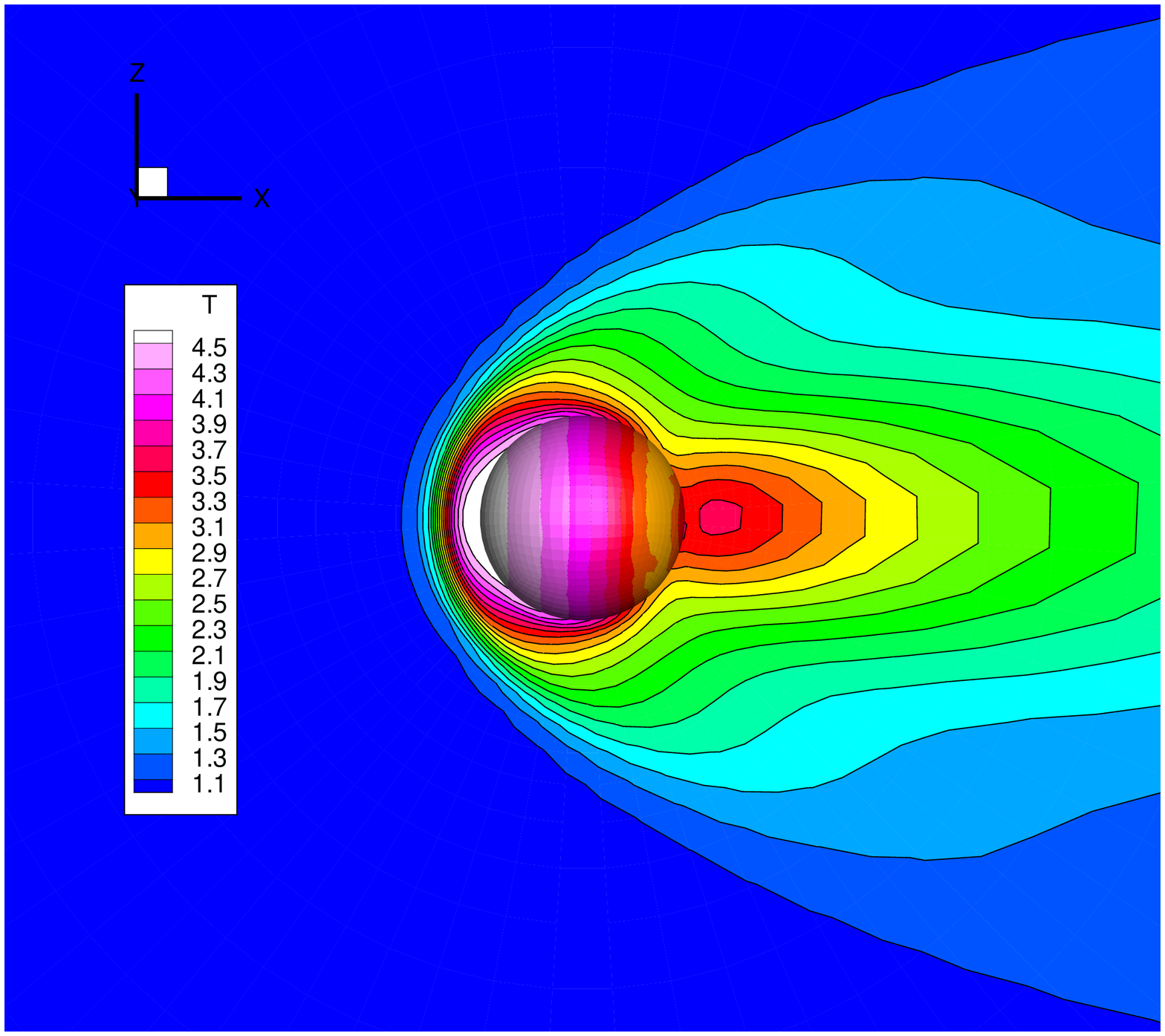}{c}
		\includegraphics[width=0.48\textwidth]{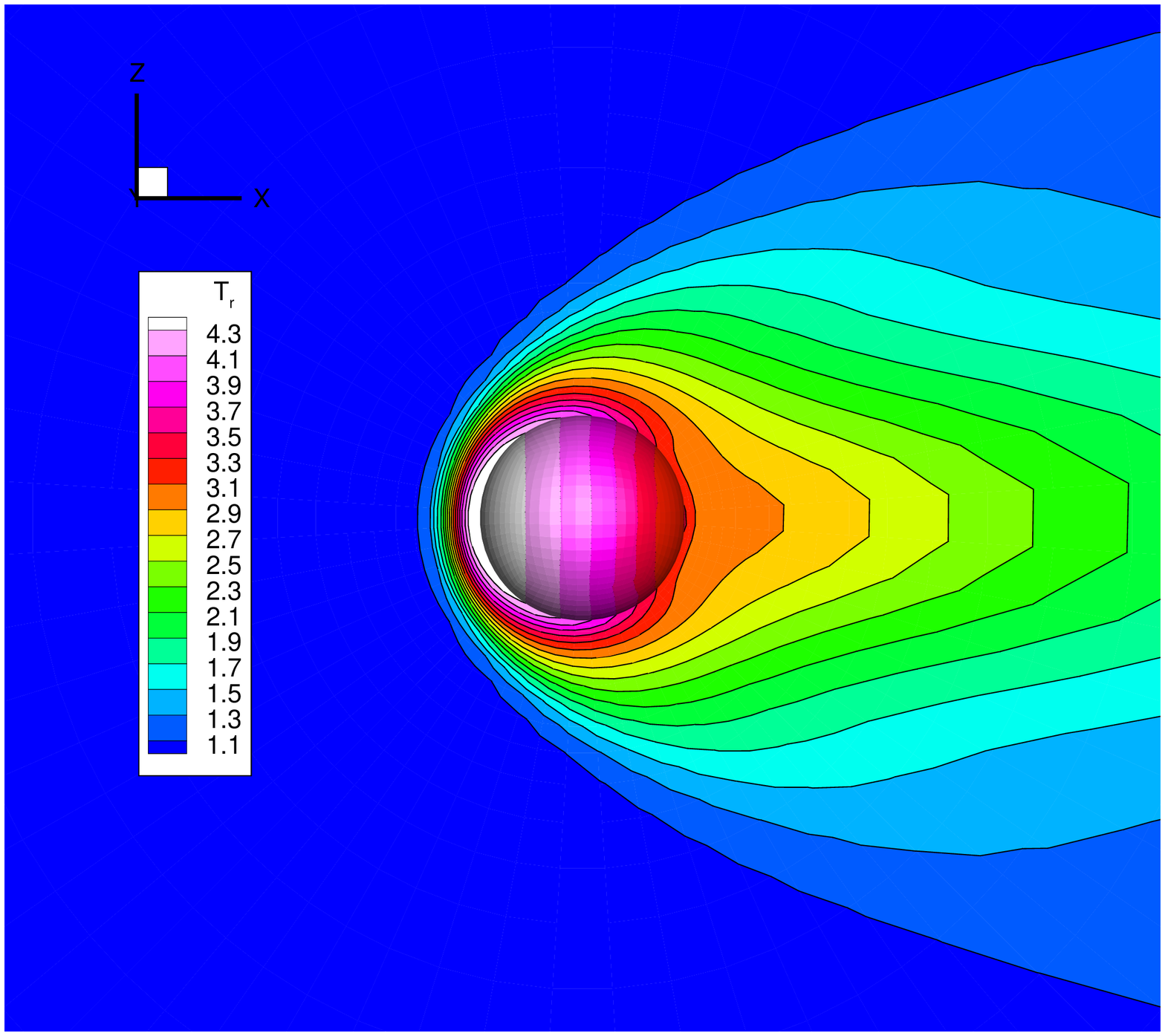}{d}
		\caption{(a)Density and (b)x direction velocity (c)temperature (d) rotational temperature contour for $Kn = 0.031$ and $\mathrm{M} = 4.25$. }
		\label{sphere}
	\end{figure}
	\begin{figure}
		\centering
		\includegraphics[width=0.48\textwidth]{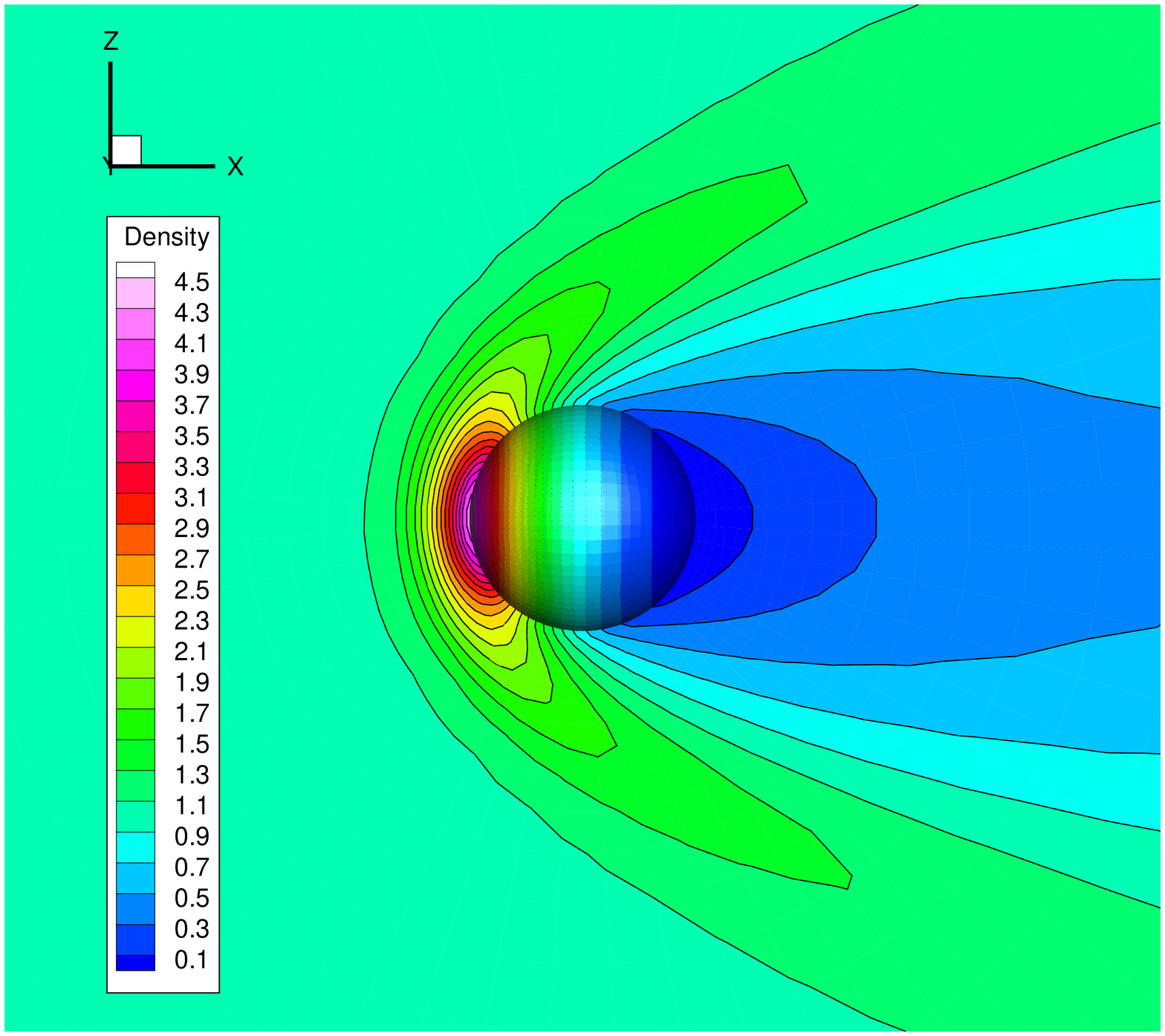}{a}
		\includegraphics[width=0.48\textwidth]{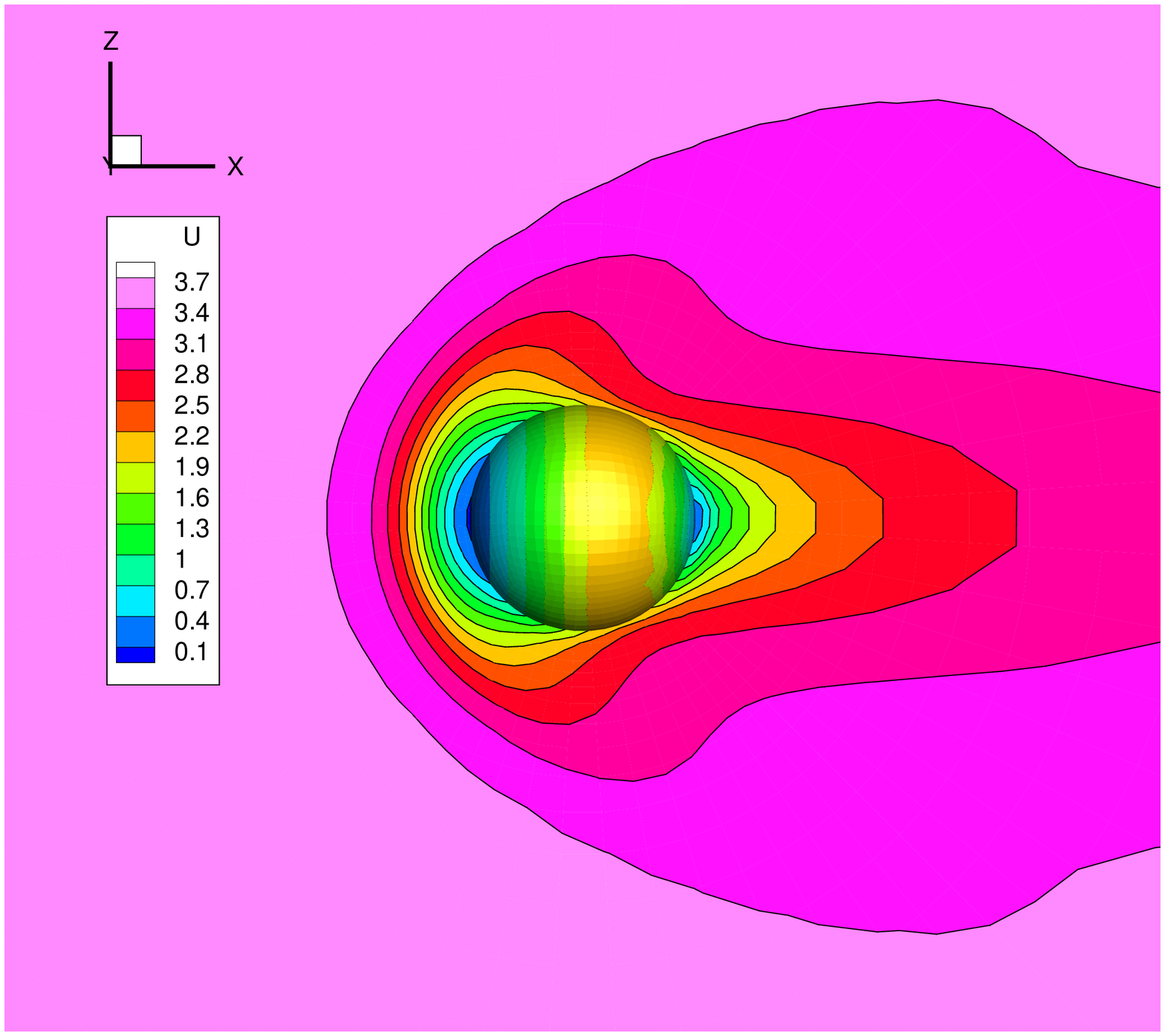}{b}
		\includegraphics[width=0.48\textwidth]{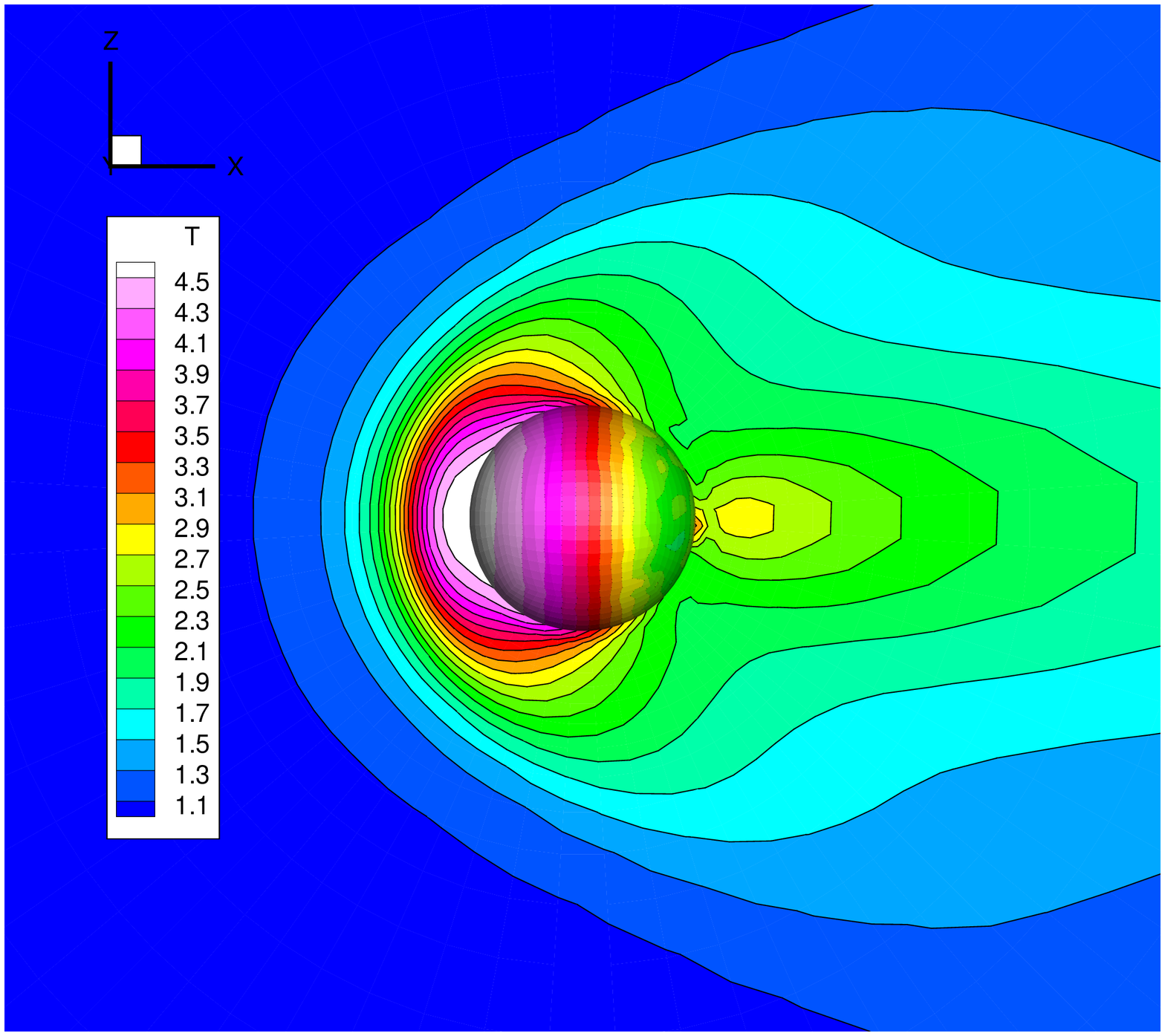}{c}
		\includegraphics[width=0.48\textwidth]{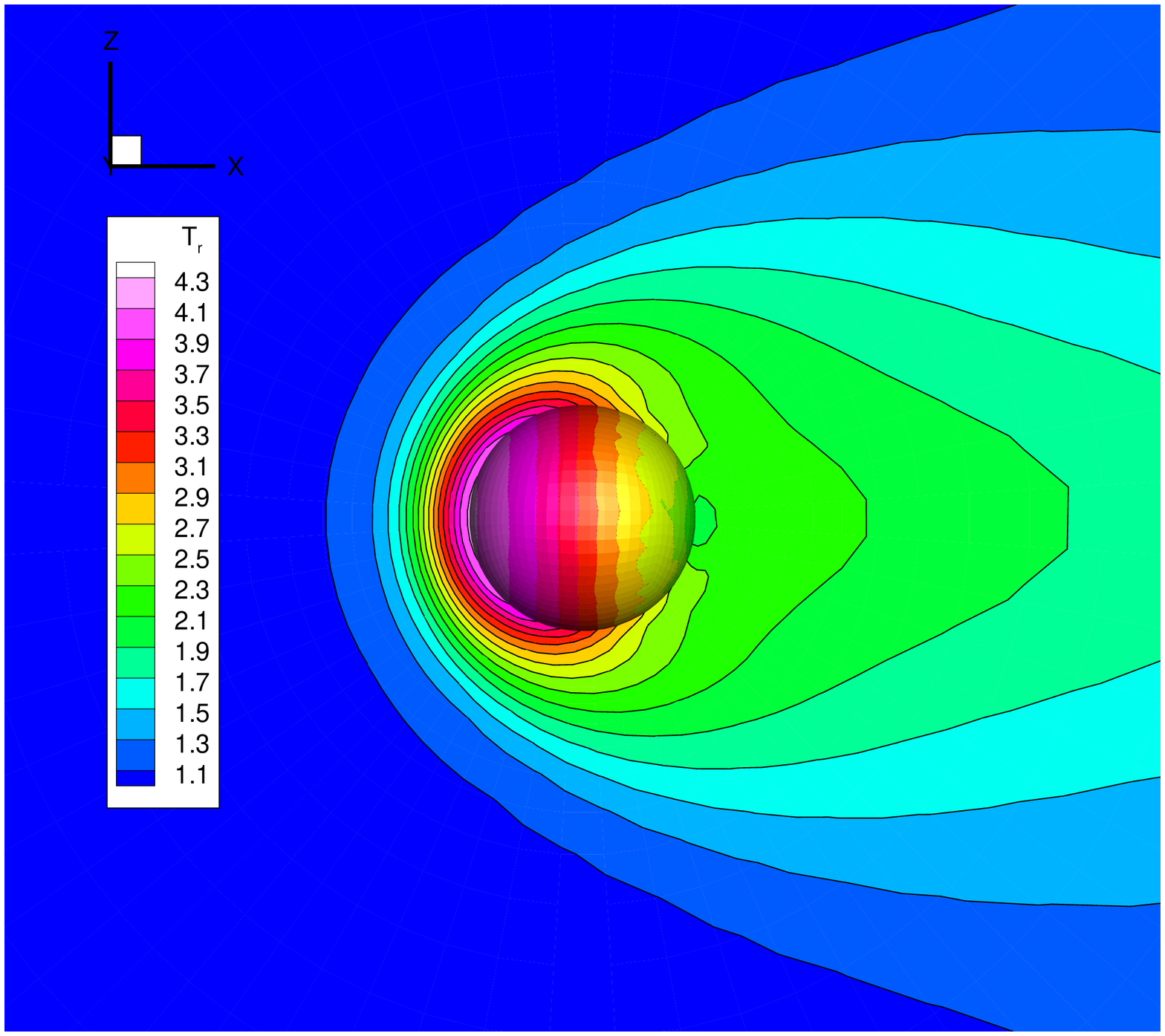}{d}
		\caption{(a)Density and (b)x direction velocity (c)temperature (d) rotational temperature contour for $Kn = 0.121$ and $\mathrm{M} = 4.25$. }
		\label{sphere2}
	\end{figure}

	\begin{figure}
		\centering
		\includegraphics[width=0.48\textwidth]{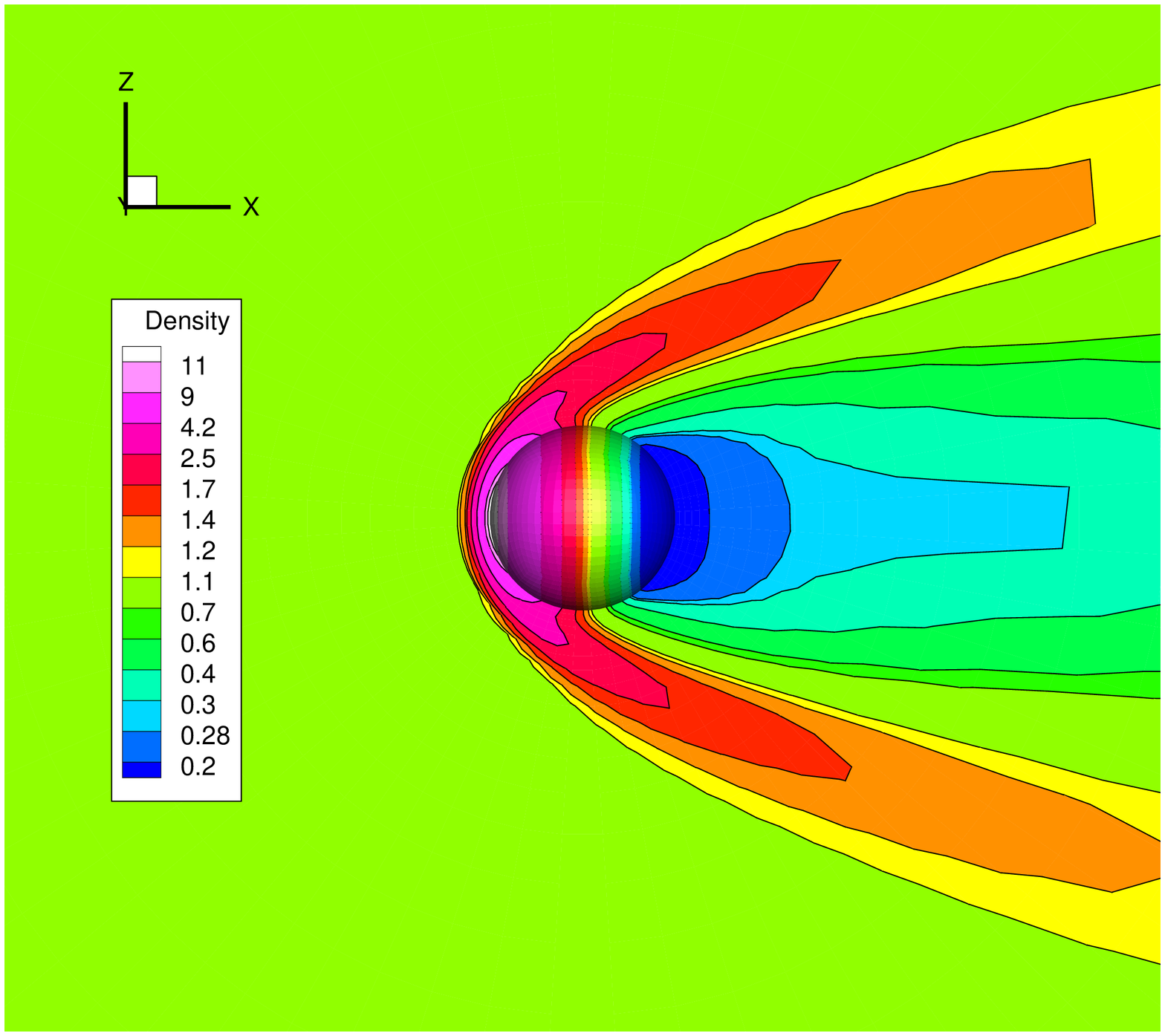}{a}
		\includegraphics[width=0.48\textwidth]{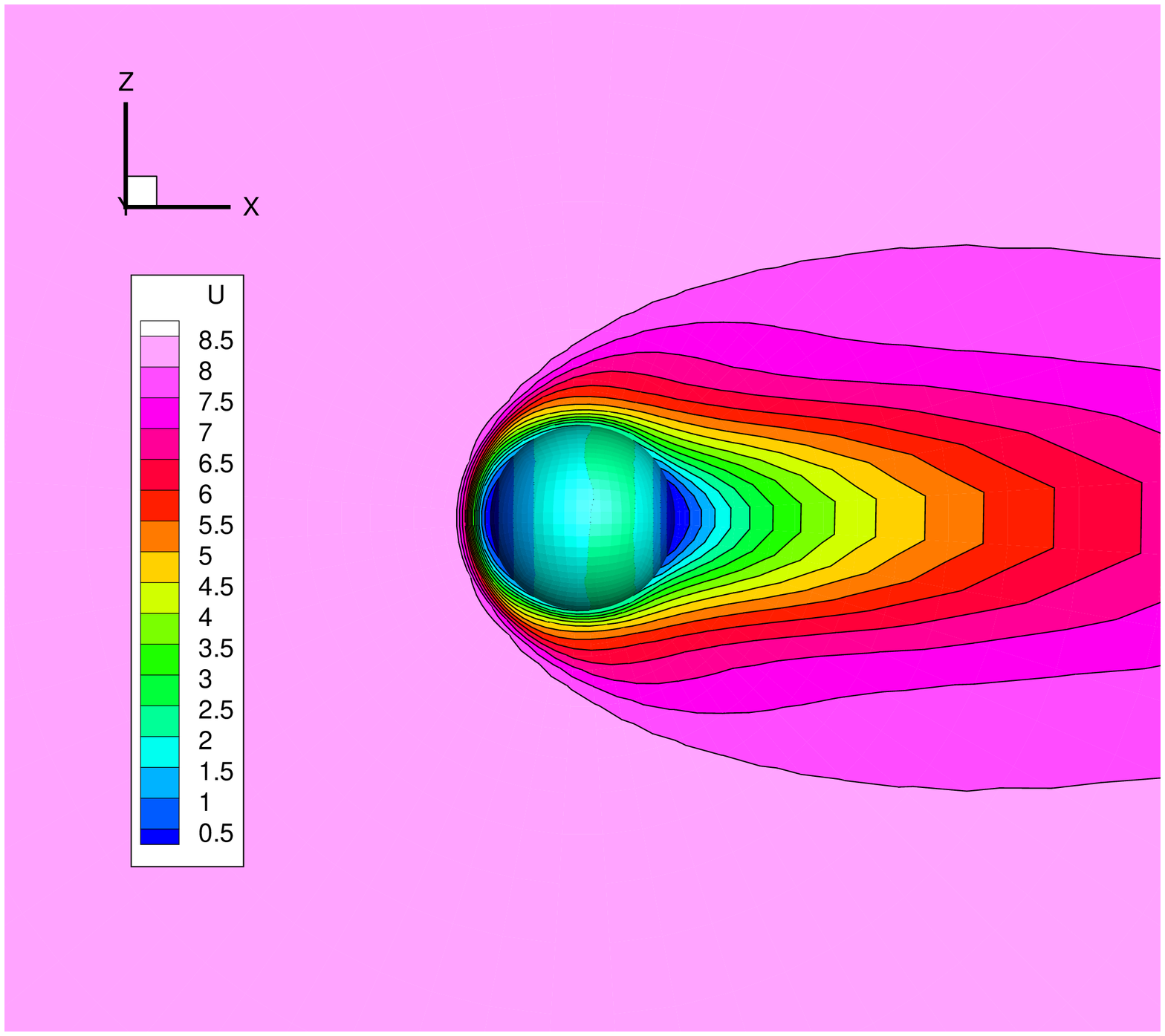}{b}
		\includegraphics[width=0.48\textwidth]{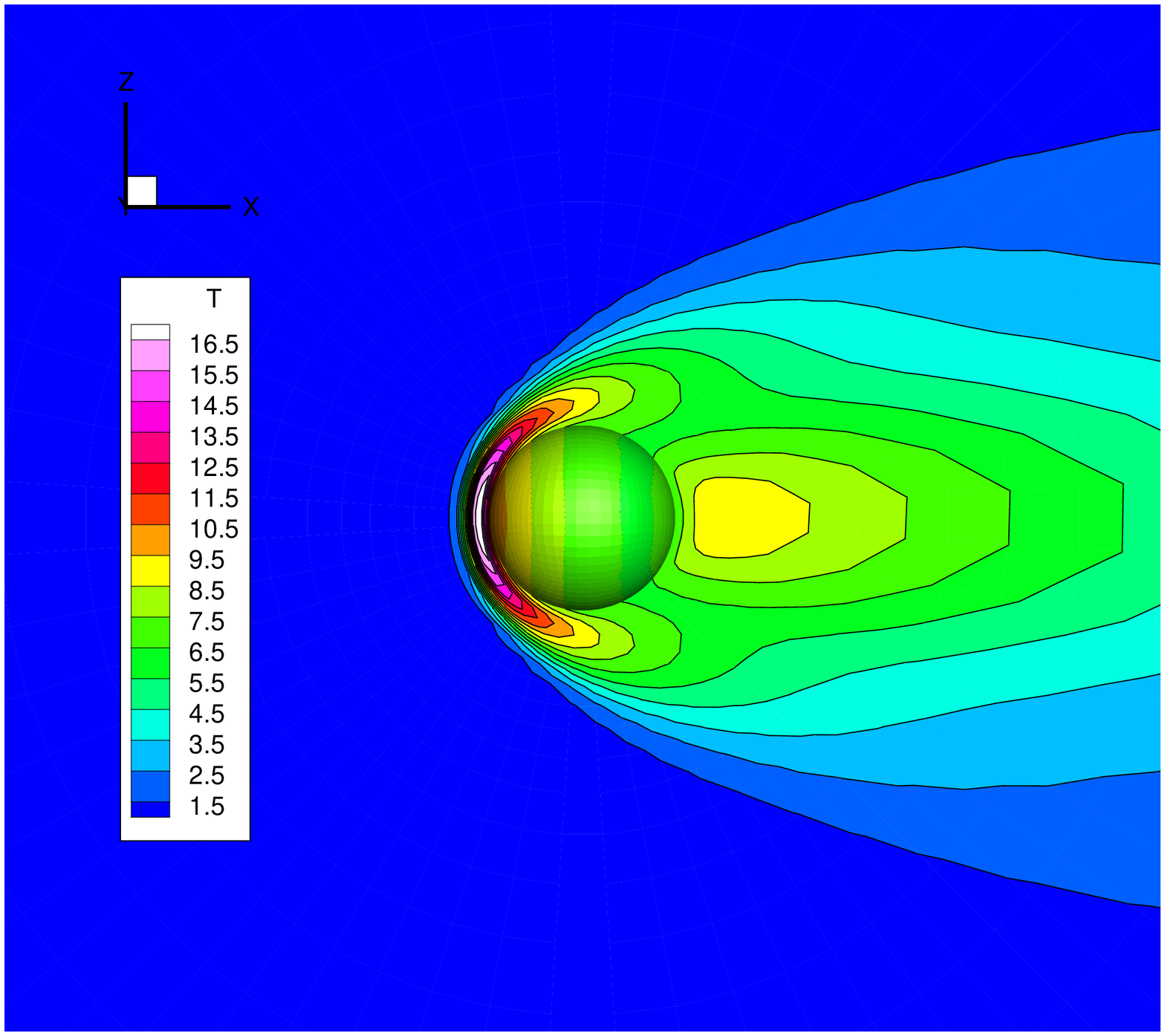}{c}
		\includegraphics[width=0.48\textwidth]{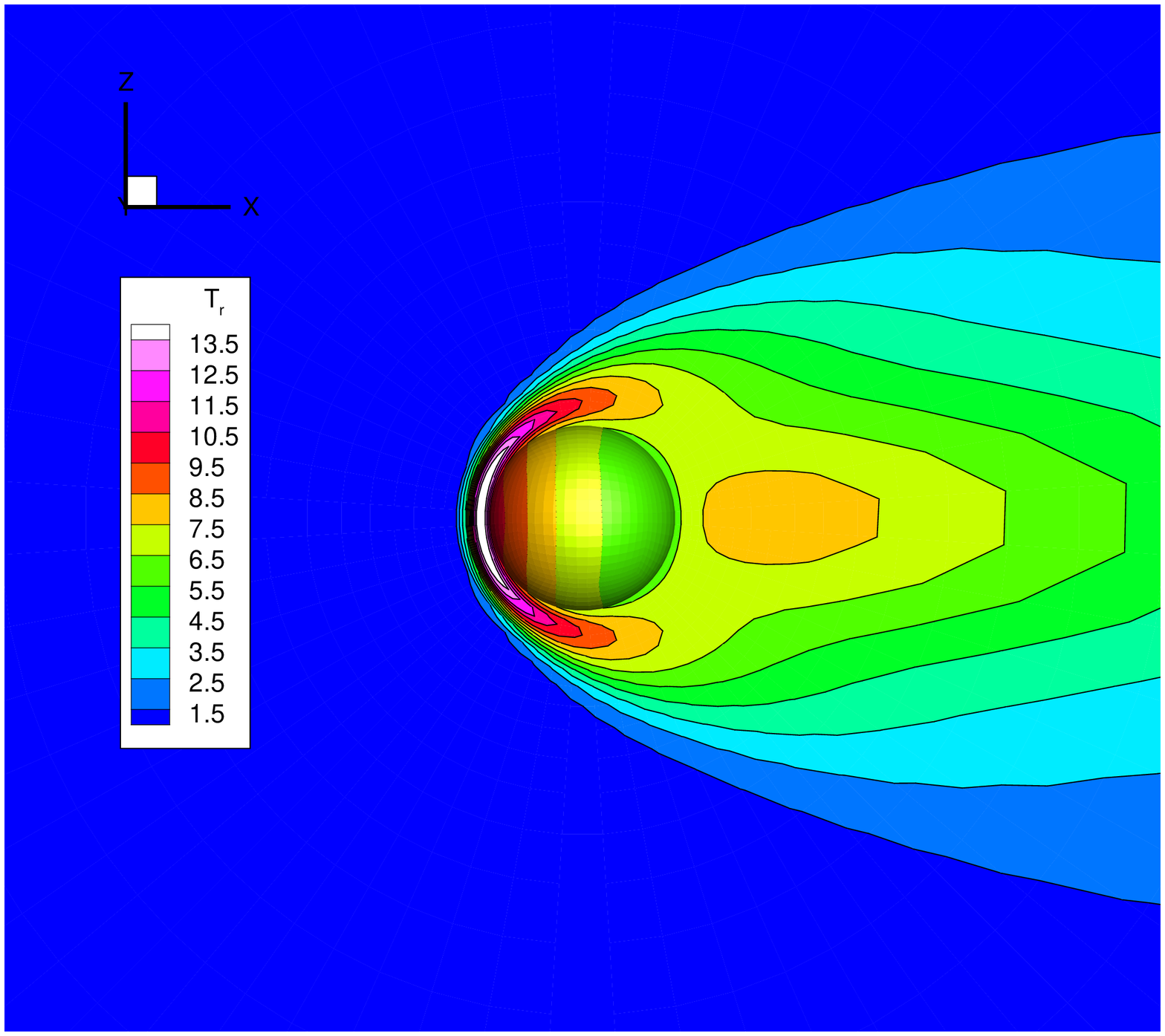}{d}
		\caption{(a)Density and (b)x direction velocity (c)temperature (d) rotational temperature contour for $Kn = 0.01$ and $\mathrm{M} = 10$. }
		\label{sphere3}
	\end{figure}

	\begin{table}[]
		\centering
		\begin{tabular}{|c|c|c|c|c|c|c|}
		\hline
		\multirow{2}{*}{$\mathrm{M}_\infty$} &
		\multirow{2}{*}{$\mathrm{Kn}_\infty$} &
		\multirow{2}{*}{Experiment(Air)} &
		\multicolumn{2}{c|}{UGKWP(Nitrogen)} &
		\multicolumn{2}{c|}{NS solution} \\ \cline{4-7}
		&       &      & Drag coefficient   & Relative error       & Drag coefficient & Relative error \\ \hline
		4.25 & 0.121 & 1.69 & $1.636\pm 0.0005$  & $-3.21\% \pm 0.03\%$ & 2.826            & $67.23\%$      \\ \hline
		4.25 & 0.031 & 1.35 & $1.346 \pm 0.0007$ & $-0.25\% \pm 0.05\%$ & 1.743            & $29.09\%$      \\ \hline
		10   & 0.01  & -    & $1.215 \pm 0.0001$ & -                    & -                & -              \\ \hline
		\end{tabular}
		\caption{Comparison of the drag coefficients}
		\label{compDragError}
	\end{table}

	\begin{figure}
		\centering
		\includegraphics[width=0.48\textwidth]{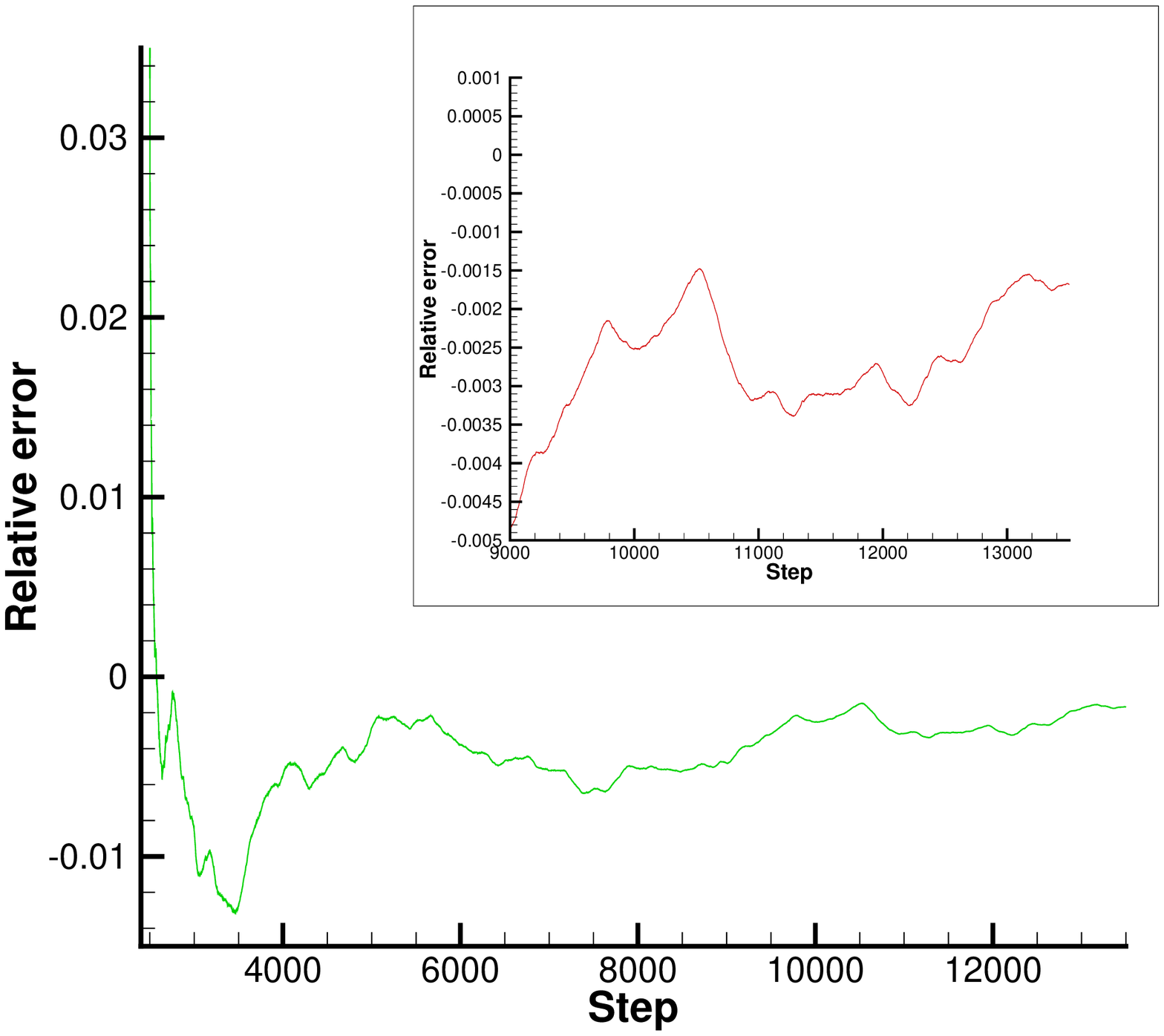}
		\includegraphics[width=0.48\textwidth]{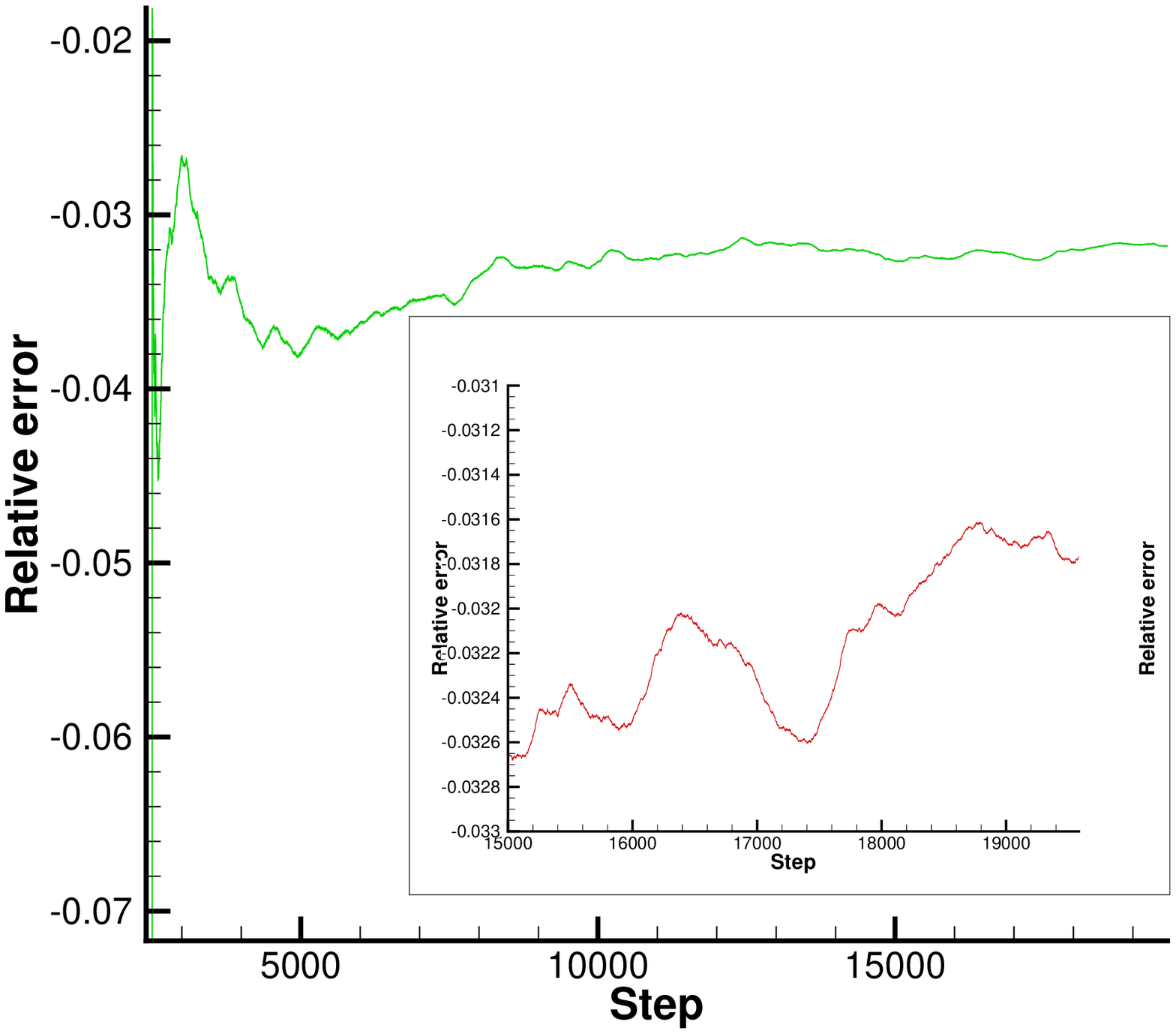}
		\caption{Relative error of the drag coefficient at $Kn = 0.031$ (left) and $Kn = 0.121$ (right).}
		\label{sphereerror}
	\end{figure}

	\section{Conclusion}\label{conclusion}
	In this paper, the unified gas-kinetic wave-particle method has been developed for diatomic gas, where the Rykov model is used for the molecular collision term with the exchange of translational and rotational energy.
Based on the direct modeling, the UGKWP constructs the discrete governing equations according to the cell's Kndsen number and computes gas dynamic solution in all flow regimes with a unified approach.
Instead of DVM approach in the UGKS, the gas distribution function in UGKWP is composed of the contribution from the particle and wave, where analytical solution can be obtained for the wave evolution. At the same time, the weights for distributing particle and wave are related to
$\exp(-\Delta t/\tau) \rho$ and   $(1-\exp(-\Delta t/\tau)) \rho$. As a result,  the UGKWP becomes a particle method in the highly rarefied regime $\Delta t \leq \tau$, and becomes a macroscopic NS solver in the continuum flow regime $\Delta t \gg \tau$. There is a smooth dynamic transition between different regimes in UGKWP.
Therefore, besides asymptotic property to the NS solver, the UGKWP has multiple efficiency preserving property for a multiscale flow problem,
such as hypersonic flow passing through a flying vehicle in near space with  $5$ to $6$ orders of magnitude difference in the local Knudsen number.
The calculation for a 3D problem at high speed and different Knudsen numbers can be conducted by UGKWP with a personal computer.
The UGKWP for diatomic gas has been validated  in many test cases.
Reasonable agreements have been obtained among UGKWP solutions, DSMC results, and experimental measurements.
In the future, the UGKWP will be further developed with the inclusion of vibrational mode and partially ionization \cite{liu2020plasma}.

	\section*{Acknowledgments}
	The current research is supported by National Numerical Windtunnel project and  National Science Foundation of China 11772281, 91852114.

	\bibliography{xxuaybib}
	\bibliographystyle{elsarticle-num}
	\biboptions{numbers,sort&compress}
	
	\clearpage

\end{document}